\documentclass[12pt,reqno]{article}

\usepackage{etex}
\usepackage{enumitem}
\usepackage{setspace}
\usepackage{amsmath}

\usepackage{tabu}
\usepackage{tikz}
\usepackage{tikz-cd}
\usepackage[margin=1in]{geometry}
\usepackage[matrix,arrow,curve,frame]{xy}
\usepackage{hyperref}
\usepackage{amsthm}
\usepackage[capitalise,noabbrev]{cleveref}
\usepackage{textcomp}
\usepackage{amsfonts}
\usepackage{amssymb}
\usepackage{wasysym}
\usepackage{mathrsfs}
\usepackage{mathtools}

\usepackage[utf8]{inputenc}

\definecolor{my-linkcolor}{rgb}{0.75,0,0}
\definecolor{my-citecolor}{rgb}{0.1,0.57,0}
\definecolor{my-urlcolor}{rgb}{0,0,0.75}
\hypersetup{
	colorlinks, 
	linkcolor={my-linkcolor},
	citecolor={my-citecolor}, 
	urlcolor={my-urlcolor}
}

\title{}
\author{Thomas Creutzig\thanks{T.~C. is supported by the Natural Sciences and Engineering Research Council of Canada (RES0020460).},\ and Davide Gaiotto\thanks{The research of D.G. is supported by the Perimeter Institute for Theoretical Physics.  Research at Perimeter Institute is supported by the Government of Canada through Industry Canada and by the Province of Ontario through the Ministry of Economic Development \& Innovation. }}
\date{}

\numberwithin{equation}{section}

\theoremstyle{definition}\newtheorem{rem}{Remark}[section]
\theoremstyle{plain}\newtheorem{prop}[rem]{Proposition}
\newtheorem{thm}[rem]{Theorem}
\theoremstyle{definition}
\theoremstyle{plain}\newtheorem{lem}[rem]{Lemma}
\newtheorem{cor}[rem]{Corollary}
\theoremstyle{definition}\newtheorem{exam}[rem]{Example}
\theoremstyle{definition}
\theoremstyle{definition}
\theoremstyle{definition}\newtheorem{conj}[rem]{Conjecture}
\theoremstyle{definition}    
    
\theoremstyle{definition}

\newcommand{\Com}{\text{Com}}
\newcommand{\OPE}{operator product expansion}
\newcommand{\voa}{vertex operator algebra}
\newcommand{\vosa}{vertex operator subalgebra}
\newcommand{\avoa}{affine vertex operator algebra}
\newcommand{\voas}{vertex operator algebras}
\newcommand{\psl}{\mathfrak{psl}(2|2)}
\newcommand{\sltwo}{\mathfrak{sl}(2)}
\newcommand{\Vir}{\text{Vir}}

\newcommand{\ch}{\text{ch}}

\newcommand{\g}{\mathfrak{g}}
\newcommand{\fA}{\mathfrak{A}}
\newcommand{\fM}{\mathfrak{M}}
\newcommand{\ZZ}{\mathbb{Z}}
\newcommand{\AG}{\mathfrak{A}[G, \Psi]}
\newcommand{\AGn}{\mathfrak{A}^{(n)}[G, \Psi]}

\newcommand{\Agn}{\mathfrak{A}^{(n)}[\g, \Psi]}
\newcommand{\Agni}{\mathfrak{A}^{(n_i)}[\g, \Psi]}
\newcommand{\Agno}{\mathfrak{A}^{(n)}[\g, \infty]}

\newcommand{\AGo}{\mathfrak{A}[G, \infty]}
\newcommand{\ASUNn}{\mathfrak{A}^{(n)}[SU(N), \Psi]}

\newcommand{\BG}{\mathfrak{\widetilde A}[G, \Psi]}

\newcommand{\CM}{\mathcal{M}}
\newcommand{\CF}{\mathcal{F}}
\newcommand{\CL}{\mathcal{L}}

\newcommand{\CB}{\mathcal{B}}
\newcommand{\CJ}{\mathcal{J}}

\allowdisplaybreaks

\begin{document}

\title{Vertex Algebras for S-duality}	
\date{}
\maketitle

\abstract{We define new deformable families of vertex operator algebras $\fA[\g, \Psi, \sigma]$ associated to a large set of S-duality operations in four-dimensional supersymmetric gauge theory. 
They are defined as algebras of protected operators for two-dimensional supersymmetric junctions 
which interpolate between a Dirichlet boundary condition and its S-duality image. 
The $\fA[\g, \Psi, \sigma]$ \voas{} are equipped with two $\g$ affine vertex subalgebras whose levels are related by the S-duality operation. 
They compose accordingly under a natural convolution operation and can be used to define an action of the S-duality operations 
on a certain space of \voas{} equipped with a $\g$ affine vertex subalgebra. We give a self-contained definition of the S-duality action on 
that space of \voas. 
The space of conformal blocks (in the derived sense, i.e. chiral homology) 
for $\fA[\g, \Psi, \sigma]$ is expected to play 
an important role in a broad generalization of the quantum Geometric Langlands program. 
Namely, we expect the S-duality action on \voas{} to extend to an action on the corresponding spaces of conformal blocks. 
This action should coincide with and generalize the usual quantum Geometric Langlands correspondence. 
The strategy we use to define the $\fA[\g, \Psi, \sigma]$ \voas{} is of broader applicability and leads to many new results and conjectures 
about deformable families of \voas.}

\setcounter{tocdepth}{2}
\setcounter{secnumdepth}{4}

\newpage

\tableofcontents

\allowdisplaybreaks

\section{Introduction}
The objective of this paper is to identify and characterize a large class of vertex operator (super)algebras $\fA[\g, \Psi; \CB_L, \CB_R;\CJ]$
which appear in the study of topologically twisted four-dimensional ${\cal N}=4$ gauge theory \cite{Nekrasov:2010ka,Gaiotto:2016wcv,Gaiotto:2017euk,CG}. 
These ``corner VOAs'' encode the algebras of certain protected operators which live at the junctions (``corners'') of topological boundary conditions 
for the four-dimensional gauge theory. This gauge theory setup is motivated by applications to the Geometric Langlands program \cite{Kapustin:2006pk,Gaiotto:2016hvd}. 

Recall that the topologically twisted four-dimensional gauge theory itself depends on a choice of gauge group $G$ (with Lie algebra $\g$) 
and on a topological gauge coupling $\Psi$ which appears as a continuous parameter in the corner VOAs.\footnote{The global form of the gauge group does not affect directly the algebra of local operators but will 
select a specific class of modules for the VOA}

The four-dimensional theory admits a large class of topological boundary conditions. One may consider a two-dimensional junction $\CJ$
between two different topological boundary conditions $\CB_L$ and $\CB_R$. Although the three-dimensional boundary conditions are topological, 
the two-dimensional junctions compatible with the topological twist are typically not topological, but rather holomorphic:  
they support an algebra of holomorphic local operators with non-trivial OPEs,
i.e. a vertex operator (super)algebra which we denote as $\fA[\g, \Psi; \CB_L, \CB_R;\CJ]$.\footnote{This setup has obvious generalizations involving two-dimensional junctions of multiple topological interfaces. }

The four-dimensional gauge theory enjoys an S-duality symmetry acting on the coupling $\Psi$ by fractional linear transformations, 
which acts non-trivially on boundary conditions and junctions. It leads to non-trivial identifications between corner VOAs.  

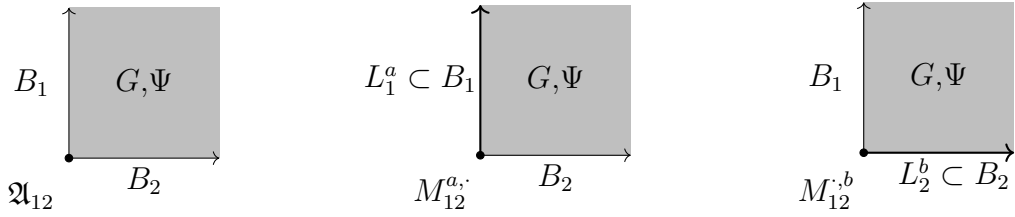
\begin{figure}[p]
\begin{center}
\begin{tikzpicture}
\path [fill=lightgray] (0,0) rectangle (2,2);
\draw[->] (0,0) --(0,2);
\draw[->] (0,0) --(2,0);
\node at (1,1) {$G$,$\Psi$};
\node at (-.5,1) {$B_1$};
\node at (1,-.3) {$B_2$};
\draw[fill] (0,0) circle [radius=0.05];
\node at (-.5,-.5) {$\mathfrak{A}_{12}$};
\end{tikzpicture}
\hspace{1.5cm}
\begin{tikzpicture}
\path [fill=lightgray] (0,0) rectangle (2,2);
\draw[->,thick] (0,0) --(0,2);
\draw[->] (0,0) --(2,0);
\node at (1,1) {$G$,$\Psi$};
\node at (-.8,1) {$L_1^a \subset B_1$};
\node at (1,-.3) {$B_2$};
\draw[fill] (0,0) circle [radius=0.05];
\node at (-.5,-.5) {$M^{a,\cdot}_{12}$};
\end{tikzpicture}
\hspace{1.8cm}
\begin{tikzpicture}
\path [fill=lightgray] (0,0) rectangle (2,2);
\draw[->] (0,0) --(0,2);
\draw[->,thick] (0,0) --(2,0);
\node at (1,1) {$G$,$\Psi$};
\node at (-.5,1) {$B_1$};
\node at (1.2,-.3) {$L_2^b \subset B_2$};
\draw[fill] (0,0) circle [radius=0.05];
\node at (-.5,-.5) {$M^{\cdot,b}_{12}$};
\end{tikzpicture}
\end{center}
\caption{Vertex Operator Algebras at corners. Left: junctions between boundary conditions of GL-twisted SYM typically support 
VOAs, determined by the choice of boundary conditions and by the specific choice of junction. Center: Boundary line defects 
can end at the junction on vertex operators associated to modules for the junction VOA. The modules fuse and braid according to the 
braided tensor category of line defects. Right: There are modules associated to lines in either boundary. The modules associated to lines in one boundary 
braid trivially with modules associated to lines in the other boundary.	They fuse to composite modules $M^{a,b}_{12}$. 
}\label{fig:one}
\end{figure}

The boundary conditions are equipped with braided (super)tensor categories $C_L$ and $C_R$ of topological line defects.
The corner VOA $\fA[\g, \Psi; \CB_L, \CB_R;\CJ]$ is equipped with a class $\fM[\g, \Psi; \CB_L, \CB_R;\CJ]$ of modules 
which braid according to the product $C_L \times \overline{C}_R$ of the braided (super)tensor categories associated to the two sides of the junction (here $\overline{C}_R$ denotes the category opposite to $C_R$, i.e. braiding is reversed). 
See Figure \ref{fig:one}.

The corner VOAs are known or conjectured explicitly for some simple examples of boundary conditions $\CB_L, \CB_R$, 
where the whole setup has a simple Lagrangian description. S-duality acts on these examples 
to provide larger families where the corner VOA is known even though the system does not admit a simple Lagrangian description. 

There is a powerful strategy to build an even larger class of examples:
``resolve'' a complicated junction $\CJ$ between $\CB_L$ and $\CB_R$ 
into sequences of simpler junctions $\CJ_i$ between a sequence of boundary conditions 
$\CB_1 \equiv \CB_L, \CB_2, \cdots, \CB_N \equiv \CB_R$.\footnote{Or even webs of junctions between topological interfaces.} 
Conjecturally, the resulting collection $\prod_i \fA[\g, \Psi; \CB_{i}, \CB_{i+1};\CJ_i]$ of simpler \voas{} is conformally embedded into the VOA 
$\fA[\g, \Psi; \CB_L, \CB_R;\CJ]$ for the original junction. 

The decomposition of the original VOA $\fA[\g, \Psi, \CB_L, \CB_R,\CJ]$ into modules for the simpler VOAs is also conjecturally known:
it is the sum of products of modules associated to topological line defects stretched between simpler junctions:
\begin{equation}\nonumber
\fA[\g, \Psi; \CB_L, \CB_R;\CJ] =  \fM[\g, \Psi; \CB_{L}, \CB_{2};\CJ_1] \boxtimes_{C_2}  \cdots \boxtimes_{C_{N-1}} \fM[\g, \Psi; \CB_{N-1}, \CB_R;\CJ_{N-1}]
\end{equation} 
This is a well-defined extension, precisely because these line defects pair up modules which braid according to dual braided tensor categories. 
See Figure \ref{fig:two}.

\begin{figure}[p]
\begin{center}
\begin{tikzpicture}
\path [fill=lightgray] (0,0) -- (2,0) -- (2,3) -- (-1,3) -- (-1,1);
\draw[->] (0,0) --(2,0);
\draw[->] (0,0) --(-.5,.5);
\draw (0,0) --(-1,1);
\draw[->] (-1,1) --(-1,3);
\node at (.5,1.5) {$G$,$\Psi$};
\node at (-.8,.2) {$B_2$};
\node at (-1.5,2) {$B_1$};
\node at (1,-.3) {$B_3$};
\draw[fill] (0,0) circle [radius=0.05];
\draw[fill] (-1,1) circle [radius=0.05];
\node at (.1,-.8) {$\mathfrak{A}_{12} \times \mathfrak{A}_{23} $};
\end{tikzpicture}
\hspace{1.2cm}
\begin{tikzpicture}
\path [fill=lightgray] (0,0) -- (2,0) -- (2,3) -- (-1,3) -- (-1,1);
\draw[->] (0,0) --(2,0);
\draw[->,thick] (0,0) --(-.5,.5);
\draw[thick] (0,0) --(-1,1);
\draw[->] (-1,1) --(-1,3);
\node at (.5,1.5) {$G$,$\Psi$};
\node at (-1.1,.2) {$L_2^b \subset B_2$};
\node at (-1.5,2) {$B_1$};
\node at (1,-.3) {$B_3$};
\draw[fill] (0,0) circle [radius=0.05];
\draw[fill] (-1,1) circle [radius=0.05];
\node at (.1,-.8) {$\mathfrak{A}_{13} \equiv \bigoplus_b M_{12}^{\cdot,b} \times M_{23}^{b,\cdot} $};
\end{tikzpicture}
\hspace{1cm}
\begin{tikzpicture}
\path [fill=lightgray] (0,0) -- (2,0) -- (2,3) -- (-1,3) -- (-1,1);
\draw[->] (0,0) --(2,0);
\draw[->,thick] (0,0) --(-.5,.5);
\draw[thick] (0,0) --(-1,1);
\draw[->,thick] (-1,1) --(-1,3);
\node at (.5,1.5) {$G$,$\Psi$};
\node at (-1.1,.2) {$L_2^b \subset B_2$};
\node at (-2,2) {$L_1^a \subset B_1$};
\node at (1,-.3) {$B_3$};
\draw[fill] (0,0) circle [radius=0.05];
\draw[fill] (-1,1) circle [radius=0.05];
\node at (.1,-.8) {$M_{13}^{a,\cdot} \equiv \bigoplus_b M_{12}^{a,b} \times M_{23}^{b,\cdot}$};
\end{tikzpicture}
\end{center}
\caption{The composition of junctions.  Left: Two consecutive junctions supporting each a VOA. Middle: the composition of the junctions 
is associated to a larger VOA. Local operators at the composite junction arise from line defect segments stretched between the 
individual junctions. The composite VOA $\mathfrak{A}_{13}\equiv \mathfrak{A}_{12} \boxtimes_{C_2} \mathfrak{A}_{23}$ is a conformal extension of $\mathfrak{A}_{12} \times \mathfrak{A}_{23}$
by products of modules associated to these line segments. Right: Lines on outer boundaries are associated to modules for the 
full composite VOA, built again with the help of extra line defect segments stretched between the 
individual junctions. 
}\label{fig:two}
\end{figure}
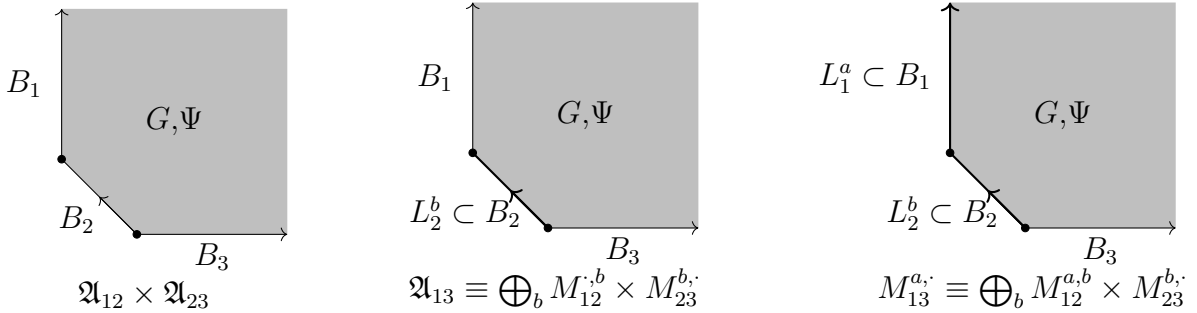

In this paper we will test these conjectures by proposing explicit descriptions of these conformal embeddings
in many important examples. In particular, we describe an explicit candidate for a VOA encoding 
the quantum Geometric Langlands duality for general self-dual gauge groups. 
We describe it very explicitly for the $U(2)$ gauge group. We also discuss 
in some detail the candidate VOA for $U(N)$ and other classical groups. 

We propose generalizations of the quantum Geometric Langlands associated to general S-duality operations. 
We both give VOAs which encode that action and define an action of the S-duality group on 
VOAs which is expected to descend to the generalized quantum Geometric Langlands
correspondences between their conformal blocks. 

The final definition of the VOAs is rather straightforward. It ultimately hinges on two simple observations (physically reasonable, but mathematically yet rather conjectural)
valid for a simply-laced $\g$:
\begin{itemize}
\item The affine $\g$ vertex algebra $G_{\Psi} \equiv V_{\Psi - h^\vee}[\g]$ has a nice class $\fM_{\g, \Psi}$ 
of modules defined as Weyl modules induced from finite-dimensional 
representations of $\g$. These are the standard vertex operators familiar to physicists. These modules braid 
as objects of a braided tensor super-category $C_{\g,q}$ which only depends on the continuous parameter 
$q \equiv e^{\frac{2 \pi i}{\Psi}}$. Furthermore, $\overline{C}_{\g,q} \simeq C_{\g,q^{-1}}$.
\item The W-algebra $W_\Psi[\g]$ defined as the regular Drinfeld-Sokolov reduction of the affine $\g$ vertex algebra 
above has a nice class $\fM^W_{\g, \Psi}$ of modules which braid as $C_{\g,q} \boxtimes C_{\g,q^\vee}$ 
\footnote{This statement is true up to important fermion number shifts, which are the reason our construction produces 
super-algebras rather than algebras, and up to restricting the weights of the representations in a manner 
associated to the action of Langlands duality on the global form of the group}
where $q^\vee$ is computed as before from the Feigin-Frenkel dual level $\Psi^\vee = \Psi^{-1}$. 
\end{itemize}
These facts allow one to built conformal extensions of chains of VOAs of the form 
\begin{equation}
G_{\Psi_0} \times W_\Psi[\g] \times \cdots \times W_\Psi[\g] \times G_{\Psi^{-1}_n}
\end{equation}
The conformal extension is built from products of nice modules and can be described schematically as 
\begin{equation}
\fM_{\g, \Psi_0} \boxtimes_{C_{\g, q_0}} \fM^W_{\g, \Psi_1} \boxtimes_{C_{\g, q_1}} \cdots \boxtimes_{C_{\g, q_{n-1}}} \fM_{\g,\Psi^{-1}_n}
\end{equation}
which makes sense as long as $q_i q^\vee_{i+1} = 1$,
i.e. 
\begin{equation}
\Psi^{-1}_i + \Psi_{i+1} = n_i
\end{equation} 
for a sequence of integers $n_i$. The resulting VOA is best defined when the $n_i$ are positive, which guarantees 
that the space of currents of a given scaling dimension is finite-dimensional.

The simplest example is
\begin{equation}
\fA[\g,\Psi] = \fM_{\g, \Psi+1} \boxtimes  \fM_{\g, \Psi^{-1} + 1}
\end{equation}

The action of S-duality on vertex algebras can also be described in a relatively simple manner. 
It factors through certain maps between 
\begin{itemize}
\item A certain space $A_{\g, \Psi}$ of vertex algebras equipped with a 
$\g$ affine subalgebra at critically shifted level $\Psi$ 
\item A certain space $B_{\g, q}$ of vertex algebras equipped with a family of nice modules 
which braids according to $C_{\g, q}$. 
\end{itemize}

The maps are defined as follows:
\begin{itemize}
\item There is a natural map $A_{\g, \Psi} \to B_{\g, q^{-1}}$ given by the operation of taking a coset by the affine sub-algebra.  
The nice modules for the coset VOA are the ones associated to the Weyl modules of the affine sub-algebra. 
\item There is a natural map $A_{\g, \Psi} \to B_{\g, q^\vee}$ given by the regular Drinfeld-Sokolov reduction. 
The nice modules for the Drinfeld-Sokolov reduction of a VOA descends from spectral flow images of the vacuum module
of the VOA.
\item There is a natural map $B_{\g, q^{-1}} \to A_{\g, \Psi}$ given by the operation $\boxtimes_{C_{\g,q}}$ 
which inverts the coset operation if the initial and final $\Psi$ coincide 
\end{itemize}
These basic maps can be composed to give non-trivial transformations acting on $A_{\g, \Psi}$ to be identified with the action of S-duality 
transformations. It is also useful to add a simpler operation $T: A_{\g, \Psi} \to A_{\g, \Psi+1}$, which consists of taking a tensor product with a $\g$ WZW model at level $1$. 
\footnote{This actually makes the second map in the above list redundant, thanks to  
the coset description of $W_\Psi[\g]$. }

\subsection{Structure of the paper}
In the remainder of this section, we will give more detailed introductions aimed either to readers interested in the 
application of gauge theory to Geometric Langlands or to readers interested only in a mathematical treatment of the Vertex Operator Algebras themselves. 
We will also say a few words about an interesting generalization of our work involving surface defects in the bulk gauge theory. 
In Section \ref{sec:two} we discuss the gauge theory construction of VOAs and the expected S-duality properties. 
In Section \ref{sec:three} we discuss the gauge theory construction of the special VOA expected to 
play the role of kernel for the quantum Geometric Langlands duality for simply laced gauge groups. 
In Section \ref{sec:four} and \ref{sec:five} we compare alternative constructions of the 
VOAs associated to classical simply-laced groups. 
In Section \ref{sec:six} we compare alternative constructions of the 
VOAs associated to non-simply-laced orthogonal and symplectic groups. 
In Section \ref{sec:seven} we discuss the relation between our VOAs and the 
S-duality group. 
In Section \ref{sec:eight} we give a careful mathematical treatment of aspects of the VOAs 
built earlier on in the paper. In Section \ref{sec:nine} we look again at the VOAs associated to 
an $SU(2)$ gauge group. 

\subsection{Geometric Langlands and branes}
The Geometric Langlands and quantum Geometric Langlands correspondences are conjectural equivalences 
between certain categories attached to a Riemann surface $C$ and a 
reductive Lie group $G$. See \cite{Fr} for a review of Geometric Langlands with many references, \cite{Gait1} for a recent outline. See \cite{Gait2} for a recent review of quantum Geometric Langlands.

The physical interpretation of these correspondences is that they encode the equivalence of distinct ``mirror''
mathematical descriptions of the same categories of branes \cite{Kapustin:2006pk}. Up to subtleties which are ameliorated by a gauge theory 
description, the Geometric Langlands correspondence should be associated to ``BAA'' branes in the Hitchin moduli space 
$\CM_H(C,G)$, mirror to ``BBB'' branes in the Hitchin moduli space $\CM_H(C,G^{\vee})$ for the Langlands dual group. 
On the other hand, the quantum Geometric Langlands correspondence should be associated to ``ABA'' branes in the Hitchin moduli space 
$\CM_H(C,G)$, mirror to ``AAB'' branes in the Hitchin moduli space $\CM_H(C,G^{\vee})$ for the Langlands dual group.

All the relevant maps from physical categories of branes to mathematical categories should be associated to the existence 
of certain special families $\CF$ of branes \cite{Gaiotto:2016hvd,Gaiotto:2016wcv}: the mathematical object associated to a given brane $X$ should be given as 
the spaces of morphisms from the branes in $\CF$ to $X$. The Geometric Langlands and quantum Geometric Langlands 
correspondences should relate the mathematical descriptions associated to mirror pairs $\CF$ and $\CF^\vee$ of 
families of branes. 

To be specific, the branes relevant for Geometric Langlands can be expressed conveniently in the complex structure 
where $\CM_H(C,G)$ is a space of Higgs bundles. The family $\CF$ consists of complex Lagrangian branes 
which wrap the locus defined by fixing the bundle and letting the Higgs field vary. The family $\CF^\vee$ consists of complex Lagrangian branes 
which wrap the locus defined by fixing the characteristic polynomial of the Higgs field and letting the bundle vary. These are mirror to skyscraper 
branes in $\CM_H(C,G^{\vee})$. 

The branes relevant for quantum Geometric Langlands can be expressed conveniently in the complex structure 
where $\CM_H(C,G)$ is a space of complex flat connections. The family $\CF$ consists again of complex Lagrangian branes 
which wrap the locus defined by fixing the bundle and letting the holomorphic part of the connection vary. 
The family $\CF^\vee$ consists of the mirrors to the family $\CF$ in $\CM_H(C,G^{\vee})$. We do not know a definition which does 
not involve mirror symmetry, though we expect them to be higher rank versions of the canonical coisotropic brane. 

From this point of view, the spaces of morphisms between branes in $\CF$ and $\CF^\vee$ plays a particularly important role. 
They can be employed as some sort of ``Fourier-Mukai kernel'' for the correspondences. For example, in the case of Geometric Langlands 
these should give a universal sheaf of Hecke eigensheaves. 

In order to acquire information about these morphism spaces, we will switch to a four-dimensional gauge theory description of the system.
The twisted compactification of ${\mathcal N}=4$ SYM on a Riemann surface $C$ gives the non-linear sigma model on the Hitchin moduli space,
up to subtleties concerning the singularities of the latter. Crucially, the families $\CF$ for the standard or quantum Geometric Langlands 
setup all descend from a specific choice boundary condition in the four-dimensional gauge theory, ``Dirichlet'' boundary conditions,
denoted here as $B^{D}_{0,1}$. 

The all-important S-duality of the four-dimensional gauge theory maps to mirror symmetry in 2d. We denote the S image of 
Dirichlet boundary conditions as $B^{D}_{1,0}$, whose 2d image is thus $\CF^\vee$

As boundary conditions in 4d descend to branes in 2d upon compactification on $C$, junctions between boundary conditions in 4d descend to 
morphisms between the corresponding branes in 2d. More precisely, there is a map from conformal blocks of a corner VOA to 
the space of morphisms between the branes corresponding to the two sides of the corner. That means we can learn about 
the spaces of morphisms between branes in $\CF$ and $\CF^\vee$ by looking at conformal blocks for a corner VOA 
between $B^{D}_{0,1}$ and $B^{D}_{1,0}$. 

We thus expect such corner VOAs to have important applications in the 
standard or quantum Geometric Langlands programs. More precisely, the VOAs we construct
for general values of the topological coupling $\Psi$ are relevant for the quantum Geometric Langlands program.
The $\Psi \to \infty$ limits of these VOAs are relevant for the standard Geometric Langlands program.

There are other interesting classes $\CF_\rho$ of BAA branes, labelled by an $\mathfrak{su}(2)$ embedding $\rho$ in $G$. 
For example, when $\rho$ is the regular embedding, $\CF_\rho$ is a single Lagrangian manifold: the canonical 
section of the Hitchin fibration for the GL setup or the oper manifold for the qGL setup. They play an important role in the 
correspondences, especially for regular $\rho$. 

These branes also descend from four-dimensional boundary conditions $B^{\rho}_{0,1}$, known as Nahm pole boundary conditions. 
For regular $\rho$, we denote them simply as $B_{0,1}$. 
They will be associated to interesting corner VOA. In particular, the S image $B_{1,0}$ of $B_{0,1}$ is well understood and there are 
important corner VOAs between $B_{1,0}$ and $B^{\rho}_{0,1}$. We review them in the main body of the paper.

\subsection{Vertex operator algebras}

The physics picture that we present in this work advocates many statements that are rather surprising and thus exciting from the \voa{} point of view. We predict many \voa{} extensions, \voa{} isomorphisms and braided equivalences between full subcategories of different vertex operator superalgebras.

We use section \ref{sec:eight} and \ref{sec:nine} to explain some rigorous results on this. 
Section \ref{sec:examplesVOAextensions} explains in a few examples of rational \voas{} how one indeed can prove braided equivalences between full subcategories of different \voas. The strategy of proof should generalize beyond rationaliy and it is an ambitious future goal to indeed prove a few of our predicted equivalences of vertex tensor subcategories. 
Section \ref{sec:nine} discusses the affine vertex operator superalgebra of the Lie superalgebra $\mathfrak d(2, 1;-\lambda)$ at level one for generic $\lambda$. We denote this deformable family of affine vertex operator superalgebras by the symbol $D(2, 1;-\lambda)_1$. 
 It serves as a first example of a junction \voa{} for the Lie group $SU(2)$ and here we can prove the \voa{} properties predicted by gauge theory. 
Generalizing the $SU(2)$ example to Lie groups of ADE-type is currently under investigation.

\subsubsection{Conjectures on vertex operator algebras}

Here we describe briefly the conjectures which may be interesting for vertex algebra experts, independently of the gauge theory and Geometric Langlands applications. 

The notion of a deformable family of \voas{} has been introduced in \cite{CL1, CL2}. These are \voas{} over a ring of functions, such that the \voas{} are free modules over this ring. Our picture predicts a zoo of deformable families of new \voas. For this let $\g$ be a reductive simply-laced Lie algebra. Pick a positive integer $n$ and let $P^+_n$ be the set of all dominant integral weights of $\g$ 
with the property that $n(\lambda,\lambda)$ in $\mathbb Z$. 
Denote by $M_{\Psi, \lambda}$ the Weyl module of highest-weight $\lambda$ of the Kac-Mody \voa{} $V_{\Psi-h}(\g)$. 

The gauge theory construction implies the conjecture
\begin{conj}\label{conj:voaextension1}
Let $\Psi$ be a complex number and $\Psi'$ satisfy
\[
\frac{1}{\Psi} + \frac{1}{\Psi'} =n 
\]
then  the $V_{\Psi-h}(\g)\otimes V_{\Psi'-h}(\g)$-module
\[
\Agn = \bigoplus_{\lambda \in P^+_n} M_{\Psi, \lambda} \otimes M_{\Psi', \lambda}
 \] 
 can be given the structure of a simple vertex operator superalgebra. 
 \end{conj} 
 In the case of $n=1$ and $\mathfrak g=\sltwo$ this vertex operator algebra is the affine vertex operator superalgebra of $D(2,1;-\lambda)$ at level one for generic $\lambda$ and $\Psi=1-\lambda$ and this is discussed in this paper. In the instance of $n=2$ and $\mathfrak g=\sltwo$ this is the large $N=4$ super conformal algebra at central charge $-6$. The proof is a nice application of deformable families of \voas{} and is work in progress.

 The $\Psi \to \infty$ limit of the above conjecture leads to another useful statement: 
\begin{conj}\label{conj:voainfty1}
The $\g \otimes V_{n^{-1}-h}(\g)$-module
\[
\Agno = \bigoplus_{\lambda \in P^+_n} R_{\lambda} \otimes M_{n^{-1}, \lambda}
 \] 
 where $R_{\lambda}$ is the finite-dimensional representation of $\g$ of weight $\lambda$, can be given the structure of a simple vertex operator superalgebra. 
 \end{conj} 
 In the instance of $\mathfrak g=\sltwo$ and $n=1,2$ one gets the coset vertex algebra $\text{Com}(L_1(\sltwo), L_1(\psl))$  (which is the rectangular $W$-algebra of $\mathfrak{sl}(4)$ at level $-5/2$ by Remark 5.3 of \cite{CKLR}) respectively the small $N=4$ super conformal algebra at central charge $-9$ studied in \cite{AdaN4}. For general $n$ but still $\g=\sltwo$ the even subalgebras of these algebras are constructed in \cite{C} and denoted by $\mathcal Y_n$.

The construction can be further extended with the help of the W-algebra $W_\Psi(\g) \simeq W_{\Psi^{-1}}(\g)$, defined as the 
regular Drinfield-Sokolov reduction of $V_{\Psi-h}(\g)$. Recall that $W_\Psi(\g)$ has a distinguished family of 
modules $M_{\Psi, \lambda, \lambda'} \simeq M_{\Psi^{-1}, \lambda', \lambda}$ labelled by two weights of $\g$. 
The gauge theory construction implies the conjecture
\begin{conj}\label{conj:voaextension1}
Let $\Psi_0, \cdots \Psi_{m+1}$ be a collection of numbers which satisfy
\[
\frac{1}{\Psi_i} + \Psi_{i+1} =n_i
\]
then  the $V_{\Psi_0-h}(\g) \otimes \left( \bigotimes\limits_{i=1}^m W_{\Psi_i}(\g)\right) \otimes V_{\Psi_{m+1}^{-1}-h}(\g)$-module
\[
\Agni = \bigoplus_{\lambda_i \in P^+_{n_i}} M_{\Psi_0, \lambda_0} \otimes \left(\bigotimes_{i=1}^m M_{\Psi_i, \lambda_{i-1},\lambda_i}\right) \otimes M_{\Psi_{m+1}^{-1}, \lambda_m}
 \] 
 can be given the structure of a simple vertex operator superalgebra. 
 \end{conj} 
 
These conjectures  can be extended to situations where $\g$ is not simply laced, with the help of the Feigin-Frenkel duality 
properties of the W-algebra $W_\Psi(\g)$. In this work we however concentrate on the simply-laced case.
 
The gauge theory construction will also often involve extra lattice \voa{} ingredients.  
This will lead to variants of \ref{conj:voaextension1} of the form 
 \begin{conj}\label{conj:voaextension2} The following can be given the structure of a simple vertex operator superalgebra as well
\[
\BG = \bigoplus_{\nu \in P/Q} \bigoplus_{\substack{\lambda \in P^+\\ \lambda \in \nu +Q}} M_{\Psi, \lambda} \otimes M_{\Psi', \lambda} \otimes V_\nu'
\]
where $V_\nu'$ is a certain lattice \voa{} module that ensures half-integer conformal weight grading.
\end{conj} 

In particular, we will produce \voa s which have a particularly nice behaviour under quantum Drinfield-Sokolov reduction. The $n=1$ case of \ref{conj:voaextension2}
has the property that a regular Drinfield-Sokolov reduction of the second $\g$ Kac-Moody subalgebra produces $V_{\Psi-h-1}(\g)\otimes L_1(\g)$,
where $L_1(\g)$ is the level $1$ WZW \voa{}. Furthermore, the Drinfield-Sokolov reduction of spectrally flown modules gives Weyl modules 
of the form $M_{\Psi-h-1,\lambda} \otimes L_1(\g)$. These properties will be instrumental for the Geometric Langlands interpretation of the \voa{} and are explained for $\mathfrak g=\sltwo$ in section \ref{sec:nine}, while explaining this for all ADE-type cases is again work in progress. 

For $\lambda=0$ and  $\mathfrak g=\sltwo$ we verify the correctness of this conjecture and for arbitrary $\lambda$ we can perform a computation of the Euler-Poincar\'e character which strongly supports our conjecture as well.  
Going beyond trivial $\lambda$ amounts to leaving the category $\mathcal O$ and studying quantum Hamiltonian reduction beyond category $\mathcal O$. As one can also perform spectral flow on the ghosts, one can reformulate this problem to studying the semi-infinite cohomology inside category $\mathcal O$ of 
\[
d_\lambda = d_{st} +\chi\circ \sigma^\lambda
\]
with $d_{st}$ the standard differential of the semi-infinite cohomology of $\widehat\g$, $\sigma^\lambda$ the spectral flow automorphism corresponding to the weight $\lambda$ and $\chi$ the usual character. The conjecture that arises is
\begin{conj}
The $\lambda$-twisted reduction of a Weyl module is 
\[
H_{d_\lambda}(M_{\Psi, \nu}) \cong L^{\Psi}(\gamma_{\nu-\Psi(\lambda +\rho^\vee)}) \cong L^{\Psi^{-1}}(\gamma_{\lambda-\Psi(\nu +\rho^\vee)}) .
\]
\end{conj}

Finally, we have an observation whose gauge theory meaning is not fully clear to us. 
Let $H_{DS,m}$ be the Drinfield-Sokolov functor corresponding to the $\mathfrak{sl}_2$ embedding in $\mathfrak{sl}_{n+m}$ of type $m, 1, 1, \dots, 1$. 
Then $H_{DS,m}\left(V_k(\mathfrak{sl}_{n+m})\right)$ contains $V_{k+m-1}(\mathfrak{gl}_n)$ as a sub vertex operator algebra. Let $\mathcal H$ be the Heisenberg \voa{} commuting with the $V_{k+m-1}(\mathfrak{sl}_n)$.
\begin{conj}
The classical limit $\lim\limits_{\Psi\rightarrow \infty}\ASUNn$ is a large center times a \voa{} extension of $$\Com\left(\mathcal H, H_{DS,n(N-1)-1}\left(V_{\frac{n+1}{n}-(n+1)(N-1)}(\mathfrak{sl}_{(n+1)(N-1)}) \right)\right).$$
\end{conj}
The Conjecture is true for  $\mathfrak g=\sltwo$ \cite{C}. We remark that the \voa{} constructions of \cite{C} aimed to find chiral algebras whose character coincides with the Schur index of certain Argyres-Douglas theories considered in \cite{Buican:2015ina, Cordova:2015nma}. 

These large $\Psi$-limits are related to logarithmic CFT results. The best-known logarithmic \voas{} are the triplet algebras \cite{AM}, which are conjecturally the DS-reductions of above coset for $N=2$ \cite{C}. In general the picture is that the large $\Psi$-limit of our deformable families of conjectural \voas{} is a \voa{} with compact Lie group acting as outer automorphisms but such \voas{} are exactly constructed by Feigin and Tipunin as new interesting logarithmic \voas{} as extensions of regular $W$-algebras of simply-laced Lie algebras \cite{FT}, see also \cite{Lentner}. 
One sees that the regular DS-reduction on our potential large $\Psi$-limits of $\Agno$ coincides on the level of characters with the $W$-algebras of \cite{FT}. 

\subsection{A brief discussion of surface defects}
In this section we will very quickly touch on a construction which we will ignore in the rest of the paper, but we expect to be rather interesting. 

The four-dimensional gauge theory is equipped with families of half-BPS surface defects. A particularly important family consists of 
Gukov-Witten defects \cite{Gukov:2006jk,Gukov:2008sn}, labelled by a Levi subgroup $G_\rho$ of $G$ and a collection of continuous parameters living in the part of the Cartan torus which 
commutes with $G_\rho$. In the topological theory, only a doubly-periodic complex combination $\alpha$ of the parameters survives. 
These defects are covariant under S-duality, with $\alpha$ transforming roughly as $\alpha \to \frac{\alpha}{c \Psi + d}$. 

These defects can be inserted in our construction along a plane perpendicular to the two-dimensional junctions, sharing one direction with the 
boundary conditions or interfaces. The intersection between the surface defect and the boundary supports a collection of line defects which are a 
module category for the braided tensor category of the boundary.  

If we denote these module categories as $C[G_\rho, \alpha]$, then the intersection of the surface defect and the junction between two boundaries 
will support a class of modules $\fM[\g, \Psi; g_\rho, \alpha;\CB_L, \CB_R;\CJ]$  
which braid with the original  $\fM[\g, \Psi;\CB_L, \CB_R;\CJ]$ modules according 
to the product $C_L[G_\rho, \alpha] \times \overline{C}_R[G_\rho, \alpha]$ of the module categories associated to the two sides of the junction. 

As we concatenate junctions into more complex junctions, keeping $G_\rho$ and $\alpha$ fixed, these modules will combine accordingly.

\section{A quick review of the gauge theory setup} \label{sec:two}
\subsection{The basic boundary conditions and interfaces}

The GL-twisted ${\cal N}=4$ gauge theory depends on a choice of gauge group $G$ and a ``topological coupling'' $\Psi$, 
a complex number which is a combination of the usual gauge coupling and the choice of supercharge defining the topological twist \cite{Kapustin:2006pk}. 
In particular, the four-dimensional gauge theory can be taken to be weakly coupled for all values of $\Psi$. The most interesting questions 
one can ask, though, involve lower dimensional defects in the theory, which will be typically strongly coupled for generic values of $\Psi$. 

At this point the gauge group $G$ is a general reductive Lie group. 

It is useful to list the set of boundary conditions and interfaces for which a weakly coupled description (say at $\Psi \to \infty$) is known or conjectured \cite{Gaiotto:2017euk,Gaiotto:2008sa}. These will provide us with 
our most useful building blocks. 
\begin{itemize}
\item We already mentioned Dirichlet boundary conditions $B^D_{0,1}$. They are defined uniformly for any gauge group $G$. 
We can use the notation $B^D_{0,1}[G]$ if we need to specify the gauge group. 
\item We also mentioned the Nahm pole variants $B^\rho_{0,1}$, labelled by an $\mathfrak{su}(2)$ embedding $\rho$ in $G$.
We can use the notation $B^\rho_{0,1}[G]$ if we need to specify the gauge group. For regular Nahm pole, 
we use the notation $B_{0,1}$ or $B_{0,1}[G]$.
\item Dirichlet boundary conditions can be generalized to Dirichlet interfaces $B^D_{0,1}[G;H]$
between a $G$ and and $H$ gauge theory with $H \subset G$. The gauge group 
is reduced from $G$ to $H$ as one crosses the interface. The topological couplings $\Psi_H$ and $\Psi_G$ coincide 
up to a rescaling and shift we will describe momentarily. 
\item Dirichlet interfaces have a Nahm pole variant $B^\rho_{0,1}[G;H]$ labelled by an $\mathfrak{su}(2)$ embedding $\rho$ in $G$ which commutes 
with the image of $H$. This modifications change the relation between the topological couplings for the $G$ and $H$ gauge theories. 
We will give the precise relation momentarily. 
\item Because the $H$ gauge symmetry survives at the interface, we can also add extra matter fields at the interface, in the form of hypermultiplets 
transforming in some quaternionic representation $R$ of $H$. This gives interfaces $B^\rho_{0,1}[G;H;R]$. 
This modifications also changes the relation between the topological couplings for the $G$ and $H$ gauge theories. 
We will give the precise relation momentarily. 
\item We can specialize the interfaces above to $H = G$, with no Nahm pole but with hypermultiplets 
transforming in some quaternionic representation $R$ of $G$. We can denote that special case as $B_{0,1}[G;G;R]$ or simply $B_{0,1}[R]$.
The topological couplings at the two sides of the interface differ by an amount proportional to the second Casimir of $R$. 
\item The second class of prototypical boundary conditions is given by Neumann boundary conditions $B_{1,0}$. They are also defined uniformly for any gauge group $G$.
We can use the notation $B_{1,0}[G]$ if we need to specify the gauge group.
\item Neumann boundary conditions can be modified in a very simple way by adding $n$ units of Chern-Simons coupling at the boundary, giving boundary conditions 
we can denote as $B_{1,n}$ or $B_{1,n}[G]$. 
\item A more interesting generalization would involve additional matter fields at the boundary, 
in the form of hypermultiplets transforming into a quaternionic representation $R$ of $G$. This is compatible with generic $\Psi$ 
if and only if $G$ can be extended to a supergroup $\hat G$ by adding fermionic generators transforming in $R$ \cite{Gaiotto:2008sd}.
This leads to boundary conditions or interfaces (by a reflection trick) $B_{1,0}^{\hat G}$ (or $B_{1,0}^{\hat G}[G]$ if we need to specify the gauge group). 
Typical examples are interfaces between $U(N)$ and $U(M)$ 
associated to $U(N|M)$ and interfaces between $SO(N)$ and $Sp(2M)$ associated to $OSp(N|2M)$. 
\item A boundary Chern-Simons coupling can be added as before, leading to $B_{1,n}^{\hat G}[G]$ (or $B_{1,0}^{\hat G}[G]$ if we need to specify the gauge group).  
\end{itemize}
In principle, one may also consider Neumann boundary conditions which are enriched by a more complicated three-dimensional 
gauge theory with eight supercharges. It is relatively rare, though, for such a boundary condition to be 
compatible with the topological twist at general $\Psi$. The only examples we know of arise as the composition of simpler
interfaces of type $B_{1,0}^{\hat G}$ \cite{Gaiotto:2008sd}. We will discuss them after we learn how to compose interfaces and junctions. 

\subsection{The weakly coupled corner VOAs}
In a similar manner, one may look for weakly coupled junctions for which we can compute the corner VOA directly. 
The most general statement we may give at this time involves the VOA at the intersection between a $B^\rho_{0,1}[G;H;R]$
interface and $B_{1,n_G}^{\hat G}$, $B_{-1,-n_H}^{\hat H}$ boundary conditions or interfaces. Based on examples, we expect such a junction $\CJ_{\mathrm{cl}}$ to exist 
if $\hat H$ is a subgroup of $\hat G$ and $R$ can be extended to a representation $\hat R$ of $\hat H$ by the addition of some extra fermionic generators.
The level shifts $n_G$ and $n_H$ must also match in a manner we will discuss momentarily. See Figures \ref{fig:three} and  \ref{fig:four}.

The corner VOA is expected to be the coset \cite{Gaiotto:2017euk}
\begin{equation}
\fA[\g,\Psi; B_{1,n_G}^{\hat G},B^\rho_{0,1}[G;H;R],B_{-1,-n_H}^{\hat H};\CJ_{\mathrm{cl}}] \equiv \frac{\mathrm{DS}_\rho \hat G_{\Psi_G - n_G} \times \mathrm{Sb}^{\hat R}}{\hat H_{\Psi_H - n_H} }
\end{equation}
Here we use the following notations:
\begin{itemize}
\item $\hat G_{\Psi_G - n_G}$ is the Kac-Moody VOA with supergroup $\hat G$ and critically shifted level $\Psi - n_G$. \footnote{The non-critically-shifted level would be 
$\Psi - n_G-h^\vee_G$.}
\item $\mathrm{DS}_\rho$ denotes the operation of quantum Drinfeld-Sokolov reduction associated to the $\mathfrak{su}(2)$ embedding $\rho$ in $G \subset \hat G$.
\item $\mathrm{Sb}^{\hat R}$ denotes a set of symplectic bosons transforming in $\hat R$. That really means the combination of 
symplectic bosons transforming in $R$ and fermions transforming in $\hat R/R$. The resulting VOA has a $\hat H$ current subalgebra. 
\item We take a coset by  the Kac-Moody VOA with supergroup $\hat H$. The VOA is the diagonal combination of the 
$\hat H \subset \hat G$ subalgebra in $\hat G_{\Psi - n_G}$ and the $\hat H$ current algebra in $\mathrm{Sb}^{\hat R}$. 
\item The critically shifted level $\Psi_H - n_H$ is determined according to the above embedding. 
\end{itemize}

The coset can be also usefully expressed (or defined) as a $\hat H$-BRST quotient  of 
\begin{equation}
\mathrm{DS}_\rho \hat G_{\Psi_G - n_G} \times \mathrm{Sb}^{\hat R} \times \hat H_{n_H -\Psi_H} 
\end{equation}

Notice that the construction makes manifest the existence of two classes of mutually local modules for the corner VOA,
induced by Weyl modules associated respectively to finite-dimensional representations of $\hat G$ and $\hat H$. 
These modules braid according to well-known braided tensor categories $C_{\hat G}[q_G \equiv e^{\frac{2 \pi i}{\Psi_G - n_G}}]$
and $C_{\hat H}[q_H \equiv e^{-\frac{2 \pi i}{\Psi_H - n_H}}]$
of quantum group representations. \footnote{Up to important fermion number shifts due to the DS reduction.}
The two classes of modules are mutually local with each other, with trivial mutual braiding. 
Both properties have a clear gauge theory meaning: these modules live at the endpoints of  line defects along the $B_{1,n_G}^{R_G}$
or $B_{1,n_H}^{R_H}$ boundaries.

The braided tensor category of line defects along the $B^\rho_{0,1}[G;H;R]$ is not generally known. 
Whatever it is, it will be associated to a third class of modules for the corner VOA, 
local with the other two classes of modules, braiding according to that tensor category. 

Extra degrees of freedom can be also added at the junction in the form of some chiral conformal field theory with $\hat H$ $WZW$ currents. 
The theory will then appear in the numerator of the coset and $n_H$ will be shifted appropriately. 

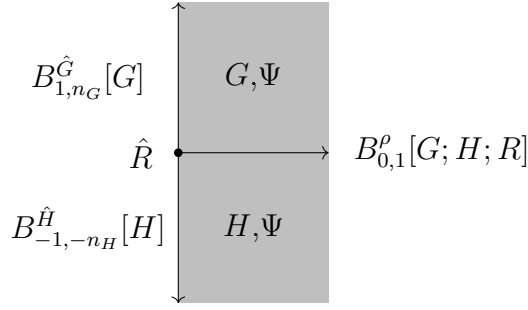
\begin{figure}[p]
\begin{center}
\begin{tikzpicture}
\path [fill=lightgray] (0,-2) rectangle (2,2);
\draw[->] (0,0) --(0,2);
\draw[->] (0,0) --(2,0);
\draw[->] (0,0) --(0,-2);
\node at (1,1) {$G$,$\Psi$};
\node at (1,-1) {$H$,$\Psi$};
\node at (-1.2,1) {$B^{\hat G}_{1,n_G}[G]$};
\node at (3.5,0) {$B^\rho_{0,1}[G;H;R]$};
\node at (-1.2,-1) {$B_{-1,-n_H}^{\hat H}[H]$};
\draw[fill] (0,0) circle [radius=0.05];
\node at (-.5,0) {$\hat R$};
\end{tikzpicture}
\end{center}
\caption{A very general weakly coupled junction with a known corner VOA, as described in the text. 
}\label{fig:three}
\end{figure}

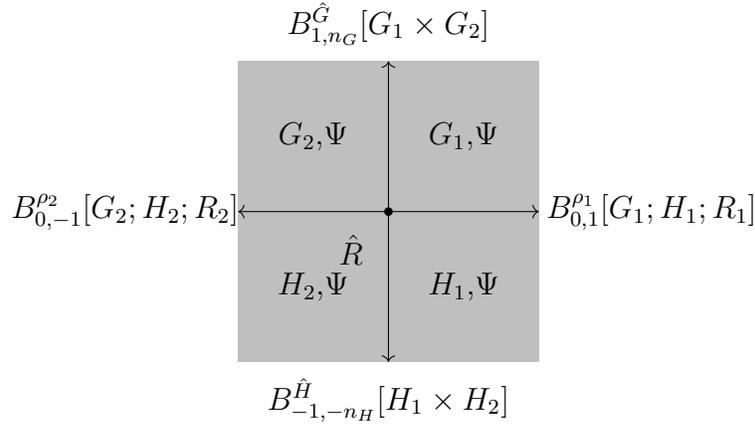
\begin{figure}[p]
\begin{center}
\begin{tikzpicture}
\path [fill=lightgray] (-2,-2) rectangle (2,2);
\draw[->] (0,0) --(0,2);
\draw[->] (0,0) --(2,0);
\draw[->] (0,0) --(0,-2);
\draw[->] (0,0) --(-2,0);
\node at (1,1) {$G_1$,$\Psi$};
\node at (1,-1) {$H_1$,$\Psi$};
\node at (-1,1) {$G_2$,$\Psi$};
\node at (-1,-1) {$H_2$,$\Psi$};
\node at (0,2.5) {$B^{\hat G}_{1,n_G}[G_1 \times G_2]$};
\node at (3.5,0) {$B^{\rho_1}_{0,1}[G_1;H_1;R_1]$};
\node at (-3.5,0) {$B^{\rho_2}_{0,-1}[G_2;H_2;R_2]$};
\node at (0,-2.5) {$B_{-1,-n_H}^{\hat H}[H_1 \times H_2]$};
\draw[fill] (0,0) circle [radius=0.05];
\node at (-.5,-.5) {$\hat R$};
\end{tikzpicture}
\end{center}
\caption{A reflection trick allows one to discuss junctions between four gauge groups, re-interpreting 
a boundary condition for $G_1 \times G_2$ as an interface between $G_1$ and $G_2$. 
}\label{fig:four}
\end{figure}
We can now consider a few basic examples. 

\begin{figure}[p]
\begin{center}
\begin{tikzpicture}
\path [fill=lightgray] (0,0) rectangle (2,2);
\draw[->] (0,0) --(0,2);
\draw[->] (0,0) --(2,0);
\node at (1,1) {$G$,$\Psi$};
\node at (-.5,1) {$B_{1,0}$};
\node at (1,-.3) {$B^D_{0,1}$};
\draw[fill] (0,0) circle [radius=0.05];
\end{tikzpicture}
\hspace{1cm}
\begin{tikzpicture}
\path [fill=lightgray] (0,0) rectangle (2,2);
\draw[->] (0,0) --(0,2);
\draw[->] (0,0) --(2,0);
\node at (1,1) {$G$,$\Psi$};
\node at (-.5,1) {$B_{1,0}$};
\node at (1,-.3) {$B_{0,1}$};
\draw[fill] (0,0) circle [radius=0.05];
\end{tikzpicture}
\hspace{1cm}
\begin{tikzpicture}
\path [fill=lightgray] (0,-2) rectangle (2,2);
\draw[->] (0,0) --(0,2);
\draw[->] (0,0) --(0,-2);
\node at (1,0) {$G$,$\Psi$};
\node at (-1.2,1) {$B_{1,0}[G]$};
\node at (-1.2,-1) {$B_{-1,1}[G]$};
\node at (-.7,0) {$L_1[G]$};
\draw[fill] (0,0) circle [radius=0.05];
\end{tikzpicture}
\end{center}
\caption{Some simple examples of corner configurations discussed in the text. Left: the semiclassical junction between $B_{1,0}$ and $B^D_{0,1}$, which supports a $G_\Psi$ Kac-Moody. 
Middle: the semiclassical junction between $B_{1,0}$ and the regular Nahm pole boundary condition $B_{0,1}\equiv B^{\rho_{\mathrm{reg}}}_{0,1}$, which supports a $W_G[\Psi]$ 
W-algebra. Right: the semiclassical junction between $B_{1,0}$ and $B_{-1,1}$ which includes auxiliary corner degrees of freedom given by a level $1$ 
WZW model $L_1[G]$. It supports a $W_G[1+\Psi^{-1}]$ W-algebra.
}\label{fig:five}
\end{figure}
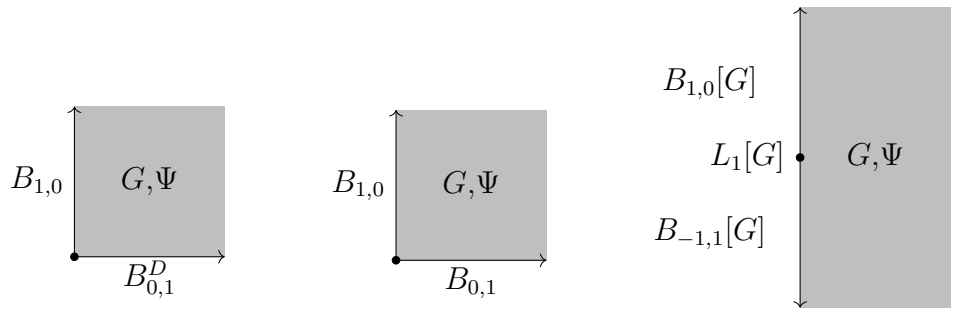

\subsubsection{Kac-Moody algebra at the corner}
The cleanest example is the corner VOA at the semiclassical junction between $B_{1,0}$ and $B^D_{0,1}$:
the Kac-Moody VOA associated to the group $G$, with (critically shifted) level $\Psi$:
\begin{equation}
\fA[\g,\Psi;B_{1,0},B^D_{0,1};\CJ_{\mathrm{cl}}] \equiv G_\Psi
\end{equation}
See Figure \ref{fig:five}.

The modules 
\begin{equation} 
M_{\lambda} \in \fM_{\g,\Psi} \equiv \fA[\g,\Psi;B_{1,0},B^D_{0,1};\CJ_{\mathrm{cl}}] 
\end{equation} associated to lines along the $B_{1,0}$ boundary are Weyl modules induced from finite-dimensional 
representations of $G$. They braid according to a well-known tensor category $C_{\g,q}$ with $q = e^{2 \pi i \Psi^{-1}}$.

The category of lines along the $B^D_{0,1}$ boundary is more mysterious, and so are the associated modules. 
Physically, the lines should be defined as ``boundary 't Hooft lines''. Mathematically, the notion of 
't Hooft operator at some point $p$ should translate to a Hecke modification at $p$ of a principal $G$ bundle.

The simplest possible example is to take $G = U(1)$, so that the corner VOA is 
a $U(1)$ Kac-Moody algebra. It is useful to bosonize the current as $J_\Psi \sim - i \partial \phi$.

 ``Electric'' modules $M_{n}$ associated to lines in $B_{1,0}$ are generated by vertex operators of the form $e^{i \frac{n}{\Psi} \phi}$,
which have integral charge $n$ under $J_\Psi$. These modules braid with phases controlled by $q = e^{\frac{2 \pi i}{\Psi}}$.

On the other hand, ``magnetic'' modules ${}^L M_{m}$ associated to lines in $B^D_{0,1} \simeq B_{0,1}$ are generated by vertex operators 
of the form $e^{i m \phi}$, which have charge $m \Psi$ under $J_\Psi$. These modules braid with phases controlled by  ${}^L q = e^{2 \pi i \Psi}$

The two sets of modules are mutually local. They fuse with a unique fusion channel to composite modules 
$M_{\Psi,n,m}$ generated by vertex operators of the form $e^{i ( m + \frac{n}{\Psi}) \phi}$. 

If we use $J_\Psi$ to couple the system to a $U(1)$ bundle $\CL$, electric vertex operators 
behave as sections of $\CL^{\otimes n}$. Magnetic vertex operators placed at some point $p$ effectively 
modify the line bundle as $\CL \to \CL(m p)$.

The relation between boundary 't Hooft lines in $G$ gauge theory and vertex operators implementing Hecke modifications of 
$G$ bundles is expected to hold in broader generality, but we do not understand it in detail. 

\subsubsection{$W_{G}$ algebra at the corner}
The corner VOA at the semiclassical junction between $B_{1,0}$ and $B_{0,1}$ is the $W_{\Psi}[\g]$ algebra, the regular qDS reduction of
the Kac-Moody VOA associated to the Lie group $G$ with (critically shifted) level $\Psi$:
\begin{equation}
\fA[\g,\Psi;B_{1,0},B_{0,1};\CJ_{\mathrm{cl}}] \equiv W_{\Psi}[\g]
\end{equation}
See Figure \ref{fig:five}.

Recall that the $W_{\Psi}[\g]$ enjoys Feigin-Frenkel duality, which exchanges $G$ with its Langlands dual ${}^L G$ and 
inverts $\Psi$, up to an rescaling dependent of conventions for the level of the dual Kac-Moody algebras. 
As we will review further in a later section, this duality is a manifestation of the S-duality covariance of the 
semiclassical junction between $B_{1,0}$ and $B_{0,1}$, which we will discuss in more detail in section \ref{sec:three}.

The modules $M^W_{\lambda}$ associated to lines along the $B_{1,0}$ boundary are the qDS reduction of Weyl modules induced from finite-dimensional 
representations of $G$. The conformal dimension of the highest weight vectors in the Weyl modules is shifted by their charge under the 
Cartan generator in $\rho$, which is half-integral and additive. The braiding properties of the modules are expected to be unchanged by the 
qDS reduction, as required by the gauge theory setup. 

The modules ${}^L M^W_{{}^L \lambda}$ associated to lines along the $B_{0,1}$ boundary, by S-duality, should be the image of Weyl modules 
associated to finite-dimensional representations of ${}^L G$. Within the original definition as a qDS reduction of the $G$ Kac-Moody algebra,
they arise from the reduction of certain spectral flowed images of the vacuum module. 

These two sets of modules are known to be mutually local, as long as $\lambda$ and ${}^L \lambda$
are chosen from the weight lattice of GL dual global forms of $G$ and ${}^L G$. 

The two types of modules fuse with a unique fusion channel to modules (simple for general $\Psi$)
$M^W_{\lambda, {}^L \lambda} \in \fM^W_{\g,\Psi}$ which are characterized as the quotient of Verma modules of $W_{\Psi}[\g]$ by a maximal 
set of null vectors. The typical mathematical notation for such modules is $L^{\Psi}\left(\gamma_{\lambda - \Psi {}^L ( \lambda+\rho^\vee)}\right)$,
indicating that they are simple quotients of Verma modules labelled by special values of the Toda momentum. 

\subsubsection{$W_{G,\rho}$ algebra at the corner}
The corner VOA at the semiclassical junction between $B_{1,0}$ and $B^\rho_{0,1}$ is the qDS reduction of
the Kac-Moody VOA associated to the group $G$, with (critically shifted) level $\Psi$, according to the $\mathfrak{su}(2)$ embedding $\rho$:
\begin{equation}
\fA[\g,\Psi;B_{1,0},B_{0,1}^\rho;\CJ_{\mathrm{cl}}] \equiv \mathrm{DS}_\rho G_\Psi
\end{equation}

The modules $M^{W_\rho}_{\lambda}$ associated to lines along the $B_{1,0}$ boundary are the qDS reduction of Weyl modules induced from finite-dimensional 
representations of $G$. The braiding properties of the modules are expected to be unchanged by the 
qDS reduction, as required by the gauge theory setup. 

Again, for general $\rho$ the category of lines along the $B^D_{0,1}$ boundary is more mysterious, and so are the associated modules.

\subsubsection{An alternative realization of $W_{G}$ algebra at the corner}
Consider here a simply-laced, semi-simple, simply-connected $G$. The $W_{G}[\Psi]$ algebra is known \cite{GKO} to possess an alternative coset definition, 
\begin{equation}
W_{G} \simeq \frac{G_{\kappa} \times L_1[G]}{G_{\kappa+1}}
\end{equation}
where we denote as $L_1[G]$ the $G$ WZW model at level $1$. 

The level $\kappa$ is related to $\Psi$ as 
\begin{equation}
\kappa = \frac{1}{\Psi-1}
\end{equation}
Notice that the correct BRST definition of the coset involves the combination 
\begin{equation}
G_{\frac{1}{\Psi-1}} \times L_1[G] \times G_{\frac{\Psi}{1-\Psi}}
\end{equation}
which is invariant under $\Psi \to \Psi^{-1}$. 

This can be engineered at a semiclassical junction between $B_{1,0}$ and $B_{-1,1}$, involving extra corner degrees 
of freedom in the form of $L_1[G]$. See Figure \ref{fig:five}. The existence of this third dual description is again associated to an S-duality relation between the relevant
semiclassical junctions, which we will discuss in more detail momentarily. \footnote{Notice that if ${}^L G = G$ then $L_1[G]$ has only the vacuum module
and is a good chiral CFT. For more general $G$ it is a relative theory and one may worry why is it OK to use it as extra junction 
degrees of freedom. This can be explained by a careful analysis of how the global form of the group changes 
under S-duality, leading to subtle discrete anomalies which are cancelled by the coupling to $L_1[G]$. We will not do so here}

The sets of modules associated to lines on $B_{1,0}$ and $B_{-1,1}$ are given by the BRST/coset reduction of Weyl modules for 
$G_{\frac{1}{\Psi-1}}$ or $G_{\frac{\Psi}{1-\Psi}}$ respectively. Notice that 
\begin{equation}
e^{\frac{2 \pi i}{\kappa}} = e^{2 \pi i \Psi} \qquad \qquad e^{-\frac{2 \pi i}{\kappa+1}} = e^{\frac{2 \pi i}{\Psi}}
\end{equation}
so that the braiding properties of the modules are compatible with the duality. 

The above coset description can be extended to more general reductive groups. In order for the maps of modules to 
work well, it is useful to give a proper definition of the WZW model $L_1[G]$, so that the Abelian factors are associated to lattice VOAs rather 
than just Abelian Kac-Moody currents. For example, $L_1[U(N)]$ is defined naturally as the VOA of $N$ complex fermions. 

It is entertaining and instructive to see this description at work for $G=U(1)$.
Here $L_1[U(1)]$ is simply the VOA of a complex fermion, with generators $\chi$ and $\psi$, OPE
\begin{equation}
\chi(z) \psi(w) \simeq \frac{1}{z-w}
\end{equation}
and $U(1)$ current $J_1 = \chi \psi$ giving charges $1$ and $-1$ to $\chi$ and $\psi$. 

It is clear that only the charge $0$ sector of the free fermion VOA gives a contribution to the vacuum module of the coset. This is the same as the 
$U(1)$ current sub-algebra VOA. The coset of the product of two $U(1)$ VOAs by their diagonal combination is obviously another $U(1)$ current VOA. 
If the diagonal current is 
\begin{equation}
J_{\kappa +1} = J_\kappa + \chi \psi
\end{equation}
then the coset current can be taken to be the 
\begin{equation}
J_{\Psi} = \kappa^{-1} J_\kappa - \chi \psi
\end{equation}

Charge $m$ electric modules of $U(1)_{\kappa}$ can be dressed with charge $-m$ modules of the free fermions. 
The resulting coset module is generated by a vertex operator of charge $m \Psi$ under $J_\Psi$, i.e. a magnetic module 
of the coset VOA. 

On the other hand, the coefficient of charge $-n$ modules for $J_{\kappa +1}$ inside charge $-n$ modules of the free fermions VOA 
are coset modules generated by a vertex operator of charge $n$ under $J_\Psi$, i.e. an electric module 
of the coset VOA. 

This matches our general expectation. \footnote{A very careful reader may wonder about the appearance of fermionic degrees of freedom 
at junctions in a bosonic theory. Some questions may also be raised about subtle interplay of fermionic and bosonic 
notions of mutual locality of modules in the coset. Such a reader is invited to explore related subtleties 
about the electric-magnetic duality group of Abelian gauge theories \cite{Metlitski:2015yqa}, such as the 
fact that the $ST$ transformation we use here maps standard gauge connections to Spin$_{\mathbb{C}}$ 
gauge connections. 
}
\subsection{The action of S-duality}
From now on we take $G$ simply laced and possibly include some Abelian factors to insure that $G$ equals its GL dual group, 
i.e. ${}^L G = G$. The prototypical example would be $G = U(N)$. This relieves us from the need to 
follow how the global form of the group changes under S-duality. 

There are two useful symmetries of the four-dimensional gauge theory. 
The first is simply the reflection of a direction in space-time combined with 
\begin{equation}
R: \Psi \to -\Psi
\end{equation} 
The second is the action of S-duality 
\begin{equation}
\Psi \to \frac{a \Psi + b}{c \Psi + d} \qquad \qquad \begin{pmatrix} a & b \cr c & d\end{pmatrix} \in \mathrm{PSL}(2,\mathbb{Z})
\end{equation} 
Two particularly useful generators of the S-duality group (combined with a reflection when useful) are
\begin{equation}
S: \Psi \to \Psi^{-1}  \qquad \qquad T: \Psi \to \Psi + 1
\end{equation} 

The boundary conditions and interfaces we have introduced until now belong to infinite families of boundary conditions which we denote by notations such as 
$B_{p,q}$, $B^D_{p,q}$, $B^{\rho}_{p,q}$, etc. The two numbers $p$ and $q$ are co-prime integers. In general, for every family of boundary conditions 
$B^{\cdots}_{p,q}$ we have that $B^{\cdots}_{-p,-q}$ is the same as $B^{\cdots}_{p,q}$ with opposite 
orientation, while the above duality generators act as 
\begin{equation}
R: B^{\cdots}_{p,q} \to B^{\cdots}_{-p,q} \qquad \qquad S: B^{\cdots}_{p,q} \to B^{\cdots}_{q,p}\qquad \qquad T: B^{\cdots}_{p,q} \to B^{\cdots}_{p,q+p}
\end{equation}
In general, $\Psi$ transforms in the same way as $q/p$, i.e. we have 
\begin{equation}
(p,q) \to (d p +c q, a q + b p) 
\end{equation} 
Notice that this duality action is compatible with the fact that all the boundary conditions of the Dirichlet family $B^{\cdots}_{0,1}$ are invariant under the
$T$ transformation: the $T$ transformation adds boundary Chern-Simons couplings which are trivialized by the Dirichlet boundary conditions. 

Earlier on, we proposed to denote regular Nahm pole boundary conditions as $B_{0,1}$ and Neumann as $B_{1,0}$. We did so because these boundary conditions 
are conjecturally mapped into each other by the $S$ duality operation. This is a crucial relationship, essentially the only non-trivial
duality relation between weakly-coupled boundary conditions which holds universally for all gauge groups. 
The duality relation also implies that Neuman-like boundary conditions should be invariant under $STS$ and that $B_{1,1}$ and $B_{1,-1}$ 
should be mapped to themselves under $S$. These are all deep, non-perturbative statements which are crucial 
for this work. See Figure \ref{fig:dualityone}.

The duality images of Dirichlet boundary conditions $B^D_{0,1}$ are instead all expected to be strongly coupled. See Figure \ref{fig:dualitytwo}.

Given some known semi-classical junctions, we get infinite families of strongly coupled junctions which should have 
the same corner VOA, up to the re-definition of $\Psi$. In some cases, semiclassical junctions may be mapped to other semiclassical 
junctions. This is always the case under $R$ and $T$, but rarely under other symmetry transformations. 

The canonical example, which we have already encountered, is the junction between $B_{1,0}$ and $B_{0,1}$,
which is expected to be mapped to the same type of junction under $S$. We have also described the semiclassical $ST^{-1}$ image of the 
above junction: the $T^{-1}$ operation gives a junction between $B_{1,-1}$ and $B_{0,1}$ at $\Psi -1$, which is then mapped to a 
junction between $B_{1,0}$ and $B_{-1,1}$ at coupling $\kappa = (\Psi-1)^{-1}$. See Figure \ref{fig:dualitythree}.

Notice that we always list the boundary conditions for a junction from the left boundary to the right boundary, oriented outwards from the junction. Both $R$ and $S$ exchange 
the left and right boundaries because of the space-time reflection. 

For general groups these are the only known dualities between junctions. For classical groups there is a more general story described in detail in \cite{Gaiotto:2017euk}.

Starting from the junctions between $B_{1,0}$ and $B^\rho_{0,1}$ we get a canonical choice of junction between any $B_{p,q}$ and $B^\rho_{p',q'}$ 
with $p q' - p' q = 1$, such that the corner VOA is $W_{G,\rho}$. We will see now how to leverage that single piece of knowledge to get many other junctions
with interesting duality properties. See Figure \ref{fig:dualityfour}.

\begin{figure}[p]
\begin{center}
\begin{tikzpicture}
\path [fill=lightgray] (0,0) rectangle (2,2);
\draw[->] (0,0) --(0,2);
\node at (1,1) {$G$,$\Psi$};
\node at (-.5,1) {$B_{1,0}$};
\end{tikzpicture}
\hspace{1cm}
\begin{tikzpicture}
\path [fill=lightgray] (0,0) -- (-2,2) -- (0,2) -- (2,0);
\draw[->] (0,0) --(-2,2);
\node at (0,1) {$G$,$\Psi$};
\node at (-1.3,.7) {$B_{1,-1}$};
\end{tikzpicture}
\hspace{1cm}
\begin{tikzpicture}
\path [fill=lightgray] (0,0) rectangle (2,1.5);
\draw[->] (0,0) --(2,0);
\node at (1,.7) {$G$,$\Psi$};
\node at (1,-.3) {$B_{0,1}$};
\end{tikzpicture}
\end{center}
\caption{Graphical conventions for Neumann boundary conditions and their S-dual images. 
Left: Neumann boundary conditions $B_{1,0}$. 
Middle: modified Neumann boundary conditions $B_{1,1}$. 
Right: Nahm pole boundary conditions $B_{0,1}$.
We use the convention that boundary conditions of type $(p,q)$ are drawn with slope $p/q$.
}\label{fig:dualityone}
\end{figure}
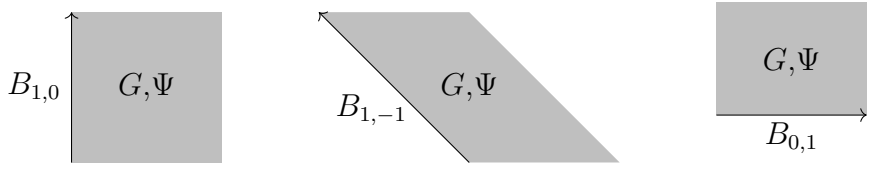

\begin{figure}[p]
\begin{center}
\begin{tikzpicture}
\path [fill=lightgray] (0,0) rectangle (2,1.5);
\draw[->,ultra thick] (0,0) --(2,0);
\node at (1,.7) {$G$,$\Psi$};
\node at (1,-.3) {$B^D_{0,1}$};
\end{tikzpicture}
\hspace{1cm}
\begin{tikzpicture}
\path [fill=lightgray] (0,0) rectangle (2,2);
\draw[->,ultra thick] (0,0) --(0,2);
\node at (1,1) {$G$,$\Psi$};
\node at (-.5,1) {$B^D_{1,0}$};
\end{tikzpicture}
\hspace{1cm}
\begin{tikzpicture}
\path [fill=lightgray] (0,0) -- (-2,2) -- (0,2) -- (2,0);
\draw[->,ultra thick] (0,0) --(-2,2);
\node at (0,1) {$G$,$\Psi$};
\node at (-1.3,.7) {$B^D_{1,-1}$};
\end{tikzpicture}
\end{center}
\caption{Graphical conventions for Dirichlet boundary conditions and their S-dual images.  
Left: Dirichlet boundary conditions $B^D_{0,1}$. 
Middle: The S-dual of Dirichlet boundary conditions $B^D_{1,0}$. It can be defined by coupling the gauge theory to a strongly-coupled three-dimensional SCFT $T[G]$. 
Right: The $(T^{-1}S)$-dual of Dirichlet boundary conditions $B^D_{1,-1}$ can be defined by coupling the gauge theory to $T[G]$ together with a boundary CS coupling. 
We use the convention that boundary conditions of type $(p,q)$ are drawn with slope $p/q$.
}\label{fig:dualitytwo}
\end{figure}
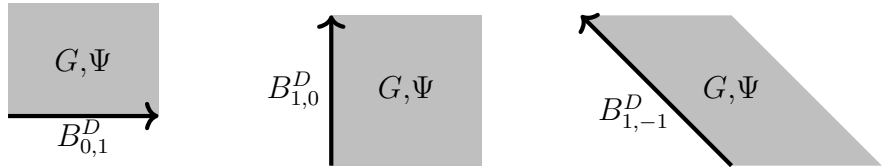

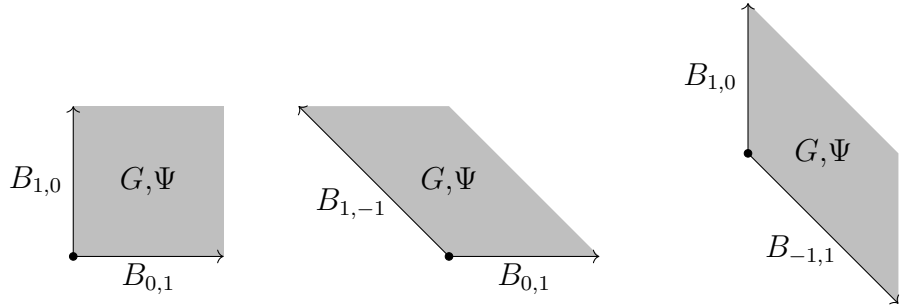
\begin{figure}[p]
\begin{center}
\begin{tikzpicture}
\path [fill=lightgray] (0,0) rectangle (2,2);
\draw[->] (0,0) --(0,2);
\draw[->] (0,0) --(2,0);
\node at (1,1) {$G$,$\Psi$};
\node at (-.5,1) {$B_{1,0}$};
\node at (1,-.3) {$B_{0,1}$};
\draw[fill] (0,0) circle [radius=0.05];
\end{tikzpicture}
\hspace{.7cm}
\begin{tikzpicture}
\path [fill=lightgray] (0,0) -- (-2,2) -- (0,2) -- (2,0);
\draw[->] (0,0) --(-2,2);
\draw[->] (0,0) --(2,0);
\node at (0,1) {$G$,$\Psi$};
\node at (-1.3,.7) {$B_{1,-1}$};
\node at (1,-.3) {$B_{0,1}$};
\draw[fill] (0,0) circle [radius=0.05];
\end{tikzpicture}
\hspace{.7cm}
\begin{tikzpicture}
\path [fill=lightgray] (0,0) -- (2,-2) -- (2,0) -- (0,2);
\draw[->] (0,0) --(2,-2);
\draw[->] (0,0) --(0,2);
\node at (1,0) {$G$,$\Psi$};
\node at (.7,-1.3) {$B_{-1,1}$};
\node at (-.5,1) {$B_{1,0}$};
\draw[fill] (0,0) circle [radius=0.05];
\end{tikzpicture}
\end{center}
\caption{Canonical junctions between basic boundary conditions and their S-dual images.
Left: The canonical junction between  $B_{1,0}$ and $B_{0,1}$ which supports a $W_G[\Psi]$ VOA defined as the qDS reduction of a $G$ Kac-Moody at level $k+h = \Psi$. 
It is conjecturally invariant under the $S$ transformation $\Psi \to \Psi^{-1}$. Indeed, $W_G[\Psi] = W_G[\Psi^{-1}]$.
Middle: The canonical junction between  $B_{1,-1}$ and $B_{0,1}$ which supports a $W_G[\Psi+1]$ VOA. 
Right: The canonical junction between  $B_{1,0}$ and $B_{-1,1}$ supports a $W_{G}[\Psi^{-1}+1]$ VOA. The two latter junctions are conjecturally related by $S$. }
\label{fig:dualitythree}
\end{figure}

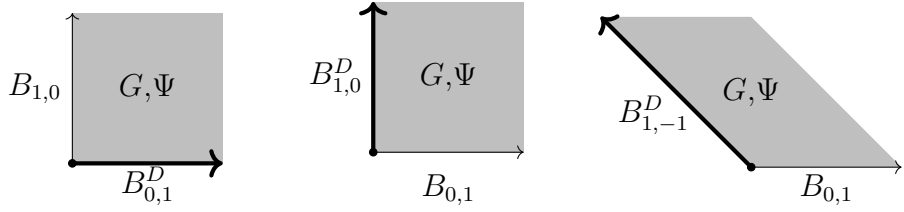
\begin{figure}[p]
\begin{center}
\begin{tikzpicture}
\path [fill=lightgray] (0,0) rectangle (2,2);
\draw[->] (0,0) --(0,2);
\draw[->,ultra thick] (0,0) --(2,0);
\node at (1,1) {$G$,$\Psi$};
\node at (-.5,1) {$B_{1,0}$};
\node at (1,-.3) {$B^D_{0,1}$};
\draw[fill] (0,0) circle [radius=0.05];
\end{tikzpicture}
\hspace{.7cm}
\begin{tikzpicture}
\path [fill=lightgray] (0,0) rectangle (2,2);
\draw[->] (0,0) --(2,0);
\draw[->,ultra thick] (0,0) --(0,2);
\node at (1,1) {$G$,$\Psi$};
\node at (1,-.5) {$B_{0,1}$};
\node at (-.5,1) {$B^D_{1,0}$};
\draw[fill] (0,0) circle [radius=0.05];
\end{tikzpicture}
\hspace{.7cm}
\begin{tikzpicture}
\path [fill=lightgray] (0,0) -- (-2,2) -- (0,2) -- (2,0);
\draw[->, ultra thick] (0,0) --(-2,2);
\draw[->] (0,0) --(2,0);
\node at (0,1) {$G$,$\Psi$};
\node at (-1.3,.7) {$B^D_{1,-1}$};
\node at (1,-.3) {$B_{0,1}$};
\draw[fill] (0,0) circle [radius=0.05];
\end{tikzpicture}
\end{center}
\caption{Canonical junctions between Neumann and Dirichlet boundary conditions and their S-dual images.
Left: The canonical junction between  $B_{1,0}$ and $B^D_{0,1}$ supports a $G$ Kac-Moody at level $k+h = \Psi$. 
Middle: The $S$ operation gives a junction between  $B^D_{1,0}$ and $B_{0,1}$ supports a $G^\vee$ Kac-Moody at level $k+h = \Psi^{-1}$. 
Right: The $ST^{-1}$ operation gives a junction between  $B_{1,0}$ and $B_{-1,1}$ supports a $G^\vee$ Kac-Moody at level $k+h = (\Psi+1)^{-1}$. 
Similar junctions exist for all pairs $B_{a,b}$ and $B^D_{c,d}$ with $ad-bc = \pm 1$. 
}\label{fig:dualityfour}
\end{figure}

\subsection{Concatenating junctions}
The following general conjecture was introduced in \cite{Gaiotto:2017euk}. See Figure \ref{fig:two}.

Consider junctions $\CJ_{12}$ and $\CJ_{23}$ between boundary conditions $\CB_1$ and $\CB_2$ and $\CB_2$ and $\CB_3$. 
We can attempt to define a new junction $\CJ_{12} \circ \CJ_{23}$ by concatenating $\CJ_{12}$ and $\CJ_{23}$ and taking a scaling limit. 

Denote as $\fA_{12}$ and $\fA_{23}$ the corner VOAs associated to $\CJ_{12}$ and $\CJ_{23}$. 
They should be equipped with two families of modules $M^R_{12,\lambda}\in \fM_{12}^R$ and $M^L_{23,\lambda}\in \fM^L_{23}$ 
where $\lambda$ runs over the elements of some braided tensor category $C_2$ associated to $B_2$. 
The two families of modules braid according to $\bar C_2$ and $C_2$ respectively. 

Then the VOA associated to $J_{12} \circ J_{23}$ is conjecturally the extension of $\fA_{12} \times \fA_{23}$ 
\begin{equation}
\fM^R_{12} \boxtimes_{C_2} \fM^L_{23} \equiv \bigoplus_\lambda M^R_{12,\lambda} \otimes M^L_{23,\lambda}
\end{equation}

Notice that $\fA_{12}$ and $\fA_{23}$ should be equipped with more general families of modules $M_{12,\lambda_1,\lambda_2}$ and $M_{23,\lambda_2,\lambda_3}$ 
where $\lambda_i$ runs over the elements of some braided tensor categories $C_i$ associated to $\CB_i$. These modules should braid according to 
$C_1 \times \bar C_2$ and $C_2 \times \bar C_3$. 

Then the $\fM^R_{12} \boxtimes_{C_2} \fM^L_{23}$ VOA inherits families of modules 
\begin{equation}
\bigoplus_{\lambda_2} M^{LR}_{12,\lambda_1,\lambda_2} \otimes M^{LR}_{23,\lambda_2,\lambda_3} \in \fM^{LR}_{12} \boxtimes_{C_2} \fM^{LR}_{23}
\end{equation}
and can be used further as an ingredient of more complicated junctions. 

Indeed, the concatenation of interfaces is associative. We can define richer 
VOAs from chains of simpler VOAs: $\fA_{12} \times_{C_2} \fA_{23} \times_{C_3} \fA_{34} \cdots$.

As we only understand the braided tensor category associated to $B_{1,0}$, all our examples will involve tensor products over 
the braided tensor categories of the form $C_G[q]$. 

The general philosophy of how junctions can be concatenated to get new junctions can be illustrated by two simple examples. 

\subsubsection{A junction between $B_{1,0}$ and $B_{-1,n}$}
Consider again a simply-laced, self-dual $G$. 

We can define a semi-classical junction between $B_{1,0}$ and $B_{-1,2}$ by adding two copies of $L_1[G]$ at the junction. The result would be the corner VOA 
\begin{equation}
\fA[\g,\Psi;B_{1,0},B_{-1,2};\CJ] \equiv W^{(2)}_{\Psi}[\g]\simeq \frac{G_{\Psi} \times L_1[G]\times L_1[G]}{G_{\Psi+2}}
\end{equation}

On the other hand, consider the concatenation of a standard junction between $B_{1,0}$ and $B_{-1,1}$ and a standard junction 
between $B_{1,-1}$ and $B_{-1,2}$, both involving a single copy of $L_1[G]$ at each junction. The product of the resulting corner VOAs would be 
\begin{equation}
W_{1+\Psi^{-1}}(\g) \times W_{1+(\Psi+1)^{-1}}(\g) \simeq \frac{G_{\Psi} \times L_1[G]}{G_{\Psi+1}} \times  \frac{G_{\Psi+1} \times L_1[G]}{G_{\Psi+2}}
\end{equation}

This product VOA is actually a subalgebra of $W^{(2)}_{\Psi}[\g]$. Indeed, at least for simply laced $G$ this is a conformal embedding, with decomposition (see the main Theorem of \cite{GKO})
\begin{equation}
W^{(2)}_{\Psi}[\g] = \bigoplus_{\lambda \in Q_+} M_{\Psi_1,\lambda} \otimes {}^L M_{\Psi_2,\lambda}
\end{equation}
Here $Q_+$ is the set of dominant weights that lie in the root lattice and
$\Psi_1=1+\Psi^{-1}, \Psi_2=1+(\Psi+1)^{-1}$.

The modules $M_{\Psi_1,\lambda}$ and ${}^L M_{\Psi_2,\lambda}$ are the modules in $W_{1+\Psi^{-1}}(\g)$ and $W_{1+(\Psi+1)^{-1}}(\g)$ 
associated to Weyl modules for $G_{\Psi+1}$ of weight $\lambda$, which are in turn the modules associated to 
the corresponding topological Wilson lines in $B_{-1,1} \simeq B_{1,-1}$. 

$W^{(2)}_{\Psi}[\g]$ is the conformal extension of $W_{1+\Psi^{-1}}(\g) \times W_{1+(\Psi+1)^{-1}}(\g)$ 
by such a sum of products of modules, precisely of the structure we proposed for the composition of junctions. 

This relation extends to any $n$. We can define a semi-classical junction between $B_{1,0}$ and $B_{-1,n}$ by adding $n$ copies of $L_1[G]$ at the junction. The result would be the corner VOA 
\begin{equation}
 W^{(2n)}_{\Psi}[\g] \simeq \frac{G_{\Psi} \times L_1[G]^{\otimes n}}{G_{\Psi+n}}
\end{equation}

On the other hand, consider the concatenation of $n-1$ standard junctions between $B_{1,0}$ and $B_{-1,1}$, $B_{1,-1}$ and $B_{-1,2}$, etc. 
The product of the resulting corner VOAs would be 
\begin{equation}
\begin{split}
W_{1+\Psi^{-1}}(\g) \times W_{1+(\Psi+1)^{-1}}(\g)& \cdots \times W_{1+(\Psi+n-1)^{-1}}(\g) \simeq \\
\simeq \frac{G_{\Psi} \times L_1[G]}{G_{\Psi+1}}& \times  \frac{G_{\Psi+1} \times L_1[G]}{G_{\Psi+2}} \cdots  \times  \frac{G_{\Psi+n-1} \times L_1[G]}{G_{\Psi+n}}
\end{split}
\end{equation}

This product VOA is actually a subalgebra of $W^{(n)}_{G}$. Indeed, this is a conformal embedding, with decomposition 
\begin{equation}
\begin{split}
W^{(n)}_{G} \cong \bigoplus_{\substack{\lambda_i \in Q_+\\ i=1,\dots, n}} &M_{\Psi_1, \lambda_1} \otimes M_{\Psi_2, \lambda_2,\lambda_1} \otimes M_{\Psi_3, \lambda_3,\lambda_2} \otimes
 \dots  \otimes M_{\Psi_{n-1}, \lambda_{n-1},\lambda_{n-2}} \otimes  {}^L M_{\Psi_n, \lambda_{n-1}}
\end{split}
\end{equation}
involving modules associated to associated to Weyl modules for $G_{\Psi_i}$, with $\Psi_1=1+(\Psi+i-1)^{-1}$, which are in turn the modules associated to 
the corresponding topological Wilson lines in $B_{-1,1} \simeq B_{1,-1}$,$B_{-1,2} \simeq B_{1,-2}$, etc.

This concatenation of interfaces can be associated freely, in the sense that we can always first extend to some $W^{(m)}_{G} \times W^{(n-m)}_{G}$
and then to $W^{(n)}_{G}$. 

\subsubsection{Resolution of a junction between $B_{1,0}$, $B^D_{0,1}[G;H]$ and $B_{-1,0}$}
Consider the corner VOA 
\begin{equation}
\fA[\g,\Psi; B_{1,0}[G],B^D_{0,1}[G;H],B_{-1,0}[H];\CJ_{\mathrm{cl}}] \equiv \frac{G_\Psi}{H_{d_{G/H}(\Psi-h_G) + h_H}}
\end{equation}
associated to an interface $B^D_{0,1}[G;H]$ for $H \subset G$ encountering Neumann boundary conditions for the two groups. 

If we have some intermediate subgroup $H \subset K \subset G$, we can obtain $B^D_{0,1}[G;H]$ as the concatenation of interfaces 
$B^D_{0,1}[G;K]$ and $B^D_{0,1}[K;H]$. The corresponding concatenation of junctions give a product of individual VOAs
\begin{equation}
\frac{G_\Psi}{K_{d_{G/K}(\Psi-h_G) + h_K}} \times \frac{K_{d_{G/K}(\Psi-h_G) + h_K}}{H_{d_{G/H}(\Psi-h_G) + h_H}}
\end{equation}
This is conformally embedded in the original VOA. The embedding gives a decomposition 
\begin{equation}
\frac{G_\Psi}{H_{d_{G/H}(\Psi-h_G) + h_H}} = \bigoplus_{\lambda} M_\lambda \otimes M'_\lambda
\end{equation}
where $M_\lambda$ and $M'_\lambda$ are the modules in $\frac{G_\Psi}{K_{d_{G/K}(\Psi-h_G) + h_K}}$ and $\frac{K_{d_{G/K}(\Psi-h_G) + h_K}}{H_{d_{G/H}(\Psi-h_G) + h_H}}$
associated to Weyl modules for $K_{d_{G/K}(\Psi-h_G) + h_K}$ of weight $\lambda$, which are in turn the modules associated to 
the corresponding topological Wilson lines in the $B_{1,0}$ boundary condition for the $K$ gauge theory.

\subsubsection{A relation between standard junctions}
It is entertaining to look at the concatenation of standard junctions between $B_{1,0}$ and $B_{-1,1}$ and 
between $B_{1,-1}$ and $B^D_{0,1}$. 

The former junction has a corner VOA  
\begin{equation}
W_{1+\Psi^{-1}}(\g) \simeq \frac{G_{\Psi} \times L_1[G]}{G_{\Psi+1}}.
\end{equation}
The latter has a corner VOA $G_{\Psi+1}$. The product of the two VOAs, extended by 
Weyl modules of $G_{\Psi+1}$ associated to lines in $B_{-1,1}$ is clearly nothing but
\begin{equation}
G_{\Psi} \times L_1[G].
\end{equation}

In other words, the composition of these standard junctions gives the standard junction 
between $B_{1,0}$ and $B^D_{0,1}$ dressed by an extra copy of $L_1[G]$. This is pretty obvious in gauge theory and 
natural in the VOA perspective.

Our main objective will be obtained from this example by the replacement $B_{1,0} \to B^D_{1,0}$.

We could have done the same concatenation but use boundary conditions $B_{1,0}$, $B_{-1,1}$ and $B_{0,1}$.
The result would have been somewhat less pleasant, the qDS reduction 
\begin{equation}
\mathrm{DS} \left[G_{\Psi} \times L_1[G] \right]
\end{equation}
of the product of VOAs by the diagonal set of $G$ currents. 

This is likely the same as $W_G[\Psi] \times L_1[G]$, but with a twisted stress tensor for the $L_1[G]$ factor. 
This would be reasonable: the collision would give the standard junction between $B_{1,0}$ and $B_{0,1}$,
dressed by extra decoupled degrees of freedom given by $L_1[G]$.

\subsubsection{A basic Abelian example}
A junction between $B_{1,-1}$ and $B_{0,1}$ will support a $U(1)_{\Psi+1}$ 
vertex algebra, with charge $n$ vertex operators associated to line defects along $B_{1,-1}$ and 
charge $m (\Psi + 1)$ vertex operators associated to line defects along $B_{0,1}$.

Acting with S-duality, we have that a junction between $B_{-1,1}$ and $B_{1,0}$ will support a $U(1)_{\Psi^{-1}+1}$ 
vertex algebra, with charge $n$ vertex operators associated to line defects along $B_{-1,1}$ and 
charge $m (\Psi^{-1} + 1)$ vertex operators associated to line defects along $B_{1,0}$.

Now, consider the concatenation of the junction between $B_{1,-1}$ and $B_{0,1}$ and the junction between $B_{-1,1}$ and $B_{1,0}$,
giving a new junction between $B_{1,0}$ and $B_{0,1}$. This is just the reverse of the Abelian coset we discussed before, 
but it is worth repeating the exercise from this perspective.

The new VOA will have three sets of local operators: 
the $U(1)_{\Psi+1}$ VOA, the $U(1)_{\Psi^{-1}+1}$ VOA and the sum of products of 
with charge $n$ vertex operators associated to line defects along $B_{1,-1}$ and charge $n$ vertex operators associated to line defects along $B_{-1,1}$. 

These products $O_n$ of vertex operators have conformal dimension 
\begin{equation}
\Delta_n = \frac{n^2}{2 (\Psi+1)} + \frac{n^2}{2 (\Psi^{-1}+1)} = \frac{n^2}{2}
\end{equation}
We will identify $O_{\pm 1}$ with a pair of complex fermions. The current local with the free fermions can be taken to be
\begin{equation}
J_\Psi = \frac{1}{\Psi^{-1}+1}(J_{\Psi+1} - J_{\Psi^{-1}+1})
\end{equation}
The resulting VOA is the $U(1)_\Psi$ expected from a bare junction between $B_{1,0}$ and $B_{0,1}$ dressed by an extra 
complex free fermion VOA. This basic example is explained in terms of the corresponding vertex tensor subcategories of the free boson \voa{} and equivalences as braided tensor categories between them in section \ref{sec:catabel}.

\subsubsection{A general Abelian example}
Consider the product 
\begin{equation}
U(1)_\Psi \times U(1)_{n-\Psi}
\end{equation}
extended by the product of magnetic vertex operators of the two theories, which have dimension 
\begin{equation}
\Delta_m = \frac{m^2}{2} \Psi +  \frac{m^2}{2} (n-\Psi) = n  \frac{m^2}{2}
\end{equation}

We can describe this VOA in a simple manner by rotating our basis of currents: 
define 
\begin{equation}
J_n = J_\Psi + J_{n-\Psi} \qquad \qquad J_{\frac{n^2}{\Psi} -n} = (1-\frac{n}{\Psi}) J_\Psi + J_{n-\Psi}
\end{equation}
Then the extension involves vertex operators charged under $J_n$ only, with charge multiple of $n$, building up the standard lattice VOA 
$V[U(1)_n]$. The full VOA is thus 
\begin{equation}
U(1)_{\frac{n^2}{\Psi} -n} \times V[U(1)_n]
\end{equation}

The electric vertex operators for $U(1)_\Psi \times U(1)_{n-\Psi}$ map to vertex operators 
for $U(1)_{\frac{n^2}{\Psi} -n} \times V[U(1)_n]$. The second family has charges $k$ under $J_{\frac{n^2}{\Psi} -n}$ and $k \,\mathrm{mod}\, n$ under $J_n$ for all integer $k$.
The first family has charges $k' (1-\frac{n}{\Psi})$ under $J_{\frac{n^2}{\Psi} -n}$ and $k' \,\mathrm{mod} \,n$ under $J_n$. 

This is the corner VOA one would assign to a junction between $B_{1,0}$ and $B_{1,n}$ in a $U(1)$ gauge theory, 
resolved into junctions between $B_{1,0}$ and $B_{0,1}$ and $B_{0,-1}$ and $B_{1,n}$.

One can iterate this construction to build more general corner VOAs for a $U(1)$ gauge theory. The general result of a concatenation of 
$n+1$ basic junctions will a $U(1)$ Kac-Moody current combined with a rank $n$ lattice VOA, equipped with two families of modules 
built by dressing $U(1)$ Kac-Moody vertex operators with appropriate modules for the lattice VOA.

Without loss of generality, we can take the boundary conditions to be $B_{p,q}$ and $B_{0,1}$. The boundary condition $B_{p,q}$ 
in an $U(1)$ gauge theory can be defined by coupling the 4d theory to a 3d Chern-Simons theory with Abelian gauge group
determined by the specific choice of $p,q$. Intuitively, the latticed VOA is just the VOA living at a boundary for the 
3d Chern-Simons theory. 

In our example, we can map $\Psi \to \Psi^{-1}$ and look at boundary conditions $B_{n,1}$ and $B_{0,1}$. The boundary theory for $B_{n,1}$
is precisely $U(1)_n$. 

It is plausible that the non-Abelian VOA extensions may be treated in a similar manner as this if one employs 
free field realizations. 

\subsection{Generalized Neumann boundary conditions}
Consider a generalized Neumann boundary condition where some set of three-dimensional hypermultiplets are coupled 
both to a four-dimensional gauge group $G$ and to a three-dimensional gauge group $H$. The precise condition for this 
boundary condition to admit a deformation compatible with general $\Psi$ has not yet been established.

A sufficient condition is that the boundary condition or interface can be decomposed into 
simpler boundary conditions or interfaces involving only three-dimensional hypermultiplets, 
with $H$ being realized by four-dimensional gauge theories compactified on a segment. 

The simplest example would be that the hypermultiplets could be combined with $G$ and $H$ into a supergroup $\hat K$. 
Then we could realize the boundary condition as the composition of a $B_{1,0}$ boundary for $H$ gauge theory and 
a $B^{\hat K}_{1,0}$ interface. 

Next, we can ask for junctions and corner VOAs between such composite boundary condition 
and some other boundary condition, say e.g. $B^D_{0,1}[G]$. 

As we only understand well the line defects of $B_{0,1}$ boundary conditions, we can compose a junction between 
$B_{1,0}[H]$ and $B_{0,1}[H]$ and a junction between $B_{0,1}[H]$, $B^{\hat K}_{1,0}$ and $B^D_{0,1}[G]$.
This choice can be interpreted as a specific choice for the boundary condition 
of the three-dimensional gauge fields and matter fields at the corner, a three-dimensional version of 
a Nahm pole boundary condition. 

The resulting corner VOA $\fA[\g,\Psi;B_{1,0}[H],B^{\hat K}_{1,0},B^D_{0,1}[G];\CJ]$ will be an extension of 
\begin{equation}
DS_{\mathrm{reg}_H} H_{\Psi+n_H} \times DS_{\mathrm{reg}_H} \hat K_{\Psi}
\end{equation}
by products of magnetic modules associated to lines in $B_{0,1}[H]$, which should be the 
images of spectral flow modules for $H$ under the qDS reduction. 

We can treat more complicated examples in the same manner. 

\subsection{The $\Psi \to \infty$ limit}
It is important to remark that the $\Psi \to \infty$ limit of families of VOAs with $G_\Psi$ Kac-Moody sub-algebras 
gives VOAs equipped with an outer $G$ automorphism. These VOAs are rather special, in the sense that 
they can be coupled to $G$ bundles with holomorphic flat connections in an algebraic manner,
by identifying the rescaled $G$ currents in the VOA OPE with the holomorphic connection. 
This is very important for Geometric Langlands applications.

\subsubsection{Relation to VOAs in three-dimensional gauge theory}
We have just discussed boundary conditions of Neumann type associated to 
coupling to a three-dimensional gauge theory $T$ with three-dimensional gauge group $H$ and 
matter fields sitting as odd generators in a supergroup $\hat K$. 

It is interesting to inquire about the $\Psi \to \infty$ of the corner configuration involving such boundary conditions 
and $B^D_{0,1}[G]$. In that limit the coupling of the four-dimensional gauge theory can be taken to be very weak and 
the four-dimensional degrees of freedom essentially decouple from the three-dimensional degrees of freedom. 

The resulting corner VOA $\fA[\g,\infty;B^T_{1,0},B^D_{0,1};\CJ_B]$ should be closely related to the VOA $A_C[T,B]$ which emerges at (deformed $(0,4)$) boundary conditions 
for the three-dimensional gauge theory subject to a Rozansky-Witten twist, discussed in the upcoming work \cite{CG}.

The VOA $A_C[T,B]$ is not fully understood at the moment, as it includes generators arising as boundary monopole operators
whose identity and OPE relations are somewhat mysterious. The perturbative generators, though, for $B$ being Dirichlet boundary conditions, form
a certain graded super-Kac-Moody algebra whose bosonic generators consist of two copies of $H$ in degrees $0$ and $2$ and whose fermionic generators 
live in degree $1$ and are associated to the hypermultiplets. Central elements valued in the global symmetry 
group $G$ can also be included at degree $2$. The specific details of the construction depend on the choice of boundary condition. 

The super-Kac-Moody algebra turns out to coincide with the $\Psi \to \infty$ limit of $H_{\Psi+n_H} \times \hat K_{\Psi}$.
In the limit we need to rescale the currents judiciously in order to keep the OPE coefficients finite. 
The total $H$ currents have finite level, as the $H$ currents in $\hat K_{\Psi}$
have level $-\Psi$. They can be kept finite in the limit. The fermionic currents in $\hat K$
need to be rescaled by a power of $\Psi^{\frac12}$, and a non-trivial OPE with the total $H$ currents. 
The remaining bosonic currents need to be rescaled by a power of $\Psi$ and become the degree $2$ 
components of the superKac-Moody algebra.

If $B$ is a regular Nahm pole boundary condition, we expect the perturbative part of $A_C[T,B]$
to be the $\Psi \to \infty$ limit of 
\begin{equation}
DS_{\mathrm{reg}_H} H_{\Psi+n_H} \times DS_{\mathrm{reg}_H} \hat K_{\Psi}
\end{equation}
This suggests that the full $A_C[T,B]$ could be obtained from the $\Psi \to \infty$ limit of $\fA[\g,\Psi;B_{1,0}[H],B^{\hat K}_{1,0},B^D_{0,1}[G];\CJ]$

We will be able to test this idea in simple examples, using the fact that $A_C[T,B]$
may admit a mirror description, which is considerably simpler and better understood \cite{Gaiotto:2016wcv,CG}. 

\subsubsection{Fermionic currents VOA and $U(1)$ flat connections}
Consider the VOA generated by two fermionic currents of dimension $1$, i.e. a $PSU(1|1)$ Kac-Moody algebra,
with OPE 
\begin{equation}
x(z) y(w) \sim \frac{1}{(z-w)^2}
\end{equation}
This VOA has an $SU(2)$ global symmetry. We will first focus on an $U(1)$ subgroup. 

The $U(1)$ symmetry is only global because we have no $U(1)$ current in the algebra. 
That means we have no good way to couple the VOA to a $U(1)$ bundle. For example, a gauge transformation 
$x(z) \to g(z) x(z)$ and $y(z) \to g(z)^{-1} y(z)$ changes the OPE to 
\begin{equation}
x(z) y(w) \sim \frac{g(z)^{-1} g(w)}{(z-w)^2} \sim \frac{1}{(z-w)^2} - \frac{g(w)^{-1} \partial g(w)}{z-w}
\end{equation}

We can cure this problem if the $U(1)$ bundle is equipped with an holomorphic connection $A(z)$. Then the OPE 
\begin{equation}
x(z) y(w) \sim \frac{1}{(z-w)^2} + \frac{A(w)}{z-w}
\end{equation}
is gauge invariant! That allows one to define conformal blocks for the VOA coupled to a $U(1)$ 
bundle with connection. The conformal blocks will depend algebraically on the connection. 

This OPE arises from the $\Psi \to \infty$ limit of $SU(1|1)_\Psi$: the bosonic generator, rescaled, 
becomes the central element $A(z)$ and the fermionic generators, rescaled, become the $PSU(1|1)$
currents. 

We can also include coupling to an $SU(2)$ connection: 
\begin{equation}
x(z) x(w) \sim  \frac{A^+ (w)}{z-w} \qquad x(z) y(w) \sim \frac{1}{(z-w)^2} + \frac{A(w)}{z-w} \qquad y(z) y(w) \sim  \frac{A^- (w)}{z-w} 
\end{equation}
This OPE arises from the $\Psi \to \infty$ limit of $OSp(1|2)_\Psi$.

We will encounter richer examples in later sections. 

\subsection{Good and bad compositions}
There is a point which is worth making here. It is very convenient working with VOAs where the currents have dimensions 
bounded from below, so that each $L_0$ eigenspace is finite-dimensional. Even if this condition holds for $\fA_{12}$ and $\fA_{23}$,
it may fail for $\fM_{12} \boxtimes_{C_2} \fM_{23}$ if the dimensions of the modules which appear in the construction 
are unbounded from below.

In the physical untwisted gauge theory, quarter-BPS junctions between half-BPS boundary conditions often require the boundaries to have 
a specific slope in the plane orthogonal to the junction. For example, the boundary conditions $B_{p,q}$ should have slopes controlled by $p \tau + q$. 
A conformal-invariant quarter-BPS junction in the physical theory will have local operators of positive scaling dimension. The property will be inherited by the 
corresponding VOA. 

The concatenation of two physical quarter-BPS junctions into a single one may already be tricky: the R-charge of local operators 
is well-defined before the scaling limit, but the IR R-symmetry may differ from the R-symmetry of the UV
concatenation of junctions, usually due to the decoupling of some free degrees of freedom in the scaling limit. 
Such a ``bad'' collision may thus give VOAs which mildly fail to have bounded scaling dimensions, 
but the problem should be solvable by a judicious re-definition of the stress tensor for the decoupled 
degrees of freedom. For example, the decoupled system may consist of some $bc$ system with non-positive dimension for $c$,
which can be corrected to a system of free fermions. We saw a potential example earlier on, where the decoupled degrees of freedom 
consisted of $L_1[G]$. 

A more serious obstruction occurs if we concatenate junctions in a manner which 
may be available in the topologically twisted theory, but not in the underlying untwisted physical theory,
as it violates the slope constraints. An example would be the concatenation of a junction between $B^{\cdots}_{1,2}$ and $B^{\cdots}_{-1,1}$ and 
one between $B^{\cdots}_{1,-1}$ and $B^{\cdots}_{2,1}$. In these examples we do indeed find that the concatenation involves 
products of modules with scaling dimensions unbounded from below. 

\section{The (quantum) Geometric Langlands kernel VOAs}\label{sec:three}

We want to define a VOA $\mathfrak{A}[G, \Psi]$ associated to a gauge group $G$ and a 
complex parameter $\Psi$, which arise at the junction between Dirichlet boundary conditions for a gauge theory of gauge group $G$ 
and the S-duality image of Dirichlet boundary conditions for the dual gauge theory with dual gauge group ${}^L G$. 

For general $\Psi$, this VOA is expected to be endowed with two Kac-Moody subalgebras, associated respectively to $G$ and ${}^L G$.
In this section we will specialize again to simply-laced groups with Abelian factors added to make them self-dual, 
i.e. ${}^L G = G$. 

In the language of the previous section, we are after a corner VOA between $B^D_{1,0}$ and $B^D_{0,1}$ boundary conditions.
We will build one in a straightforward manner: we concatenate the standard junction between $B^D_{0,1}$ and $B_{1,-1}$ 
and the standard junction between $B_{-1,1}$ and $B^D_{1,0}$.

\subsection{The basic extension}

Our candidate for $\mathfrak{A}[G, \Psi]$ is the conformal extension of a product VOA of the form 
\begin{equation}
G_{\Psi + 1} \times G_{\Psi^{-1} + 1}
\end{equation}
by a sum of products of Weyl modules associated to the topological lines on the $B_{1,-1}$ boundary. See Figure \ref{fig:basic}.
This sum simply runs over all finite-dimensional representations:
\[
\AG = \bigoplus_{\lambda} M_{\Psi+1, \lambda} \otimes M_{\Psi^{-1}+1, \lambda}
 \] 

For general groups, this is as much as we can do. It is a relatively precise definition, in the sense that 
the conformal extension should be determined by the braided inverse equivalence between the category of 
Weyl modules of $G_{\Psi + 1}$ and the category of Weyl modules of $G_{\Psi^{-1} + 1}$.

In particular, the space of conformal blocks for $\mathfrak{A}[G, \Psi]$ should be a sub-space 
of the product of conformal blocks for $G_{\Psi + 1}$ and $G_{\Psi^{-1} + 1}$,
determined by a projector built from the algebra element encoding the braided equivalence. 
It should be possible to describe this in a mathematically concise fashion as the 
space of (derived) sections of some D-module. We will not attempt to do so here. 

For specific classical groups we can provide further information on the 
resulting VOA by leveraging extra information about the S-duality of 
special interfaces. 

Before doing so, we observe that we can build generalizations $\mathfrak{A}^{p,q}[G, \Psi]$
of $\mathfrak{A}[G, \Psi]$ as corner VOAs between $B^D_{p,q}$ and $B^D_{0,1}$. These encode generalized 
quantum Geometric Langlands relationships associated to general S-duality elements. 

We specialize again to simply-laced, self-dual $G$. Then a junction between $B^D_{n,1}$ and $B^D_{0,1}$
resolved through a $B_{1,0}$ segment gives a conformal extension of 
\begin{equation}
G_{\Psi} \times G_{\Psi'}
\end{equation}
with 
\begin{equation}
\frac{1}{\Psi} + \frac{1}{\Psi'} = n
\end{equation}
of the form 
\[
\AGn = \bigoplus_{\lambda} M_{\Psi, \lambda} \otimes M_{\Psi', \lambda}
 \] 
See Figure \ref{fig:basictwo}.

More general junctions may require a sequence of $B_{a,b}$ intervals and thus 
will involve an extensions of product VOAs of  the form
\begin{equation}
G_{\Psi_0} \times W_{\Psi_1}(\g) \times \cdots \times G_{\Psi_{m+1}}
\end{equation}
by modules of the form 
\[
\Agni = \bigoplus_{\lambda_i} M_{\Psi_0, \lambda_0} \otimes \left(\otimes_{i=1}^m M^W_{\Psi_i, \lambda_{i-1},\lambda_i}\right) \otimes M_{\Psi_{m+1}^{-1}, \lambda_m}
 \] 
See Figure \ref{fig:basictwo}.
 
\subsubsection{The Geometric Langlands kernel} 
If we take the $\Psi \to \infty$ limit of the algebra $\mathfrak{A}[G, \Psi]$ 
we obtain a simpler algebra, the conformal extension of $G_{1}$ by modules of the form 
\[
\AGo = \bigoplus_{\lambda} R_{\lambda} \otimes M_{1, \lambda}
 \] 
 where $R_{\lambda}$ is the weight $\lambda$ finite-dimensional representation of an outer 
 automorphism $G_{\mathrm{out}}$ global symmetry. Notice that $G_{1}$ is the Kac-Moody algebra at critically shifted level $1$, 
 i.e. standard level $1-h$. 
 
 A similar structure had been conjectured for $SU(2)$ and $SU(3)$ gauge groups in \cite{Gaiotto:2016wcv}. We will discuss $SU(2)$ further in the next section. 
 
The conformal blocks of $\AGo$ on a Riemann surface equipped with an $G_{\mathrm{out}}$
flat connection conjecturally coincide with the Hecke eigensheaves labelled by the same flat connection. 
\footnote{The fact that such conformal blocks can be defined in a manner which is algebraic in the $G_{\mathrm{out}}$
flat connection is rather non-trivial and it is intimately related to the fact that the $G_{\mathrm{out}}$ outer automorphism 
symmetry is the remnant of a $G$ Kac-moody algebra whose level is sent to infinity. }

\subsection{qDS reduction of $\mathfrak{A}[G, \Psi]$}
The quantum DS reduction of $\mathfrak{A}[G, \Psi]$ by a regular embedding in 
one of the $G$ Kac-Moody algebras is rather striking: it gives an extension of 
\begin{equation}
G_{\Psi + 1} \times W_{G}[\Psi^{-1} + 1]
\end{equation}
which we have encountered before: $G_\Psi \times L_1[G]$! 
Indeed, the qDS reduction is the VOA manifestation of an operation which maps $B^D_{1,0}$ to $B_{1,0}$. 
See also Figure \ref{fig:basic}.

We can check this claim in full detail for $G=U(2)$. We will do so in Theorem \ref{thm:DS}. The generalization of statements 
to all simply laced $G$ uses \cite{GKO} and will be presented in future work.  

This fact is particularly significant in light of the expected relationship with the quantum Geometric Langlands program. 
We believe it is the VOA version of a crucial statement relating the twisted D-module on the space of $G$-bundles which quantizes the oper manifold and 
the D-module of all twisted differential operators  on the space of $G$-bundles \cite{Gait2}.

The relationship is strengthened further by the observation that the qDS reduction of spectral flow modules for the $G$ Kac-Moody algebra
results in Weyl modules for $G_\Psi$ (see Section \ref{sec:DSspectralflow} for the $SU(2)$ case). This statement also has a natural quantum Geometric Langlands interpretation. 

Finally, the $\Psi \to \infty$ limit of this statement becomes a statement about the properties of $\AGo$ when it is coupled to 
a $G_{\mathrm{out}}$ flat connection which is actually an oper. One can show that the coupling to an oper reduces the 
VOA to a central quotient of $G_0 \times L_1[G]$. 

The $L_1[G]$ VOA has trivial one-dimensional conformal blocks. Thus we hope to recover a well known fact about the Geometric Langlands correspondence: 
the Hecke eigensheaves labelled by opers are conformal blocks of a a central quotient of  the Kac-Moody algebra at critical level. 

\begin{figure}[p]
\begin{center}
\begin{tikzpicture}
\path [fill=lightgray] (0,0) -- (2,0) -- (2,3) -- (-1,3) -- (-1,1);
\draw[->,ultra thick] (0,0) --(2,0);
\draw[->] (0,0) --(-.5,.5);
\draw (0,0) --(-1,1);
\draw[->,ultra thick] (-1,1) --(-1,3);
\node at (.5,1.5) {$G$,$\Psi$};
\node at (-.8,.2) {$B_{1,-1}$};
\node at (-1.5,2) {$B^D_{1,0}$};
\node at (1,-.3) {$B^D_{0,1}$};
\draw[fill] (0,0) circle [radius=0.05];
\draw[fill] (-1,1) circle [radius=0.05];
\end{tikzpicture}
\hspace{2cm}
\begin{tikzpicture}
\path [fill=lightgray] (0,0) -- (2,0) -- (2,3) -- (-1,3) -- (-1,1);
\draw[->,ultra thick] (0,0) --(2,0);
\draw[->] (0,0) --(-.5,.5);
\draw (0,0) --(-1,1);
\draw[->] (-1,1) --(-1,3);
\node at (.5,1.5) {$G$,$\Psi$};
\node at (-.8,.2) {$B_{1,-1}$};
\node at (-1.5,2) {$B_{1,0}$};
\node at (1,-.3) {$B^D_{0,1}$};
\draw[fill] (0,0) circle [radius=0.05];
\draw[fill] (-1,1) circle [radius=0.05];
\end{tikzpicture}
\end{center}
\caption{An important concatenation of junctions. Left: In order to find a junction between $B^D_{1,0}$ and $B^D_{0,1}$ 
we interpolate with an $B_{1,-1}$ segment. This resolution defines the algebra $\mathfrak{A}[G,\Psi]$ as an extension the product of 
Kac-Moody $G$ at level $k+h = \Psi^{-1}+1$ and Kac-Moody $G$ at level $k'+h = \Psi+1$.
Right: The same construction applied to $B_{1,0}$ and $B^D_{0,1}$ gives the product of 
Kac-Moody $G$ at level $k+h = \Psi$ and the minimal WZW model $L_1[G]$. This is related to the previous construction by a qDS reduction. 
}\label{fig:basic}
\end{figure}
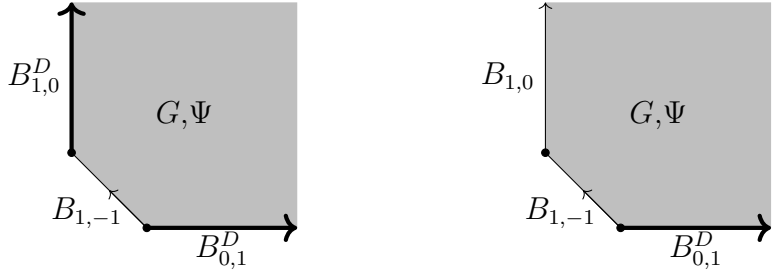

\begin{figure}[h]
\begin{center}
\begin{tikzpicture}
\path [fill=lightgray] (0,0) -- (2,-1) -- (2,3) -- (-1,3) -- (-1,1);
\draw[->,ultra thick] (0,0) --(2,-1);
\draw[->] (0,0) --(-.5,.5);
\draw (0,0) --(-1,1);
\draw[->,ultra thick] (-1,1) --(-1,3);
\node at (.5,1.5) {$G$,$\Psi$};
\node at (-.8,.2) {$B_{1,-1}$};
\node at (-1.5,2) {$B^D_{1,0}$};
\node at (1,-.8) {$B^D_{-1,2}$};
\draw[fill] (0,0) circle [radius=0.05];
\draw[fill] (-1,1) circle [radius=0.05];
\end{tikzpicture}
\hspace{3cm}
\begin{tikzpicture}
\path [fill=lightgray] (0,0) -- (2,-1) -- (5,-2) -- (5,3) -- (-1,3) -- (-1,1);
\draw[->,ultra thick] (2,-1) --(5,-2);
\draw[->] (0,0) --(1,-.5);
\draw (0,0) --(2,-1);
\draw[->] (0,0) --(-.5,.5);
\draw (0,0) --(-1,1);
\draw[->,ultra thick] (-1,1) --(-1,3);
\node at (.5,1.5) {$G$,$\Psi$};
\node at (-.8,.2) {$B_{1,-1}$};
\node at (-1.5,2) {$B^D_{1,0}$};
\node at (1,-.8) {$B_{-1,2}$};
\node at (3.5,-2) {$B^D_{-1,3}$};
\draw[fill] (0,0) circle [radius=0.05];
\draw[fill] (-1,1) circle [radius=0.05];
\end{tikzpicture}
\end{center}
\caption{More intricate resolutions. Left: A useful way to resolve a junction between $B^D_{1,0}$ and $B^D_{-1,2}$. The final junction VOA conformally extends the product of 
$G^\vee$ Kac-Mody at $k+h = \Psi^{-1}+1$ and Kac-Moody $G$ at level $k'+h = \frac{\Psi+1}{\Psi+2}$. This can be generalized to any positive integer. 
Right: A useful way to resolve a junction between $B^D_{1,0}$ and $B^D_{-1,3}$.  The final junction VOA conformally extends the product of 
$G^\vee$ Kac-Mody at $k+h = \Psi^{-1}+1$,  $W_G$ at level $k'+h = \frac{\Psi+1}{\Psi+2}$ and Kac-Moody $G$ at level $k''+h = \frac{\Psi+2}{\Psi+3}$.
}\label{fig:basictwo}
\end{figure}
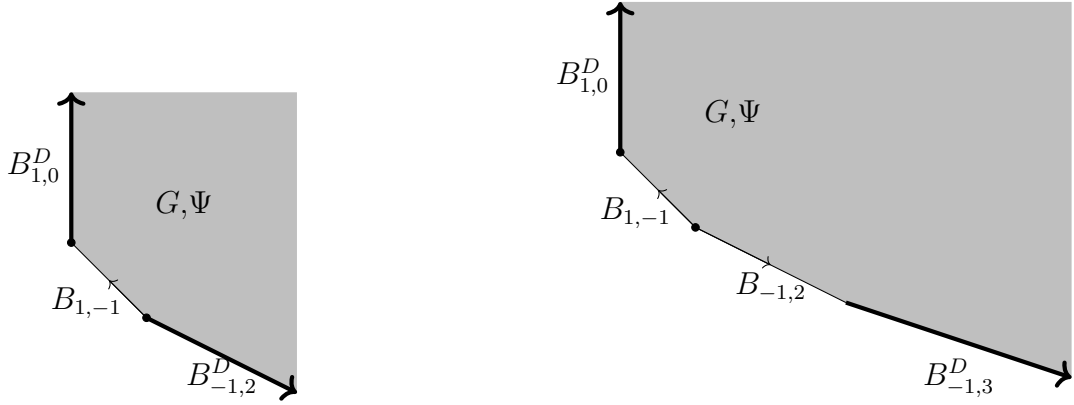

\section{The $U(2)$ kernel and an unexpected exceptional supergroup.} \label{sec:four}
We will now implement our prescription for the quantum Geometric Langlands kernel for the $U(2)$ gauge group. 

\subsection{Simple junction and $U(1) \times \Vir$}
The junction between $B_{1,0}$ and $B_{0,1}$ in $U(2)$ gauge theory supports a $W_{\psi}(U(2)) \equiv U(1)_{2\Psi} \times \Vir_{b^2 = - \Psi}$ 
vertex algebra. 

The relevant families of modules combine degenerate modules of Virasoro and $U(1)$ modules of appropriate charge: 
\begin{itemize}
\item Charge $(n_1,n_2)$ Wilson line defects along $B_{1,0}$ end on the module generated by a primary of dimension 
\begin{equation}
\Delta_{(n_1,n_2)} = \frac{n_1^2 + n_1}{2 \Psi} - \frac{n_1}{2} + \frac{n_2^2-n_2}{2 \Psi}+ \frac{n_2}{2}
\end{equation}
which is the combination of a charge $n_1 + n_2$ vertex operator for $U(1)_{2\Psi}$ and a degenerate Virasoro primary 
of type $(1,1+n_1-n_2)$. Here $n_1 \geq n_2$. 
\item Charge $(m_1,m_2)$ 't Hooft line defects along $B_{0,1}$ end on the module generated by a primary of dimension 
\begin{equation}
\Delta_{(m_1,m_2)} = \frac{m_1^2 + m_1}{2}\Psi + \frac{m_2^2-m_2}{2}\Psi
\end{equation}
which is the combination of a charge $(m_1 +m_2)\Psi$ vertex operator for $U(1)_{2\Psi}$ and a degenerate Virasoro primary 
of type $(1+m_1-m_2,1)$. Here $m_1 \geq m_2$
\end{itemize}
Notice that the two sets of vertex operators are mutually local and map to each other under $\Psi \to \Psi^{-1}$. 

\subsection{Mixed junction and $U(2)_\Psi$}
The junction between $B_{1,0}$ and $B^{D}_{0,1}$ in $U(2)$ gauge theory supports a $U(2)_{\Psi}$ 
Kac-Moody algebra. 

The level of the diagonal $U(1)$ sub-algebra is $2 \Psi$. The $U(2)_\Psi$ Weyl modules are labelled by a weight 
$(n_1,n_2)$ with $n_1 \geq n_2$ and have dimensions proportional to the $U(2)$ Casimir: 
\begin{equation}
\Delta_{(n_1,n_2)} = \frac{n_1^2 + n_1}{2 \Psi} + \frac{n_2^2-n_2}{2 \Psi}
\end{equation}
These modules live at the end of charge $(n_1,n_2)$ Wilson line defects along $B_{1,0}$.

The set of line defects available along $B^{D}_{0,1}$ is potentially rather rich and not obvious from the gauge theory description. 
They will give a rich collection of modules for $U(2)_\Psi$ which are mutually local with 
the Weyl modules. We will discuss some of them later. 

\subsection{Collision of mixed junctions and the $D(2,1;-\Psi)$ exceptional supergroup.}
A junction between $B_{1,-1}$ and $B^D_{0,1}$ will support a $U(2)_{\Psi+1}$ 
vertex algebra, with charge $(n_1,n_2)$ vertex operators associated to line defects along $B_{1,-1}$.

Acting with S-duality, we have that a junction between $B_{-1,1}$ and $B^D_{1,0}$ will support a $U(2)_{\Psi^{-1}+1}$ 
vertex algebra, with charge $(n_1,n_2)$ vertex operators associated to line defects along $B_{-1,1}$.

Now, consider the concatenation of the junction between $B_{1,-1}$ and $B^D_{0,1}$ and the junction between $B_{-1,1}$ and $B^D_{1,0}$,
giving a new junction between $B^D_{1,0}$ and $B^D_{0,1}$. The new VOA will have three sets of local operators: 
the $U(2)_{\Psi+1}$ VOA, the $U(2)_{\Psi^{-1}+1}$ VOA and the sum of products of 
with charge $(n_1,n_2)$ vertex operators associated to line defects along $B_{1,-1}$ 
and charge $(n_1,n_2)$ vertex operators associated to line defects along $B_{-1,1}$. 

These products $O_{n_1,n_2}$ of vertex operators are Weyl modules for $U(2)_{\Psi+1}\times U(2)_{\Psi^{-1}+1}$ 
and have conformal dimension 
\begin{equation}
\Delta_{(n_1,n_2)} = \frac{n_1^2 + n_1}{2} + \frac{n_2^2-n_2}{2}
\end{equation}
We conjecture that the resulting vertex algebra $\mathfrak{A}[U(2), \Psi]$ is a current algebra for the supergroup $D(2,1;-\Psi)\times U(1)_{2 \Psi}$.

The overall $U(1)$ should be the anti-diagonal combination of the centers of $U(2)_{\Psi+1}\times U(2)_{\Psi^{-1}+1}$. 
It gives a current which we can normalize as $U(1)_{2 \Psi}$ and has trivial OPE with the $O_{n_1,n_2}$
vertex operators. 

The second linear combination of the centers of $U(2)_{\Psi+1}\times U(2)_{\Psi^{-1}+1}$ 
can be taken to be the level $2$ current
\begin{equation}
J_{2}=  \frac{1}{\Psi+1} J_{2 \Psi + 2} + \frac{\Psi}{\Psi+1}J_{2 \Psi^{-1}+2} 
\end{equation}
so that $O_{n_1,n_2}$ have charge $\frac{n_1 + n_2}{2}$. 

The $O_{(1,1)}$ and $O_{(-1,-1)}$ vertex operators are dimension $1$ currents which are singlets of $SU(2)_{\Psi+1}\times SU(2)_{\Psi^{-1}+1}$
and have charge $\pm 1$ under $J_{2}$. Together they form an $SU(2)_1$ WZW current algebra. 
The $SU(2)_{\Psi+1}\times SU(2)_{\Psi^{-1}+1} \times SU(1)_1$ algebra is the bosonic subalgebra of $D(2,1;-\Psi)_1$.

On the other hand, $O_{(1,0)}$ and $O_{(0,-1)}$ are dimension $1$ currents which transform in doublet representations of
the $SU(2)_{\Psi+1}\times SU(2)_{\Psi^{-1}+1}$ subalgebra. They are naturally identified with the odd generators of $D(2,1;-\Psi)_1$.

The conformal embedding of $SU(2)_{\Psi+1}\times SU(2)_{\Psi^{-1}+1}\times SU(2)_1 \subset D(2,1;-\Psi)_1$ has been discussed a while ago
\cite{BFST}. The vacuum character decomposes exactly as described above, see Theorem \ref{thm:dec}. 

\subsection{Further resolutions of the full junction}
The analysis of \cite{BFST} uses intensively the relation of an $SU(2|1)$ current
algebra to $D(2,1;-\Psi)_1$. We describe this in section \ref{sec:constVOAs}.
This structure has a neat interpretation in the language of gauge theory junctions. 

We can recover the structure by resolving $B^{D}_{1,0}$ into a combination of a $B_{1,0}$ interface between $U(2)$ and $U(1)$ gauge theory 
and a $B_{1,0}$ boundary condition for $U(1)$. The full junction is then resolved into 
a junction between the $B_{1,0}$ interface and the $B_{0,1}^D$ boundary conditions for $U(2)$ and $U(1)$ gauge theories 
and the standard junction for the $U(1)$ gauge theory. 

These junctions are well understood and support respectively an $U(2|1)_\Psi$ and 
an $U(1)_\Psi$  Kac-Moody algebras. The full junction VOA should be an extension of that 
product VOA by modules associated to 't Hooft line defects on the $B_{0,1}$ boundary for the $U(1)$ theory. 

We already encountered the corresponding modules for $U(1)_\Psi$: they are generated by vertex operators of charge $m \Psi$,
which can be thought of as spectral flow operators for $U(1)_\Psi$. We have not encountered before the corresponding modules for 
$U(2|1)_\Psi$, but we expect them to be also spectral flow modules for the super-group Kac-Moody algebra, 
associated to $m$ units of spectral flow for the block-diagonal $U(1)$ subgroup. 
The corresponding decomposition of $D(2,1;-\Psi)_1$ is precisely described in \cite{BFST}.

We can also take the resolved junction of the previous section and decompose both $B^{D}_{1,0}$ 
and $B^{D}_{0,1}$ for $U(2)$ into interfaces from $U(2)$ to $U(1)$ and boundary conditions for $U(1)$. 

The ``maximally resolved'' picture involves the product of four VOAs: 
\begin{equation}
U(1)_\Psi \times \frac{U(2|1)_\Psi}{U(2)_{\Psi + 1}} \times \frac{U(2)_{\Psi + 1}}{U(1)_\Psi} \times U(1)_\Psi
\end{equation}
This set of VOAs is invariant under $\Psi \to \Psi^{-1}$, exchanging the outer pair of factors and the inner pair of factors,
as $\frac{U(2|1)_\Psi}{U(2)_{\Psi + 1}} \simeq \frac{U(2)_{\Psi^{-1} + 1}}{U(1)_{\Psi^{-1}}}$.
This product should be extended by three sets of modules, corresponding to appropriate line defects:
\begin{itemize}
\item The product of magnetic modules for $U(1)_\Psi$ and the coset image in $\frac{U(2|1)_\Psi}{U(2)_{\Psi + 1}}$ of spectral flow modules for $U(2|1)_\Psi$.
\item The product of coset modules in $\frac{U(2|1)_\Psi}{U(2)_{\Psi + 1}}$ associated to Weyl modules of $U(2)_{\Psi + 1}$ and the coset image in $\frac{U(2)_{\Psi + 1}}{U(1)_\Psi}$ 
of Weyl modules of $U(2)_{\Psi + 1}$.
\item The product of coset modules in $\frac{U(2)_{\Psi + 1}}{U(1)_\Psi}$ associated to Weyl modules of $U(1)_\Psi$ and electric modules for $U(1)_\Psi$.
\end{itemize}

If we only include the second and third sets of modules, we reassemble the sequence of nesting cosets 
to 
\begin{equation}
U(1)_\Psi \times U(2|1)_\Psi
\end{equation}
If we include the first two sets we reassemble a dual sequence of cosets to 
\begin{equation}
U(2|1)_{\Psi^{-1}} \times U(1)_\Psi
\end{equation}
If we include the first and third sets we get the original product of Kac-Moody algebras:
\begin{equation}
U(2)_{\Psi^{-1} + 1} \times U(2)_{\Psi + 1}
\end{equation}

The gauge theory picture predicts the existence of two families of modules for the full VOA, 
associated to line defects in either $B^{D}_{1,0}$ or $B^{D}_{0,1}$.
The set of line defects living on these boundary conditions is not fully understood, but 
we can at least discuss the line defects which result from line defects in the 
resolved junction. 

The first set of modules is then simply induced from Weyl modules for $U(1)_\Psi \times U(2|1)_\Psi$.
The second set is induced from Weyl modules for $U(2|1)_{\Psi^{-1}} \times U(1)_\Psi$. Gauge theory predicts 
these two sets of modules will be mutually local. 

\subsection{qDS reduction}

It is natural to consider the qDS reduction of $\mathfrak{A}[U(2), \Psi]$ by either of the two Kac-Moody 
$SU(2)_{\Psi^{\pm 1} -1}$ subalgebras. In gauge theory terms, this should correspond to replacing 
either $B^D_{1,0}$ or $B^D_{0,1}$ with $B_{1,0}$ or $B_{0,1}$ and is thus expected to 
produce something similar to $U(2)_{\Psi^{\pm 1}}$.

If we focus on a sub-algebra 
\begin{equation}
U(2)_{\Psi^{-1} + 1} \times U(2)_{\Psi + 1}
\end{equation}
then the qDS reduction of the second factor gives 
\begin{equation}
U(2)_{\Psi^{-1} + 1} \times (U(1)_{2 \Psi + 2} \times \Vir_{b^2 = - \Psi-1})
\end{equation}
The coset description of Virasoro allows one to rewrite that as 
\begin{equation}
U(2)_{\Psi^{-1} + 1} \times \frac{U(2)_{\Psi^{-1} } \times \mathrm{Ff}^2}{U(2)_{\Psi^{-1} + 1} }
\end{equation}
where $\mathrm{Ff}^2$ is the VOA generated by two complex fermions. 

The extension of $U(2)_{\Psi^{-1} + 1} \times U(2)_{\Psi + 1}$ to $\mathfrak{A}[U(2), \Psi]$ 
involves precisely the products of modules which extend the above product to 
$U(2)_{\Psi^{-1} } \times \mathrm{Ff}^2$. 

This is reasonable. The qDS reduction of $D(2,1;-\Psi)_1$ 
makes four of the odd generators into free fermions. Stripping off these free fermions 
kills the $SU(2)_1$ currents. It is less obvious, but apparently true, that the total stress tensor 
is equal to the sum of the free fermion stress tensor and the  $U(2)_{\Psi^{-1} }$ Sugawara stress tensor.  We prove these statements in Theorem \ref{thm:DS}.

If we start from the sub-algebra 
\begin{equation}
U(1)_\Psi \times U(2|1)_\Psi
\end{equation}
and do the qDS reduction, we may use the relation between the qDS reduction of $U(2|1)_\Psi$
and ${\cal N}=2$ super-Virasoro and the Kazama-Suzuki coset description of ${\cal N}=2$ super-Virasoro in order to 
recover in a different way the identification with $U(2)_{\Psi^{-1} } \times \mathrm{Ff}^2$.

We can say a bit more about modules. The qDS reduction of spectral flow modules of $ U(2)_{\Psi + 1}$ 
gives ``magnetic'' degenerate modules for $U(1)_{2 \Psi + 2} \times \Vir_{b^2 = - \Psi-1}$.
In turns, these induce Weyl modules of $U(2)_{\Psi^{-1}}$. 

\subsection{The $\Psi \to \infty$ limit}
It is easy to see that the $\Psi \to \infty$ limit of $D(2,1;-\Psi)_1$ 
is $PSU(2|2)_1$: the $SU(2)_{1+\Psi}$ Kac-Moody subalgebra 
becomes the $SU(2)$ outer automorphism of $PSU(2|2)_1$.

In particular, conformal blocks for $PSU(2|2)_1$ should give the kernel for 
$SU(2)$ Geometric Langlands. 

We expect the VOA $PSU(2|2)_1$ to arise as a boundary VOA for the Rozansky-Witten twist 
of $T[SU(2)]$, aka $SQED_2$, the three-dimensional gauge theory which can be used to define $B_{1,0}^D$.
We refer to \cite{CG} for confirmation of this statement. 

In this particular case, the mirror dual description of the boundary condition is also known, 
as a natural boundary condition for the mirror Rozansky-Witten twist 
of $T[SU(2)]$, which is a self-mirror theory. The mirror description of the boundary condition involves 
boundary degrees of freedom given by the pair of complex fermions, necessary to cancel a 
certain boundary gauge anomaly.  

The mirror description of the VOA is that of a $U(1)$ BRST coset of 
the product of two symplectic bosons and two complex fermions. 
The symplectic bosons arise from the two hypermultiplets in $SQED_2$
and the $U(1)$ BRST coset from the $U(1)$ gauge field in $SQED_2$.

That free VOA before the $U(1)$ BRST coset gives an $U(2|2)_1$ 
WZW model. The $U(1)$ BRST coset reduces it to $PSU(2|2)_1$.

We should compare this proposal with an earlier proposal for the kernel VOA for 
$SU(2)$ Geometric Langlands \cite{Gaiotto:2016wcv}. The main difference lies in the choice of auxiliary boundary degrees of freedom:
a pair of complex fermions versus a $U(1)_2$ lattice VOA. The latter is obtained from the former by a coset by $SU(2)_1$. 

In particular, the VOA of \cite{Gaiotto:2016wcv} can be described as the coset 
\begin{equation}
\frac{PSU(2|2)_1}{SU(2)_1}
\end{equation}
which makes manifest the $SU(2)$ global symmetry conjectured there. 

This coset has a generalization for generic $\Psi$:
\begin{equation}
\frac{D(2,1;-\Psi)_1}{SU(2)_1}
\end{equation}

\subsection{$U(2)$ vs $SU(2)$} 
Another perspective on the VOA of \cite{Gaiotto:2016wcv} becomes available if we base the gauge theory description 
of an $SU(2)$ gauge theory, rather than a $U(2)$ one. This raises some subtleties concerning the global form of the gauge group, 
but they are manageable. 

If we completely ignore these subtleties, we could simply write down the obvious answer for the candidate 
VOA at the junction between $B^D_{1,0}[SU(2)]$ and $B^D_{0,1}[SU(2)]$:
the extension of $SU(2)_{\Psi+1}\times SU(2)_{\Psi^{-1}+1}$
by products of Weyl modules
\begin{equation}
\mathfrak{A}[SU(2);\Psi] = \bigoplus_{j \in \mathbb{Z}} M_{j,\Psi +1} \otimes  M_{j,\Psi^{-1} +1} 
\end{equation}
In order to get integral dimensions, we only sum over Weyl modules associated to representations of integral spin. 

This reproduces the answer from the previous section:
\begin{equation}
\mathfrak{A}[SU(2);\Psi] = \frac{D(2,1;-\Psi)_1}{SU(2)_1}
\end{equation}
In the $\Psi \to \infty$ limit this matches the proposal of \cite{Gaiotto:2016wcv}.

In order to use this VOA in Geometric Langlands applications, though, we likely need to understand better 
the global properties of the gauge group and how do they reflect on the calculation of conformal blocks. 
We leave that for future work. 

\section{Classical simply-laced Lie groups} \label{sec:five}
In this section we describe some interesting sub-algebras of $\mathfrak{A}[G, \Psi]$ 
for classical groups $G$. 
\subsection{$\mathfrak{A}[U(N), \Psi]$}
Our knowledge of S-duality properties of boundary conditions and junctions is 
much better developed for $U(N)$ gauge theory than for other gauge groups, thanks to 
string constructions involving brane webs \cite{Gaiotto:2017euk}. We may employ that knowledge to 
seek further information about the $\mathfrak{A}[U(N),\Psi]$ VOA.

A particularly important duality fact is that the $S$ operation relates 
the boundary condition $B^{U(N|M)}_{1,0}$ and $B^{\rho_{N-M}}[U(N);U(M)]$ for $N>M$, 
with $\rho_{N-M}$ being the embedding were the fundamental representation of $\mathfrak{u}(N)$ 
is decomposed into an irreducible $\mathfrak{su}(2)$ representation of dimension $N-M$ and 
$M$ copies of the trivial representation. We can denote the corresponding family of interfaces as 
$B_{p,q}^{N|M}$. 

There is a large set of conjectural equivalences between vertex algebras which follow from the 
conjecture that string theory junctions involving a $(1,0)$, a $(0,1)$ and a	$(1,1)$ fivebranes 
should be duality co-variant \cite{Gaiotto:2017euk}. 

These equivalences can be used to understand junctions involving $B^D_{1,0}$. 
The basic idea is that $B^D_{0,1}[U(N)]$ can be decomposed into an
of $B^{N|N-1}_{0,1}$ interface and a $B^D_{0,1}[U(N-1)]$ boundary condition.
The S-dual decomposition breaks down $B^D_{1,0}[U(N)]$ into an
$B^{N|N-1}_{1,0}$ interface and a $B^D_{1,0}[U(N-1)]$ boundary condition.

For example, the basic decomposition of $B^D_{0,1}[U(N)]$ applied to a junction with $B_{1,0}$ leads to the obvious coset relation \footnote{In order to understand the level shifts, 
recall that the critical level for $U(N)$ is $-N$ and the non-shifted level of the $U(N-1)$ sub-algebra equals the 
non-shifted level of the $U(N)$ Kac-Moody algebra. Hence the non-critically shifted levels are both $\Psi-N+1$.}
\begin{equation}
\frac{U(N)_{\Psi+1}}{U(N-1)_{\Psi}} \times U(N-1)_{\Psi} \subset  U(N)_{\Psi+1}
\end{equation}
The action of $S$ and covariance of junctions gives the far less obvious 
\begin{equation}
 \frac{DS_{N-1} U(N|N-1)_{\Psi}}{U(N)_{\Psi+1}} \times  U(N-1)_{\Psi^{-1}}\subset  U(N)_{\Psi^{-1}+1}
\end{equation}
which follows from the non-trivial VOA equivalence \footnote{In order to understand the level shifts, 
recall that the critical level for $U(N|M)$ is $M-N$ and hence the non-shifted level of $U(N|N-1)_{\Psi^{-1}}$ is $\Psi^{-1}-1$. 
The DS reduction modifies the non-shifted level of the $U(N)$ sub-algebra from $\Psi^{-1}-1$ to the desired $\Psi^{-1}-N+1$.}
\begin{equation}
\frac{U(N)_{\Psi+1}}{U(N-1)_{\Psi}} \simeq \frac{DS_{N-1} U(N|N-1)_{\Psi^{-1}}}{U(N)_{\Psi^{-1}+1}}  
\end{equation}
where the modules associated to Weyl modules of $U(N)_{\Psi+1}$ and $U(N-1)_{\Psi}$ respectively 
map to modules associates to Weyl modules of $U(N)_{\Psi^{-1}+1}$ and spectral flow modules 
for the DS reduction. 

We can apply this resolution of $B^D_{0,1}[U(N-1)]$ to the duality kernel setup in a manner analogous to what we did for 
$U(2)$. We start from the conformal embedding
\begin{equation}
U(N)_{\Psi+1} \times U(N)_{\Psi^{-1}+1} \subset  \mathfrak{A}[U(N),\Psi]
\end{equation}
and decompose it further 
\begin{equation}
U(N)_{\Psi+1} \times \frac{U(N)_{\Psi^{-1}+1}}{U(N-1)_{\Psi^{-1}}} \times U(N-1)_{\Psi^{-1}} \subset U(N)_{\Psi+1} \times U(N)_{\Psi^{-1}+1} \subset  \mathfrak{A}[U(N),\Psi]
\end{equation}
map it to 
\begin{align}
U(N)_{\Psi+1} &\times \frac{DS_{N-1} U(N|N-1)_{\Psi}}{U(N)_{\Psi+1}} \times  U(N-1)_{\Psi^{-1}} \cr &\subset U(N)_{\Psi+1} \times U(N)_{\Psi^{-1}+1} \subset  \mathfrak{A}[U(N),\Psi]
\end{align}
and associate the composition of junctions in an alternative manner:
\begin{align}
U(N)_{\Psi+1} & \times \frac{DS_{N-1} U(N|N-1)_{\Psi}}{U(N)_{\Psi+1}} \times  U(N-1)_{\Psi^{-1}} \cr
&\subset DS_{N-1} U(N|N-1)_{\Psi} \times  U(N-1)_{\Psi^{-1}} \subset  \mathfrak{A}[U(N),\Psi]
\end{align}

As a result, we have expressed $\mathfrak{A}[U(N),\Psi]$ as the extension of a product of three VOAs 
by two mutually local sets of modules and given explicitly the result of the extension 
by either set.

\subsection{$\mathfrak{A}[SO(2N), \Psi]$}
In this section we give some alternative descriptions for $\mathfrak{A}[SO(2N), \Psi]$. Recall the basic embedding 
\begin{equation}
SO(2N)_{\Psi +1} \times SO(2N)_{\Psi^{-1} +1}  \subset \mathfrak{A}[SO(2N), \Psi]
\end{equation}
Recall that the critical level for $SO(2N)$ is $2-2N$, hence the above critically shifted levels correspond to 
non-shifted levels of $\Psi - 2N +3$ and $\Psi^{-1} - 2N +3$.

The extension involves a sum over products of Weyl modules in the same representation 
for the two groups. There are some potential subtleties here concerning the global form of the gauge group
and the precise choice of representations which should be allowed in the sum. We will mostly neglect them. 

Consider the sub-algebra 
\begin{equation}
SO(2N-1)_{\Psi} \times \frac{SO(2N)_{\Psi +1}}{SO(2N-1)_{\Psi}} \times SO(2N)_{\Psi^{-1} +1} \subset \mathfrak{A}[SO(2N), \Psi]
\end{equation}
We can use the identification 
\begin{equation}
\frac{SO(2N)_{\Psi +1}}{SO(2N-1)_{\Psi}} \simeq \frac{DS_{Sp(2N-2)} OSp(2N|2N-2)_{\Psi^{-1}}}{SO(2N)_{\Psi^{-1} +1} }
\end{equation}
which can be derived from brane constructions, to find a different conformal embedding: 
\begin{equation}
SO(2N-1)_{\Psi} \times DS_{Sp(2N-2)} OSp(2N|2N-2)_{\Psi^{-1}} \subset \mathfrak{A}[SO(2N), \Psi]
\end{equation}

Notice that this works as stated if in the initial sum over Weyl modules we include the representations which enter in the 
coset $\frac{DS_{Sp(2N-2)} OSp(2N|2N-2)_{\Psi^{-1}}}{SO(2N)_{\Psi^{-1} - 2N +3} }$, which are all $SO(2N)$ 
representations rather than $Spin(2N)$ representations. We leave a precise interpretation of this fact to future work. 

\section{$\mathfrak{A}[SO(2N+1), \Psi] \sim \mathfrak{A}[Sp(2N), \Psi^{-1}]$} \label{sec:six}
In this section we discuss how to extend our construction to classical non-simply laced groups. 
We will use duality statements which follow from brane constructions 
for orthogonal and symplectic groups. 

As a preparation, we should discuss the action of S-dualities on 
$SO(2N+1)$ and $Sp(2N)$ gauge theories.

It is useful to normalize the topological coupling of an $SO(2N+1)$ gauge theory as 
$\Psi$ and the coupling of an $Sp(2N)$ gauge theory as 
$\Psi/2$. In that normalization, the S operation maps $SO(2N+1)$ to $Sp(2N)$ in 
the usual $\Psi \to \Psi^{-1}$ manner. 

The unusual normalization of the  $Sp(2N)$ gauge coupling means that a T operation
does not leave the $Sp(2N)$ gauge theory invariant. It is useful to define 
an $Sp(2N)'$ gauge theory of coupling $\Psi/2$ as coinciding with 
an $Sp(2N)$ gauge theory of coupling $(\Psi+1)/2$. 
Then $Sp(2N)$ and $Sp(2N)'$ are mapped into each other by T and $Sp(2N)'$ is mapped to itself by S. 

Neumann boundary conditions $B_{1,0}[Sp(2N)]$ for $Sp(2N)$ gauge theory map under S 
to a regular Nahm pole boundary condition $B_{0,1}[SO(2N+1)]$ for $SO(2N+1)$ gauge theory. 

Similarly, Neumann boundary conditions $B_{1,0}[SO(2N+1)]$ for $SO(2N+1)$ gauge theory map under S to a regular Nahm pole boundary condition $B_{0,1}[Sp(2N)]$ for 
$Sp(2N)$ gauge theory. On the other hand, a regular Nahm pole boundary condition $B_{0,1}[Sp(2N)']$ for $Sp(2N)'$ gauge theory
map to a modified Neumann boundary condition $B_{1,0}^{OSp(1|2N)}[Sp(2N)']$.

Notice that the space of boundary 't Hooft lines at a regular Nahm pole boundary condition for $Sp(2N)$ and $Sp(2N')$ boundary conditions 
should coincide. They must match respectively the space of boundary Wilson lines for $B_{1,0}[SO(2N+1)]$ and for $B_{1,0}^{OSp(1|2N)}[Sp(2N)']$.
That means there should be a bijection between representations of $SO(2N+1)$ and of $OSp(1|2N)$.

\subsection{A coset description of $W_{SO(2N+1)} \simeq W_{Sp(2N)}$}
At the junction of $B_{1,0}[SO(2N+1)]$ and $B_{0,1}[SO(2N+1)]$ we should find 
the non-simply laced W-algebra 
\begin{equation}
W_{SO(2N+1)}[\Psi] \equiv DS_{\mathrm{reg}} SO(2N+1)_\Psi
\end{equation}
On the other hand, at the junction of $B_{1,0}[Sp(2N)]$ and $B_{0,1}[Sp(2N)]$ we should find 
the non-simply laced W-algebra
\begin{equation}
W_{Sp(2N)}[\Psi/2] \equiv DS_{\mathrm{reg}} Sp(2N)_{\Psi/2}
\end{equation}
The two junctions are expected to be related by S duality and the two W-algebras 
accordingly coincide up to $\Psi \to \Psi^{-1}$. 

This equivalence shows the existence of two families of modules, associated to finite-dimensional 
representations of $SO(2N+1)$ and of $Sp(2N)$ respectively. 

It is useful to act with $T$ on the latter setup, to get a junction between $B_{1,-1}[Sp(2N)']$ and $B_{0,1}[Sp(2N)']$.
This gives a third description of the same W-algebra: 
\begin{equation}
W_{Sp(2N)}[\Psi/2+1/2] \equiv \frac{OSp(1|2N)_{-\Psi^{-1}}}{Sp(2N)_{\Psi^{-1}/2+1/2}}
\end{equation}
In this description, modules associated to Weyl modules of $Sp(2N)_{\Psi^{-1}/2+1/2}$ 
give the family of modules associated to $Sp(2N)$ representations, while 
 modules associated to Weyl modules of $OSp(1|2N)_{-\Psi^{-1}}$ 
give the family of modules associated to $SO(2N+1)$ representations.

\subsection{Another family of W-algebras}
It is also interesting to consider the coset 
\begin{equation}
\tilde W_{SO(2N+1)}[\Psi+1] \equiv \frac{SO(2N+1)_{\Psi^{-1}}\times L_1[SO(2N+1)]}{SO(2N+1)_{\Psi^{-1}+1}}
\end{equation}
The usual manipulations of boundary conditions lead to 
\begin{equation}
\tilde W_{SO(2N+1)}[\Psi] \equiv DS_{Sp(2N)} OSp(1|2N)_{-\Psi}
\end{equation}

\subsection{The candidate GL kernels}
We can now proceed as usual to study a junction between $B^D_{1,0}[SO(2N+1)]$ 
(i.e. the S dual of $B^D_{0,1}[Sp(2N)]$) and $B^D_{0,1}[SO(2N+1)]$.

The resulting basic conformal embedding becomes
\begin{equation}
SO(2N+1)_{\Psi +1} \times OSp(1|2N)_{-\Psi^{-1} -1}  \subset \mathfrak{A}[SO(2N+1), \Psi]
\end{equation}

This vertex algebra behaves well under DS reductions by the $Sp(2N)$ Kac-Moody sub-algebra: it reduces to 
$SO(2N+1)_{\Psi}\times L_1[SO(2N+1)]$\footnote{Recall that the VOA of $2N+1$ fermions is a simple current extension of $L_1(so(2N+1))$ and here this is (as usual in CFT) meant by $L_1[SO(2N+1)]$}. This is proven for $N=2$ in Theorem \ref{thm:ospreduction}.

We can resolve further 
\begin{equation}
SO(2N)_{\Psi} \times \frac{SO(2N+1)_{\Psi +1}}{SO(2N)_{\Psi} } \times OSp(1|2N)_{-\Psi^{-1}- 2N}  \subset \mathfrak{A}[SO(2N+1), \Psi]
\end{equation}
and use 
\begin{equation}
\frac{SO(2N+1)_{\Psi +1}}{SO(2N)_{\Psi} } \simeq \frac{DS_{SO(2N-1)}OSp(2N|2N)_{\Psi^{-1}}}{OSp(1|2N)_{-\Psi^{-1}}}
\end{equation}
to get the second conformal embedding 
\begin{equation}
SO(2N)_{\Psi} \times DS_{SO(2N-1)}OSp(2N|2N)_{\Psi^{-1}}  \subset \mathfrak{A}[SO(2N+1), \Psi]
\end{equation}

We can specialize to $N=1$ as a check: we have embeddings 
\begin{equation}
SO(2)_{\Psi} \times \frac{SO(3)_{\Psi+1}}{SO(2)_{\Psi} } \times OSp(1|2)_{-\Psi^{-1}- 1}  \subset \mathfrak{A}[SO(3), \Psi]
\end{equation}
partially extended to 
\begin{equation}
SO(3)_{\Psi+1} \times OSp(1|2)_{-\Psi^{-1}- 1}  \subset \mathfrak{A}[SO(3), \Psi]
\end{equation}
and 
\begin{equation}
SO(2)_{\Psi} \times OSp(2|2)_{\Psi^{-1}} \subset \mathfrak{A}[SO(3), \Psi]
\end{equation}

We still expect $\mathfrak{A}[SO(3), \Psi]$ to coincide with $D(2,1;-\Psi)_1$.
The latter conformal embedding seems a re-formulation of the $U(1) \times SU(2|1)$ embedding. 

The former conformal embedding is novel and proven in Corollary \ref{cor:ospascoset}.

We can also study a junction between $B^D_{1,0}[Sp(2N)']$ 
(i.e. the S dual of $B^D_{0,1}[Sp(2N)']$) and $B^D_{0,1}[Sp(2N)']$.

The resulting basic conformal embedding becomes
\begin{equation}
Sp(2N)_{\frac{\Psi +1}{2}} \times Sp(2N)_{\frac{\Psi^{-1} +1}{2}} \subset \mathfrak{A}[Sp(2N)', \Psi]
\end{equation}

This vertex algebra behaves well under DS reductions either of the two $Sp(2N)$ Kac-Moody sub-algebras: it reduces to 
$OSp(1|2N)_{-\Psi^{\pm 1}}$.

\section{An S-duality action on VOAs} \label{sec:seven}
\subsection{Intermission: good and bad VOAs}
Our definition of $\mathfrak{A}^{(n)}[G,\Psi]$ could in principle 
work for all integral $n$, but is best behaved for positive $n$, 
where the dimensions of the VOA generators are positive and 
the $L_0$ eigenspaces are finite-dimensional. 

For $n=0$ we have infinitely many generators of dimension $0$, while for negative $n$ 
the dimensions are badly unbounded from below. This makes a lot of manipulations of the VOAs 
more complicated or ill-defined. 

In the following, we will restrict ourselves to positive $n$. 

\subsection{Convolution of VOAs over $G$}

If we are given two VOAs $V_\kappa$ and $V'_{-\kappa}$, with $G$ Kac-Moody sub-algebras of
opposite critically shifted levels, we can combine $V_\kappa \times V'_{-\kappa}$ with a set of $bc$ ghosts valued in $\g$ and 
take the cohomology with respect to a standard BRST charge which makes the total $G$ currents BRST exact \cite{Karabali:1989dk,Hwang:1993nc}.
We denote the resulting new VOA as 
\begin{equation}
V_\kappa \boxtimes_\g V'_{-\kappa}
\end{equation}
Essentially, this is the result of gauging the chiral $G$ symmetry acting on the product VOA. 

It is known that this type of BRST quotient can be used as an alternative definition 
of a coset: 
\begin{equation}
V_\kappa \boxtimes_\g G_{-\kappa} \simeq \frac{V_\kappa}{G_\kappa}
\end{equation}
The coset modules associated to Weyl modules of $G_\kappa$ can be obtained 
from the BRST reduction of the corresponding modules of $G_{-\kappa}$ combined with the vacuum module
of $V_\kappa$, etc. 

If we write $V'_{-\kappa}$ as a conformal extension of $\frac{V'_{-\kappa}}{G_{-\kappa}} \times G_{-\kappa}$,
we can derive a reasonable conformal extension 
\begin{equation}
\frac{V_\kappa}{G_\kappa} \times \frac{V'_{-\kappa}}{G_{-\kappa}} \subset V_\kappa \boxtimes_\g V'_{-\kappa}
\end{equation}

It should be almost obvious that 
\begin{equation}
\mathfrak{A}^{(n)}[G,\Psi] \boxtimes_\g \mathfrak{A}^{(m)}[G,\Psi'] \simeq \mathfrak{A}^{(n+m)}[G,\Psi] 
\end{equation}
with $\Psi^{-1} = n + (\Psi')^{-1}$. 

If we intersperse some $L_1[G]$ we can get the more complicated extensions involving W-algebras at intermediate steps. 
For example, 
\begin{equation}
\mathfrak{A}^{(n)}[G,\Psi] \otimes_\g \left(L_1[G] \times \mathfrak{A}^{(m)}[G,\Psi']\right) \simeq \mathfrak{A}^{(n+1,m+1)}[G,\Psi] 
\end{equation}

\subsection{The S-duality action as convolution}
It is useful to introduce some new notation. As before, denote as $S$ and $T$ the $PSL(2,\mathbb{Z})$
generators $\Psi \to \Psi^{-1}$ and $\Psi \to \Psi+1$. We define the following operations on a VOA with a Kac-Moody $G$ 
sub-algebra at level $\kappa$: 
\begin{equation}
S T^n S: V_\kappa \to (S T^n S\circ V)_{\frac{\kappa}{1+n \kappa}} \equiv \mathfrak{A}^{(n)}[G,\frac{\kappa}{1+n \kappa}] \otimes_\g V_\kappa
\end{equation}
We can write $(S T^n S\circ V)_{\frac{\kappa}{1+n \kappa}}$ as a conformal extension of the form 
\begin{equation}
G_{\frac{\kappa}{1+n \kappa}} \times \frac{V_\kappa}{G_\kappa} \subset (S T^n S\circ V)_{\frac{\kappa}{1+n \kappa}} 
\end{equation}
These operations compose in an appropriate manner. 

We are also induced to define 
\begin{equation}
T: V_\kappa \to \left(T \circ V \right)_{\kappa + 1} \simeq L_1[G] \times V_\kappa
\end{equation}

We can try to extend our definition to negative $n$. This is tricky to do so using the auxiliary VOAs. We can attempt to use 
the conformal extension above as a definition. This will work particularly well if the 
coset modules for $\frac{V_\kappa}{G_\kappa}$ have anomalous dimensions which grow fast enough 
with the weight of the representation. A trivial example is $V_\kappa = G_\kappa$. Then the coset is trivial and the image of $ST^nS$ is $G_{\frac{\kappa}{1+n \kappa}}$ for all 
$n$. 

If we restrict ourselves to positive $n$, the $STS$ and $T$ operations do not satisfy extra relations in 
the $SL(2,\mathbb{Z})^+$ group of positive $SL(2,\mathbb{Z})$ matrices. If we allow negative $n$, 
we may ask if these operations really represent faithfully some duality operations. 
We do not expect this to be the case at the level of vertex algebras, though hopefully it will be true 
at the level of conformal blocks for the VOAs. 

We can gain some insight by looking at $U(1)$ examples. If $V_\kappa$ is a lattice VOA, then $S T^n S$  will map it to a lattice VOA 
of the same rank, but $T$ will increase the rank by $1$, as it tensors by a free complex fermion VOA. 
No operation will lower the rank. This makes clear that no relations of the form $T^{n_1} S T^{n_2} S \cdots T^{n_N}=1$ 
can hold at the level of VOA. 

There is another natural operation we may want to consider: sent $V_\kappa$ to the extension of 
\begin{equation}
DS_{\mathrm{reg}} V_\kappa \times G_\kappa'
\end{equation}
by spectral flow modules for $DS_{\mathrm{reg}} V_\kappa$ combined with Weyl modules of $G_{\kappa'}$. 
This will be possible if 
\begin{equation}
\kappa + \kappa'^{-1} = n
\end{equation}
If we define the generator $\tilde S: \kappa \to - \kappa^{-1}$, 
this is a good candidate for $\left( \tilde S T^{-n} \circ V\right)_{\kappa'}$, especially well defined for 
positive $n$. 

This new generators of $SL(2,\mathbb{Z})$ satisfy reasonable relations with the previously defined ones. 
For example, $(\tilde S T^{-1})\circ(STS)$ coincides with $T$, thanks to the coset definition of $W(\g)$ algebras. 
More generally, $(\tilde S T^{-n})\circ(ST^mS) = (ST^{n-1}S) T (ST^{m-1}S)$.

\subsection{Gauge theory interpretation}

We believe these constructions and results can be given a gauge theory interpretation as consequences of a general principle: 
four-dimensional gauge theory can be cut along a hyperplane by imposing Dirichlet boundary conditions on the two sides of the cut. 
The cut can be healed by gauging back the diagonal combination of the $G$ global symmetries at the two Dirichlet boundaries. 
We expect that the VOA formulae describe the healing of cuts which pass through a junction and split it into two junctions, each 
with a Dirichlet boundary condition. 

We expect that the above operations on VOAs should be compatible with the S-duality action on boundary conditions, in the sense that 
if $V_\kappa$ arises at a junction between some boundary $B$ and a $B^D_{0,1}[G]$ boundary, then 
$(g \circ V)_\kappa$ will arise at some junction between a boundary $g \circ B$ and a $B^D_{0,1}[G]$ boundary.

In particular, all these transformations should descend, at the level of conformal blocks, 
to duality actions on D-modules on the space of twisted $G$ bundles. 

\section{Some mathematical comments} \label{sec:eight}

Let us fix some conventions and notations so that we can relate physics and mathematics terminology. Let $\g$ be a reductive Lie algebra.
We denote by $V_k(\g)$ the universal \avoa{} of $\g$ at level $k$ and its simple quotient by $L_k(\g)$. For generic complex choice of $k$ the two coincide. 
In conformal field theory these \voas{} arise as the chiral algebras of the Wess-Zumino-Witten (WZW) theories based on a usually compact real Lie group $G$ with Lie algebra $\g$. 
The WZW theory is then often denoted as $G_k$. In the main text, though, we instead use the convention $G_{k+h^\vee} \equiv V_k(\g)$. 

Given an \avoa{} there are two standard constructions of new \voas, the coset construction and quantum Drinfeld-Sokolov (or quantum Hamiltonian reduction). 
Let $V$ be a \voa{} and $T$ a \vosa, then the set $C$ of fields that has regular OPE with all fields of $T$ is called the coset \voa{} of $T$ in $V$. 
It is denoted as $C=\Com(T, V)$ for commutant in the \voa{} literature while in physics one denotes the conformal field theory based on $C$ by a quotient symbol $C=V/T$. 
It is widely believed that most DS-reductions also allow for a realization as coset algebra. These believes are usually based on coinciding sets of strong generators and on comparison of characters. The coset realization of regular $W$-algebras of simply laced Lie algebras has finally been established \cite{GKO} and it seems likely that the techniques used in that work will allow to prove many more coset realizations of DS-reductions.

\subsection{Affine \voas}

Let
\[
\widehat{\mathfrak{g}} = \mathfrak{g} \otimes \mathbb C[t, t^{-1}] \oplus \mathbb C K \oplus \mathbb C d 
\]
be the affinization of a simple Lie algebra $\mathfrak g$. We set $\mathfrak g_{\pm} =  \mathfrak{g} \otimes t^{\pm 1}\mathbb C[ t^{\pm 1}]$ and $\mathfrak g_0 = \mathfrak g \oplus \mathbb C K \oplus \mathbb C d$. 
Let $\rho: \mathfrak g \rightarrow  \text{End}(V)$ be an irreducible representtion of $\mathfrak g$ which is extended to an representation of $\mathfrak g_0 \oplus \mathfrak g_+$ by letting $K$ act by multiplication with the scalar $k$ and $d$ and $\mathfrak g_+$ act as zero. 
One then defines
\[
V_k(\rho) := \text{Ind}^{\widehat g}_{\mathfrak g_0 \oplus \mathfrak g_+} \rho.
\]
If $\rho$ is an irreducible highest-weight representation of highest-weight $\lambda$ we just write $V_{k}(\lambda)$ and for its simple quotinet $L_k(\lambda)$. Note that for generic choice of $k$ the module $V_{k}(\lambda)$ is allready simple. It's often called the Weyl module of weight $\lambda$ at level $k$. In the instance of $\mathfrak g=\sltwo$ the weight lattice is just $\mathbb Z \omega_1$ with $\omega_1$ the fundamental weight. In this case we just write $L_k(m)$ for $L_k(m\omega_1)$. 
Let $\Psi=k-h^\vee$ with $h^\vee$ the dual Coxeter number. In the physics part we write $M_{\Psi, \lambda}$ for $L_k(\lambda)$. 

\subsection{Regular $W$-algebras}\label{sec:W}

The notation here is taken from \cite{GKO}. Given a simple Lie algebra $\g$ and an embedding $\rho$ of $\sltwo$ in $\g$ one has a functor that associates to $V_k(\g)$ another \voa{} called the $W$-algebra of $\g$ for the embedding $\rho$ at level $k$. The most familiar instance is the regular embedding with corresponding regular $W$-algebra $W^k(\g)$. We denote the functor from $V_k(\g)$ to the $W$-algebra by $H_{DS}$ for Drinfeld-Sokolov. The simple quotient of $W^k(\g)$ is denoted by $W_k(\g)$ and for generic $k$ the two coincide. 

Let $\gamma : Z(\g) \rightarrow \mathbb C$ be a central character then the Verma module of $W^k(\g)$ of weight $\gamma$ is denoted by $M(\gamma)$ and its simple quotient by $L(\gamma)$. The evaluation of $Z(\g)$ on the Verma  module of highest-weight $\lambda$ of $\g$ defines a central character which is denoted by $\gamma_\lambda$. One then has by Theorem 9.1.4 of \cite{Ara07} that $H_{DS}(L_k(\lambda))\cong L(\gamma_{\lambda -\Psi \rho^\vee})$ for generic $k$. 

It had been common physics belief that for simply-laced $\g$ the regular $W$-algebras can be realized as coset \voas. This is finally proven in \cite{GKO} for both  generic $k$ and for the minimal series. In the generic instance the answer is:
\begin{enumerate}
\item 
As \voas
$$\text{Com}(V_{k+1}(\g),V_k(\g)\otimes L_1(\g))\cong W^\ell(\g),$$
where $\ell+h^{\vee}=(k+h^{\vee})/(k+h^{\vee}+1)$.
\item For $\mu\in P_+$ and $\nu\in P^1_+$,
\begin{align*}
V_k(\mu)\otimes L_1(\nu)\cong \bigoplus_{\substack{\lambda\in P_+
\\ \lambda-\mu-\nu\in Q}} V_{k+1}(\lambda)\otimes \mathbf{L}^\ell(\gamma_{\mu-(\ell+h^{\vee})(\lambda+\rho^{\vee})}).
\end{align*}
\end{enumerate}
For us it is convenient to write 
\[
M_{\Psi, \lambda, \mu} := \mathbf{L}^\ell(\gamma_{\mu-(\ell+h^{\vee})(\lambda+\rho^{\vee})}, \qquad \Psi= \ell +h^\vee. 
\]
We also note that for the special case of $\sltwo$ the reduction is the Virasoro algebra and Virasoro modules are usually labelled by a pair of integers $(r, s)$ which relate to the weight labels via $(\lambda, \mu)=((r-1)\omega_1, (s-1)\omega_1)$ so that we will write $M_{\Psi,\lambda, \mu}= M^\Psi(r, s)$. If $\Psi$ is a rational number $u/v$, then we will sometimes replace it by the tuple $(u, v)$. 

\subsection{Extending \voas}\label{sec:examplesVOAextensions}

Recently, Shashank Kanade, Robert McRae and one of us \cite{CKM} have developped a tensor theory for vertex operator superalgebra extension building on \cite{KO, HKL}. The starting point is the notion of a commutative algebra object $A$ in the vertex tensor category $\mathcal C$ of a given \voa, see Definition 2.2 of \cite{CKM}; and section 2.2 for the notion of supercategoires and superalgebras. 
Then there is a one-to-one correspondence between vertex operator (super)algebra extensions of $V$ and commutative (super)algebra objects in $\mathcal C$ \cite{HKL} (and \cite{CKL} for the super case). Moreover the main result of \cite{CKM} is that the category of local (super)algebra modules is braided equivalent to the vertex (super)tensor category of the extended vertex operator (super)algebra. Another result is the following \cite[Theorem 2.67]{CKM}. Let $\mathcal C^0$ the full subcategory of modules of $\mathcal C$ with the property that they cenralize $A$, i.e. $X$ in $\mathcal C^0$ if and only if the momodromy $M_{A, X}=c_{A, X}\circ c_{X,A}$ is the identity in $A\boxtimes_{\mathcal C}X$. Here $c_{A, X} :A\boxtimes X \rightarrow X\boxtimes A$ is the braiding. Then the statement is that the induction functor $\mathcal F$ restricted to $\mathcal C^0$,
\[
\mathcal F: \mathcal C^0 \rightarrow \mathcal C_A^{\text{local}}, \qquad X\mapsto A\boxtimes_{\mathcal C} X
\]
 is a braided tensor functor. We will see in a moment, Proposition \ref{prop:ff}, that in our set-up it will be even a fully faithfull functor.

Let $V$ be a \voa{} and $C$ be a full vertex tensor category of $V$-modules. In particular, $C$ is braided. We assume that every object is completely reducible. Let $\overline{C}$ be the opposite category, i.e. braiding is reversed. Then one has always the regular representation or coend $R(C)$, which is the commutative and associative algebra object in $C \boxtimes \overline{C}$ of the form 
\[
R(C) \cong \bigoplus_{X} X \otimes \overline{X}
\]
where the sum is over all inequivalent simple objects of $C$. Let now $W$ be a second \voa{} with full subcategory of modules $D$ braided equivalent to $\overline{C}$ then
one has the corresponding commutative and associative algebra object in the module category of $V\otimes W$ 
\begin{equation}\label{eq:extVOA}
A \cong \bigoplus_{X} X \otimes \tau(\overline{X})
\end{equation}
with $\tau$ mapping objects of $\overline{C}$ to the equivalent ones in $D$. As mention above, according to \cite{HKL} such algebra objects are equivalent to \voa{} extension, i.e. $A$ is a \voa{} extending $V\otimes W$ and conversely $V\otimes W$ embeds conformally in $A$. 
If $C$ and $D$ are fusion, then it is actually a consequence of \cite{DMNO} (see \cite{Lin, OS}) that having a \voa{} extension of $V\otimes W$ is possible if and only if $D$ is braid reversed equivalent to $C$. 

We will see that we are interested in subcategories $C, \widetilde{C}$ of the same \voa{} that are local to each other. The category term for this is that $\widetilde{C}\subset C'$ and $C\subset \widetilde{C}'$ as well where $C'$ denotes the centralizer of $C$, i.e. all those modules that have trivial monodromy with those of $C$. To make contact with the set-up for junction \voas{}, let $A$ be a \voa{} that is an extension of $V\otimes W$ as above. Let $\widetilde{C}$ respectively $\widetilde{D}$ subcategories of $V$ resp. $W$-modules that centralize $C$ resp. $D$. Then these modules induce to modules of the big \voa{} $A$, i.e. for any pair $\widetilde{X}\otimes \widetilde{Y}$ in $\widetilde{C} \boxtimes\widetilde{D}$ the module
\[
\mathcal F\left(  \widetilde{X}\otimes \widetilde{Y} \right) \cong A \boxtimes_{V\otimes W} \widetilde{X}\otimes \widetilde{Y}
\]
is a $A$-module \cite{HKL}. Moreover the induction functor restricted to $\widetilde{C} \boxtimes\widetilde{D}$ is a braided tensor functor \cite{CKM} and it is a natural expectation that it even furnishes an equivalence of braided subtensor categories. 

We can actually generalize a little bit, i.e. we can assume that $A$ as in equation \eqref{eq:extVOA} is a vertex operator superalgebra. It then corresponds to a superalgebra object in the representation category of $V\otimes W$ \cite{CKL}.
A useful statement seems to be
\begin{prop}\label{prop:ff} 
Consider the set-up of above with $A$ as in equation \eqref{eq:extVOA}. For $Y, Z$ in $C'$ then 
\[
\text{Hom}_A\left( \mathcal F(Y \otimes W), \mathcal F(Z \otimes W)\right) \cong \text{Hom}_{C'}(Y, Z).
\]
\end{prop}\label{prop:ff}
\begin{proof}
The idea for the argument is similar to the proof of Theorem 5. 1 of \cite{OS}. Namely 
\begin{equation}
\begin{split}
\text{Hom}_A\left( \mathcal F(Y \otimes W), \mathcal F(Z \otimes W)\right) &\cong \text{Hom}_C\left( Y \otimes W, \mathcal F(Z \otimes W)\right) \\
&\cong   \text{Hom}_C\left( Y \otimes W, A \boxtimes_{V\otimes W}(Z \otimes W)\right)\\
&\cong   \text{Hom}_C\left( Y \otimes W, \bigoplus_X (X \boxtimes_{V}Z) \otimes \tau(\overline{X})\right)\\
&\cong \text{Hom}_C\left( Y \otimes W, Z \otimes W\right)\\
&\cong \text{Hom}_C\left( Y, Z\right).
\end{split}
\end{equation}
Here, the first isomorphism is a property of the induction and restriction functor, see Lemma 2.61 of \cite{CKM} (or \cite{KO, EGNO} for the statement in the non-super setting) and the fourth one uses that $W$ appears only once in $A$. More precisely one can embed 
\[
\text{Hom}_C\left( Y \otimes W, \bigoplus_X (X \boxtimes_{V}Z) \otimes \tau(\overline{X})\right)  \quad \text{in} \quad
\prod_X\text{Hom}_C\left( Y \otimes W, (X \boxtimes_{V}Z) \otimes \tau(\overline{X})\right)
\]
and $\text{Hom}_C\left( Y \otimes W, (X \boxtimes_{V}Z) \otimes \tau(\overline{X})\right) \cong \text{Hom}_{C_V}\left( Y, (X \boxtimes_{V}Z)\right) \otimes \text{Hom}_{C_W}\left( W,  \tau(\overline{X})\right)$ with $C_V$ and $C_W$ the vertex tensor categories of the \voas{} $V$ and $W$. The seond factor clearly vanishes unless $\tau(\overline{X})\cong W$ which happens by assumption if and only if $X=V$. 
\end{proof}
\subsubsection{The basic Abelian example}\label{sec:catabel}

Let us discuss an easy example. Consider a pair of two free bosons, i.e. the tensor product $\mathcal H \otimes \mathcal H$ of two Heisenberg \voas.
Let $C_\Psi$ be the full category of Fock modules of the first Heisenberg \voa{} whose weight is an integer multiple of $1/\sqrt{\Psi}$ and let $D_\Psi$ be the corresponding one for the second factor. Then 
\[
C_\Psi' =C_{\Psi^{-1}}.
\]
We now consider the four categories $C_{\Psi+1}, C_{\Psi+1}', D_{\Psi^{-1}+1}, D_{\Psi^{-1}+1}'$ and the following object in $C_{\Psi+1} \boxtimes D_{\Psi^{-1}+1}$.
\[
A = \bigoplus_{n\in\mathbb Z} F_{\frac{n}{\sqrt{\Psi+1}}} \otimes F_{\frac{n}{\sqrt{\Psi^{-1}+1}}}.
\]
It is half-integer graded and clearly isomorphic to the Heisenberg \voa{} times a $bc$-\voa, i.e. the vertex operator superalgebra of the lattice $\mathbb Z$.  As discussed above it defines a commutative and associative superalgebra object.
We have that every module that is local with  $C_{\Psi+1} \boxtimes D_{\Psi^{-1}+1}$ induces to a local module of $A$. Consider for example $F_{m\sqrt{\Psi+1}} \otimes F_0$ which induces to
\[
A \boxtimes_{\mathcal H\otimes \mathcal H} \left( F_{m\sqrt{\Psi+1}} \otimes F_0\right) \cong  \bigoplus_{n\in\mathbb Z} F_{m\sqrt{\Psi+1} +\frac{n}{\sqrt{\Psi+1}}} \otimes F_{\frac{n}{\sqrt{\Psi^{-1}+1}}}
\cong \widetilde{F}_{m\sqrt{\Psi}} \otimes bc
\] 
where the second is as $\mathcal H \otimes bc$-module and to avoid confusion we dentoe the Fock-module of the latter Heisenberg \voa{} by $\widetilde{F}$ and corresponding subcategories of modules by $E_\Psi$. The statement is seen from writing
\[
\frac{1}{\sqrt{\Psi+1}}\left(m(\Psi+1)+n, n\sqrt{\Psi}\right) = \frac{m+n}{\sqrt{\Psi+1}}\left(1, \sqrt{\Psi}\right) + \frac{m}{\sqrt{\Psi^{-1}+1}}\left(\sqrt{\Psi}, -1\right).
\]
Similar we have that
$F_0 \otimes F_{m\sqrt{\Psi^{-1}+1}}$ induces to
\[
A \boxtimes_{\mathcal H\otimes \mathcal H} \left( F_0 \otimes F_{m\sqrt{\Psi^{-1}+1}}\right) \cong  \bigoplus_{n\in\mathbb Z} F_{\frac{n}{\sqrt{\Psi+1}}} \otimes F_{m\sqrt{\Psi^{-1}+1}+\frac{n}{\sqrt{\Psi^{-1}+1}}}
\cong \widetilde{F}_{\frac{m}{\sqrt{\Psi}}} \otimes bc.
\] 
First of all, we can identify $C_{\Psi+1}'$  with $C_{\Psi+1}'\boxtimes \mathcal H$ ($\mathcal H$ viewed as the tensor identity). 
The induction functor is a braided tensor functor from $C_{\Psi+1}'\boxtimes \mathcal H$ to the subcategory $E_{\Psi}'\boxtimes bc$ of local $A$-modules. In this case the functor is even fully faithful by Proposition \ref{prop:ff} and thus furnishes an equivalence of braided tensor categories  $C_{\Psi+1}'\boxtimes \mathcal H$ and $E_{\Psi}'\boxtimes bc$.
Similarly we also have the braided equivalence of $\mathcal H \boxtimes D_{\Psi^{-1}+1}'$ and $E_{\Psi}\boxtimes bc$.

\subsubsection{Rational examples associated to Virasoro, $\mathfrak{osp}(1|2)$ and $\sltwo$}

Let $\text{Vir}(u, v)$ be the simple and rational Virasoro \voa{} at central charge $1-6\frac{(u-v)^2}{uv}$. The complete list of inequivalent simple modules of $\text{Vir}(u, v)$ are $M^{u, v}(n, m)$ with $1\leq n \leq u-1$, $1\leq m \leq v-1$ and the identification $M^{u, v}(n, m)\cong M^{u, v}(u-n, v-m)$. Fusion rules are well-known and we only need that $M^{u, v}(n, 1) \boxtimes_{\text{Vir}(u, v)}  M^{u, v}(1, m) \cong M^{u, v}(n, m)$ and that the $M^{u, v}(n, 1)$ form a closed subcategory which we call $C^{u, v}(\bullet, 1)$. Its centralizer is the full subcategory with simple objects the $M^{u, v}(1, m)$ with $m$ odd, which we denote by $C^{u, v}_{\text{odd}}(1, \bullet)$. 
Similarly, we denote by $C^{u, v}(1, \bullet)$ the full subcategory whose simple objects are the  $M^{u, v}(1, m)$ and whose centralizer $C^{u, v}_{\text{odd}}(\bullet, 1)$ has as simples the 
the $M^{u, v}(n, 1)$ with $n$ odd. $C^{u, v}_{\text{even}}(1, \bullet)$ and $C^{u, v}_{\text{even}}( \bullet, 1)$ are defined in the obvious way.

\begin{exam}{\bf A rational $\mathfrak{osp}(1|2)$ example }

There is also a rational $\mathfrak{osp}(1|2)$ example following Remark 3.8 of \cite{CFK}. Namely one has for positive integer $k$
\[
L_k\left( \mathfrak{osp}(1|2) \right) \cong  \bigoplus_{n=0}^k L_k(n) \otimes M^{k+2, 2k+3}(n+1, 1)
\]
as $L_k(\sltwo) \otimes \text{Vir}(k+2, 2k+3)$-module.
Then a main finding of \cite{CFK} is that local $L_k\left( \mathfrak{osp}(1|2) \right)$-modules are braided equivalent to the subcategory of $C^{u+2, 2u+3}_{\text{odd}}(1, \bullet)$ whose simple objects  are of type $M^{u+2, 2u+3}_{1, r}$ for odd $r$.

\end{exam}

The set-up for the next two examples is as follows.
Consider now $L_k(\sltwo)$ for $k+2=u \in \mathbb Z_{>2}$, i.e. the rational series of affine \voas{} of $\sltwo$. 
It is then well-known \cite{Iohara} that 
\begin{equation}\label{eq:GKOsl2}
\begin{split}
L_k(n-1 ) \otimes L_1(\sltwo) &\cong \bigoplus_{\substack{m=1\\ m +n\ \text{even} }}^{u} L_{k+1}(m-1) \otimes M^{u, u+1}(n, m)\\
L_k(n-1 ) \otimes L_1(1) &\cong \bigoplus_{\substack{m=1\\ m +n\ \text{odd} }}^{u} L_{k+1}(m-1) \otimes M^{u, u+1}(n, m)
\end{split}
\end{equation}
as $L_{k+1}(\sltwo) \otimes \text{Vir}(u, u+1)$-module.  
We now iterate \eqref{eq:GKOsl2} and using that the vertex operator superalgebra of four free fermions is just $F(4) \cong L_1(\sltwo) \otimes L_1(\sltwo) \oplus L_1(\omega_1) \otimes L_1(\omega_1)$ to get
\begin{equation}\nonumber
\begin{split}
L_k(n-1 ) \otimes F(4) &\cong \bigoplus_{m=1}^{u}\bigoplus_{m'=1}^{u+1} L_{k+2}((m'-1)) \otimes M^{u+1, u+2}(m, m')\otimes M^{u, u+1}(n, m)
\end{split}
\end{equation}
Especially, we get the \voa{} extension 
\[
A := \text{Com}\left(L_{k+2}(\sltwo), L_k(\sltwo) \otimes F(4)\right) \cong \bigoplus_{m=1}^{u} M^{u+1, u+2}(m, 1)\otimes M^{u, u+1}(1, m)
\]
of $\text{Vir}(u, u+1) \otimes \text{Vir}(u+1, u+2)$ and
\[
B:= L_k(\sltwo) \otimes F(4) \cong \bigoplus_{m=1}^{u}\bigoplus_{m'=1}^{u+1} L_{k+2}(m'-1) \otimes M^{u+1, u+2}(m, m')\otimes M^{u, u+1}(1, m)
\]
of $L_{k+2}(\sltwo) \otimes A$.
We will know discuss these two cases.

\begin{exam}{\bf A rational Virasoro example }

Let us denote by $\mathcal F_A$ the induction functor from $\text{Vir}(u, u+1) \otimes \text{Vir}(u+1, u+2)$-mod to $A$-mod. 
Then the vertex operator superalgebra $A$ is constracted out of $C^{u+1, u+2}(\bullet, 1) \otimes C^{u, u+1}(1, \bullet)$ which is centralized by  $C^{u+1, u+2}(1, \bullet)_{\text{odd}} \otimes C^{u, u+1}(\bullet, 1)_{\text{odd}}$ as well as $C^{u+1, u+2}(1, \bullet)_{\text{even}} \otimes C^{u, u+1}(\bullet, 1)_{\text{even}}$ so that these modules lift to local $A$-modules. 
One can analyze the lifting problem in exactly the same manner as in \cite{CFK}. We refrain from doing so and only list the results. 
Essentially by Proposition 4.4 of \cite{CKM} every simple module of $C^{u+1, u+2}(1, \bullet) \otimes C^{u, u+1}(\bullet, 1)$ lifts to a simple local or twisted $A$-module.
Namely,
\[
\mathcal F\left(M^{u+1, u+2}(1, m')\otimes M^{u, u+1}(m'', 1)\right) \cong \bigoplus_{m=1}^{u} M^{u+1, u+2}(m, m')\otimes M^{u, u+1}(m'', m)
\]
and the formula for conformal weights
\[
h_{m, m'}^{u+1, u+2} = \frac{\left( (u+1)m'-(u+2)m\right)^2-1}{4(u+1)(u+2)}, \qquad
h_{m'', m}^{u, u+1} = \frac{\left( (u+1)m''-um\right)^2-1}{4u(u+1)}
\]
tells us that these induced modules are integer graded if and only if $m'+m''$ is odd and otherwise they are half-integer graded. In other words, the simple local modules are exactly those for which $m'+m''$ is even. 
 Using Corollary 3.32 of \cite{DMNO} one can indeed compute that these are all local and Ramond twisted $A$-modules. In summary the category of local $A$ modules is braided equivalent (as braided supercategories) to the full subcategory of $C^{u+1, u+2}(1, \bullet) \otimes C^{u, u+1}(\bullet, 1)$ whose simples are the $M^{u+1, u+2}(1, m')\otimes M^{u, u+1}(m'', 1)$ with $m'+m''$ even. This is exactly the picture advocated in Figure \ref{fig:two}.

\end{exam}

\begin{exam}{\bf A rational $\sltwo$ example }

We consider now the extension from $L_{k+2}(\sltwo) \otimes A$ to $B$. We see that the local $A$-modules that are lifts of $C^{u+1, u+2}(1, \bullet)$ are used in the construction of $B$ and so the centralizer is (the lift of) $C^{u, u+1}_{\text{odd}}(\bullet, 1)$. Using Proposition 4.4 of \cite{CKM} one again sees with similar computations as in \cite{CFK} that inducing simple modules gives simple modules and using Corollary 3.32 of \cite{DMNO} one again gets that these exhaust all simple local and Ramond twisted $B$-modules. So here we see that the supercategory of local $B$-modules is braided equivalent to the supercategory $C^{u, u+1}(\bullet, 1)$.
These two examples can be viewed as a fairly simple instance of illustrating the expected behaviour of module categories for two consecutive junctions of \voas.

\end{exam}

\subsection{The $\Psi\to \infty$ limit of $\mathfrak{A}^{(n)}[SU(2),\Psi]$}

Let $n$ be a positive integer. Let $k_1, k_2$ be complex numbers satisfying
\[
\frac{1}{k_1+2} + \frac{1}{k_2+2} = n, 
\]
let $\Psi=k_1+2$,
then 
\[
\mathfrak{A}^{(n)}[SU(2), \Psi] :=  \bigoplus_{\substack{ m=0\\  m \ \text{even}}}^\infty  L_{k_1}(m) \otimes  L_{k_2}(m).
 \]
In the limit $\Psi$ to $\infty$ we have that $k_2+2=\frac{1}{n}$ and 
\[
\lim_{\Psi\rightarrow \infty} \mathfrak{A}^{(n)}[SU(2). \Psi] = Z \otimes \bigoplus_{\substack{ m=0\\  m \ \text{even}}}^\infty  \rho(m) \otimes  L_{k_2}(m)
\]
with $Z$ the abelian \voa{} of three free bosons at level zero. This limit \voa{} of $\mathfrak{A}^{(n)}[SU(2). \Psi]$ has been constructed in a very different way in \cite{C} and has been denotet $\mathcal Y_n$ there. The construction is via a series of standard \voa{} operations. The starting point is the affine \voa{}  $L_k(\mathfrak{sl}(n+1))$ at level $k=-(p^2-1)/p$. One then considers the embedding of $\sltwo$ in $\mathfrak{sl}(n+1)$, such that 
\[
\mathfrak{sl}(n+1) \cong4\rho(0) \oplus 4\rho(n\omega_1) \oplus \rho(2\omega_1) \oplus \rho(3\omega_1) \dots \oplus \rho((n-1)\omega_1).
\]
The DS-reduction for this embedding of $L_k(\mathfrak{sl}(n+1))$ was then denoted by $\mathcal W_n$ in \cite{C}. 
It turns out that $L_{k'}(\mathfrak{gl}(2))$ at level $k'=-2+\frac{1}{n}$ embeds conformally in $\mathcal W_n$ and the Heisenberg coset is
\[
\text{Com}\left(\mathcal H, \mathcal W_n\right) \cong \bigoplus_{\substack{m=0\\ m \ \text{even}}} L_{k'}(m)
\]
as $L_{k'}(\mathfrak{sl}(2))$ module. This coset then allows for an infinite order simple current extension $\mathcal Y_n$ due to Theorem 4.1 of \cite{CKLR} that is of the desired form
\[
\mathcal Y_n \cong  \bigoplus_{\substack{m=0\\ m \ \text{even}}} (2m+1) L_{k'}(m).
\]
We mention that the DS-reduction of $\mathcal Y_n$ coincides with the best-known family of logarithmic \voas, the triplet algebras $W(n)$ (see e.g. \cite{AM}), at least at the level of Virasoro modules \cite{C}.

It would be nice to construct $\mathfrak{A}^{(n)}[SU(2), \Psi]$. So let us consider $D(2, 1;-\lambda)$ at level $1/2$. If we perform the DS-reduction for the principal embedding into the level $1/2$ affine $\sltwo$ then the reduced $W$-algebra is the large $N=4$ superconformal algebra at central charge $c=-6$ with affine vertex operator subalgebra $L_{k_1}(\sltwo) \otimes L_{k_2}(\sltwo)$ at levels $k_1=\frac{1}{2}\left(\lambda-\frac{3}{2}\right)$ and $k_2=\frac{1}{2}\left(\lambda^{-1}-\frac{3}{2}\right)$. We see that
\[
\frac{1}{k_1+2} + \frac{1}{k_2+2} = 2
\]
and that the central charges of $L_{k_1}(\sltwo) \otimes L_{k_2}(\sltwo)$ is also $c=-6$. We thus have a conformal embedding \cite{AKMPP}. We conjecture that this \voa{} decomposes as
\[
H_{DS}\left( D(2, 1;-\lambda)_{1/2}\right) \cong \bigoplus_{ m=0}^\infty  L_{k_1}(m) \otimes  L_{k_2}(m)
\]
so that its even vertex operator subalgebra is $\mathfrak{A}^{(2)}[SU(2), \lambda]$.
We note, that the large $\Psi$-limit seems to be a large center times Adamovi\'c's small $N=4$ super conformal algebra at central charge $c=-9$ \cite{AdaN4}.

\subsection{The general $\Psi \to \infty$ limit}

Let now $G$ be a simple compact Lie group with Lie algebra $\g$. 
Let $P_+$ be the set of dominant weights and $Q$ the root lattice of $\g$. Let $P_+^\vee$ the set of dominant coweights and $Q^\vee$ the coweight lattice and $\lambda^\vee$ the coweight corresponding to the weight $\lambda$. Let $m$ be the lacity of $\g$ and let $\Psi=k_1+h^\vee$ such that
\[
\frac{1}{k_1+h^\vee} +\frac{1}{k_2+h^\vee} =mn
\]
We then claim that 
\[
\mathfrak{A}^{(n)}[G,\Psi] = \bigoplus_{\lambda \in P^+\cap\ Q} L_{k_1}(\lambda, \g)\otimes   L_{k_2}(\lambda^\vee, {}^L\g)
\]
is a simple \voa.
So $\mathfrak{A}[G, \infty]$ is the conformal extension of $L_{k_2}({}^L\g)$ for $k_2+h^\vee=\frac{1}{nm}$ by the product of Weyl modules and finite dimensional representations of the dual group (times a large center)
\begin{equation}
\lim_{\Psi\rightarrow \infty} \mathfrak{A}^{(n)}[G,\Psi] =Z \otimes  \bigoplus_{\lambda \in P^+\cap Q} \rho(\lambda)\otimes   L_{k_2}(\lambda^\vee, {}^L\g).
\end{equation}
We wonder if the construction of the limit for $G=SU(2)$ generalizes somehow. Let us restrict to $\ASUNn$. 
Let $H_{DS,m}$ be the Drinfield-Sokolov functor corresponding to the $\mathfrak{sl}_2$ embedding in $\mathfrak{sl}_{N+m}$ of type $m, 1, 1, \dots, 1$. 
Then $H_{DS,m}\left(V_k(\mathfrak{sl}_{N+m})\right)$ contains $V_{k+m-1}(\mathfrak{gl}_N)$ as a vertex operator subalgebra. 
Set 
\[
m:=(n+1)(N-1)-N \quad\text{and}\quad  k := \frac{n+1}{n}-(n+1)(N-1)
\]
Central charges indicate that this is a conformal embedding. Let $\mathcal H$ be the Heisenberg \voa{} commuting with the $V_{k+m-1}(\mathfrak{sl}_N)$. Then the coset \voa{} $\Com\left(\mathcal H, H_{DS,m}\left(V_{k}(\mathfrak{sl}_{m+N}) \right)\right)$ is an extension of 
$V_{-N+\frac{1}{n}}(\mathfrak{sl}_{N})$.
It is tempting to believe that this \voa{} can be extended to $\lim_{\Psi\rightarrow \infty}\ASUNn$ as in the $SU(2)$ case. 

The triplet \voas{} have a natural higher rank generalization to all simply laced Lie algebras, see \cite{FT, CM}. The results there indicate that these higher rank generalizations are the regular DS-reductions of our conjectural $\lim_{\Psi\rightarrow \infty} \mathfrak{A}^{(n)}[G,\Psi]$ \voas.

\section{$D(2, 1; -\lambda)_1$ and $\mathfrak{A}[SU(2),\Psi]$}  \label{sec:nine}

In this section we look in further detail to the example of $SU(2)$, involving the vertex operator superalgebra $D(2, 1; -\lambda)_1$.
For this we have to introduce the Lie superalgebra $\mathfrak d(2, 1; \alpha)$. We use Appendix B of \cite{FS}.
The even subalgebra is the direct sum of three copies of $\sltwo$. Fix a basis $e_i, f_i, h_i$ for $i=1,2,3$ with commutation relations
\[
[e_i, f_j] = 2\delta_{i, j}h_i, \qquad [h_i, e_j] = \delta_{i, j}e_i, \qquad [h_i, f_j]=-\delta_{i, j}f_i.
\]
The odd part is the tensor product of the standard representations of the three $\sltwo$'s. A basis is given by $\psi(\beta, \gamma, \delta)$ with $\beta, \gamma, \delta = \pm$. 
The non-vanishing commutation relations with the even subalgebra are
\begin{equation}\nonumber
\begin{split}
[e_1, \psi(-, \beta, \gamma)] &= -\psi(+, \beta, \gamma), \qquad [h_1, \psi(\pm, \beta, \gamma)] = \pm \frac{1}{2}\psi(\pm, \beta, \gamma), \\
[f_1, \psi(+, \beta, \gamma)] &= -\psi(-, \beta, \gamma),\qquad
[e_2, \psi( \beta, -, \gamma)] = -\psi( \beta, +,  \gamma), \\
[h_2, \psi( \beta, \pm, \gamma)] &= \pm \frac{1}{2}\psi(\beta, \pm, \gamma), \qquad
[f_2, \psi(\beta, + \gamma)] = -\psi( \beta, -, \gamma)\\
[e_3, \psi( \beta, \gamma, -)] &= -\psi( \beta, \gamma, +), \qquad [h_3, \psi(\beta, \gamma, \pm)] = \pm \frac{1}{2}\psi( \beta, \gamma, \pm), \\
[f_3, \psi( \beta, \gamma, +)] &= -\psi( \beta, \gamma, -).
\end{split}
\end{equation}
In order to write the commutation relations of the odd fields in a compact form \cite{FS} introduce the notation
\[
O_{+ +}^i = -2e_i, \qquad O_{--}^i= 2f_i, \qquad O_{+-}^i = O_{-+}^i=-2h_i, \qquad \epsilon_{+-}=-\epsilon_{-+}=1.
\]
Then 
\begin{equation}\nonumber
\begin{split}
[\psi(\beta_1, \gamma_1, \delta_1), \psi(\beta_2, \gamma_2, \delta_2)] &= \alpha_1O^1_{\beta_1 \beta_2}\epsilon_{\gamma_1\gamma_2}\epsilon_{\delta_1\delta_2} + \alpha_2O^2_{\gamma_1 \gamma_2}\epsilon_{\beta_1\beta_2}\epsilon_{\delta_1\delta_2} + \alpha_3O^3_{\delta_1 \delta_2}\epsilon_{\beta_1\beta_2}\epsilon_{\gamma_1\gamma_2}
\end{split}
\end{equation}
with $\alpha_1+\alpha_2+\alpha_3=0$ and scaling all $\alpha_i$ by the same scalar just amounts to a rescaling of the odd elements of the Lie superalgebra. We assume all three $\alpha_i\neq0$ and set $\alpha_3=1$ and the other two are then determined by
\[
1+\alpha_1+\alpha_2=0, \qquad \alpha=-1-\frac{1}{\alpha_2}.
\]
The non-degenerate invariant supersymmetric bilinear form is 
\begin{equation}\nonumber
\begin{split}
(e_i, f_j) &= \alpha_i^{-1} \delta_{i, j}, \qquad (h_i, h_j)= (2\alpha_i)^{-1} \delta_{i, j}, \\ (\psi(\beta_1, \gamma_1, \delta_1),
 \psi(\beta_2, \gamma_2, \delta_2))&= -2 \epsilon_{\beta_1\beta_2} \epsilon_{\gamma_1\gamma_2} \epsilon_{\delta_1\delta_2}. 
\end{split}
\end{equation}
Define 
\begin{equation}\label{eq:oddosp}
x_1:= \psi(+, +, -) - \psi(+, -, +) \qquad\text{and} \qquad y_1:= \psi(-, +, -) - \psi(-, -, +)
\end{equation}
so that $e_1, f_1, h_1, x_1, y_1$ generate the Lie superalgebra $\mathfrak{osp}(1|2)$ as subalgebra. 

We now set $\lambda=-\alpha$ and consider the affine vertex operator superalgebra of $\mathfrak d(2, 1; -\lambda)$ at level $k$. We denote it by $D(2, 1; -\lambda)_k$. Its affine vertex operator subalgebra is 
\[
L_{k(-1+\lambda^{-1})}(\sltwo) \otimes L_{k(-1+\lambda)}(\sltwo)  \otimes L_{k}(\sltwo)
\]
and we now set $k=1, k_1=-1+\lambda^{-1}$ and $k_2=-1+\lambda$
 so that 
\[
\frac{1}{k_1+2} + \frac{1}{k_2+2} = \frac{1}{1+\lambda^{-1}} +\frac{1}{1+\lambda} =1.
\]
Note, that the central charge of $L_{k_1}(\sltwo) \otimes L_{k_2}(\sltwo)  \otimes L_{1}(\sltwo)$ is one, i.e. it coincides with the central charge of $D(2, 1; -\lambda)_1$ and this is a conformal embedding \cite{KMP, AP}. The main out come of section \ref{sec:dec} can be summarized as
\begin{thm}\label{thm:dec}
For generic $\lambda$: As $L_1(\sltwo) \otimes L_{k-1}(\sltwo)\otimes L_{k_2}(\sltwo)$-modules
\[
D(2, 1; -\lambda)_1 \cong \left(L_1(\sltwo) \otimes \bigoplus_{\substack{ m=0\\  m \ \text{even}}}^\infty  L_{k_1}(m) \otimes  L_{k_2}(m)\right)  \oplus \left(L(\omega_1) \otimes \bigoplus_{\substack{ m=0\\  m \ \text{odd}}}^\infty L_{k_1}(m) \otimes L_{k_2}(m)\right)
\]
with $L_k(m)$ the Weyl module of the $m+1$-dimensional highest-weight module of $\sltwo$ at level $k$. 
\end{thm}

\subsection{Quantum Hamiltonian reduction}\label{sec:reduction}

Let $V_{\ZZ}$ be the lattice \voa{} of the integer lattice. It is strongly generated by two odd fields $b(z), c(w)$ with operator products
\[
b(z)c(w) = (z-w)^{-1}.
\]
This vertex operator superalgebra is also often called the $bc$-ghost \voa. 
Consider the \voa{} $D(2, 1; -\lambda)_1 \otimes V_{\ZZ}$ and define
\begin{equation}\label{eq:d}
d(z) = :b(z)e_1(z): +b(z).
\end{equation}
We denote by $H_\lambda$ the cohomology of the zero-mode of $d$ on $D(2, 1; -\lambda)_1 \otimes V_{\ZZ}$. In other words $H_\lambda$ is the DS-reduction of  $D(2, 1; -\lambda)_1$ with respect to the first of the three affine $\sltwo$ vertex operator supalgebras. We will now prove
\begin{thm}\label{thm:DS}
Let $\lambda$ be generic, then as \voas{}  $H_\lambda \cong V_{-2+\lambda} \otimes F(4)$ where $F(4)$ is the vertex operator superalgebra of four free fermions. 
\end{thm}
We start with the weaker statement
\begin{lem}\label{lem:DS}
Let $\lambda$ be generic, then $H_\lambda \cong L_{-2+\lambda}(\sltwo) \otimes F(4)$
 as $L_1(\sltwo)\otimes L_{k}(\sltwo)\otimes \Vir(c_k)$-modules.
 \end{lem}
\begin{proof}
Let $c = 13- 6(t+t^{-1})$ with $t=k_1+2$. Denote by $M(m, c)$ the Virasoro module obtained via DS-reduction from $ L_{k_1}(m)$ so that from Theorem \ref{thm:dec} we have that 
for generic $\lambda$ and as $L_1(\sltwo) \otimes L_{k_2}(\sltwo)\otimes \Vir(c)$-modules
\[
H_\lambda \cong \left(L_1(\sltwo) \otimes \bigoplus_{\substack{ m=0\\  m \ \text{even}}}^\infty  L_{k_2}(m) \otimes M(m, c) \right)  \oplus \left(L(\omega_1) \otimes \bigoplus_{\substack{ m=0\\  m \ \text{odd}}}^\infty L_{k_2}(m) \otimes M(m, c)\right)
\]
By Theorem 9.1.4 of \cite{Ara07} that $H_{DS}(L_k(\lambda))\cong L(\gamma_{\lambda -\Psi \rho^\vee})$ for generic $k$. 
so that by \cite{GKO} (see section \ref{sec:W})
for generic $k$ and as $L_{k}(\sltwo)\otimes \Vir(c_k)$-modules
\[
L_1(\sltwo)\otimes L_{k-1}(\sltwo) \cong \bigoplus_{\substack{ m=0\\  m \ \text{even}}}^\infty  L_{k}(m) \otimes M(m, c_k)   
\]
and
\[
L(\omega_1)\otimes L_{k-1}(\sltwo) \cong \bigoplus_{\substack{ m=0\\  m \ \text{odd}}}^\infty  L_{k}(m) \otimes M(m, c_k) 
\]
so that we see that the Lemma is indeed true. 
  \end{proof}
We are almost done with the proof of Theorem \ref{thm:DS}.
The \voa{} $L_{-2+\lambda}(\sltwo) \otimes F(4)$ is strongly generated by the four free fermions and the $L_{-1+\lambda}(\sltwo)$. The four free fermions in $H_\lambda$ are found as follows. 
 The four fermionic currents corresponding to $\psi(+, \beta, \gamma)$ are clearly not in the image but in the kernel of $d_0$. They have conformal dimension $1/2$ and the \OPE{} is deduced from the commutation relations to be 
\begin{equation}\label{eq:OPEfermions}
\psi(+,\beta_1, \gamma_1)(z)\psi(+, \beta_2, \gamma_2)(w) \sim -\frac{2\alpha_1e_1(w)\epsilon_{\beta_1, \beta_2}\epsilon_{\gamma_1, \gamma_2}}{(z-w)}
\end{equation}
but $e_1$ is in the same class as the vacuum since $[d_0, c] =1+e_1$. It is a similar but simpler computation as the proof of \cite[Lemma 8.2]{ACL2} to show that the Jacobi identity implies that the operator product algebra of the $\psi(+,\beta, \gamma)(z)$ and the $L_{-1+\lambda}(\sltwo)$ is fixed by the following three: above \OPE, the \OPE{} of $L_{-1+\lambda}(\sltwo)$ and that the $\psi(\beta, \gamma, +)(z)$ carry the standard representation of $L_{-1+\lambda}(\sltwo)$. 
It follows that the operator product algebras of the $\psi(\beta, \gamma, +)$ together with the $L_{-1+\lambda}(\sltwo)$ is the same as the one of the corresponding generators of $L_{-2+\lambda}(\sltwo) \otimes F(4)$.

We thus have two \voas{} and identified vertex operator subalgebras that have the same operator product algebra. In the case of $L_{-2+\lambda}(\sltwo) \otimes F(4)$ this is already the full \voa, which for generic $\lambda$ is in addition simple. This means that for generic $\lambda$ the corresponding fields of $H_\lambda$ generate the same \voa. We thus have that $L_{-2+\lambda}(\sltwo) \otimes F(4)$ is isomorphic to a vertex operator subalgebra of $H_\lambda$. But Lemma \ref{lem:DS} especially says that both \voas{} have the same graded dimension, so they must coincide. 

\begin{rem}
One can also view $H_\lambda$ and  $L_{-2+\lambda}(\sltwo) \otimes F(4)$ as \voas{} over the field of rational functions in $\lambda$. 
They are isomorphic as \voas{} over this field. 
\end{rem}

\subsubsection{DS-reduction of modules}\label{sec:DSspectralflow}

We now investigate the DS-reduction of spectrally flown modules.
We use the notation for spectral flow of \cite{CR}. 

Fix a basis $\{ e_n, f_n, h_n, K, d | n \in\ZZ\}$ of $\widehat\sltwo$ with $d$ the derivation, $K$ central and the other commuation relations 
\begin{equation}\nonumber
\begin{split}
[h_m, e_n] &=2e_{m+n}, \qquad[h_m, f_n]=-2f_{m+n}, \qquad [h_m, h_n]= 2m\delta_{m+n, 0}K\\
[e_m, f_n]&= -h_{m+n} -m\delta_{n+m,0}K, \qquad [e_m, e_n]=[f_m, f_n]=0.
\end{split}
\end{equation}
Spectral flows are then the following family of automorphisms of $\widehat\sltwo$
\begin{equation}\nonumber
\sigma^\ell(e_n)= e_{n-\ell}, \qquad \sigma^\ell(f_n)= f_{n+\ell}, \qquad  \sigma^\ell(h_n)= h_{n}-\delta_{n, 0}\ell K,
\end{equation}
for integer $\ell$. 
Consider now $V_k(\sltwo)$ so that $K$ acts by multiplication with the complex number $k$ on modules. Then spectral flow acts on the Virasoro zero-mode as
\[
\sigma^\ell(L_0) = L_0 - \frac{1}{2} \ell h_0 +\frac{1}{4} \ell^2 k.
\]
Consider now a module $M$ of $V_k(\sltwo)$. The module $\sigma^\ell(M)$ twisted by $\ell$ units of spectral flow is as a vector space isomorphic to $M$, but the action of $X$ in $\widehat\sltwo$ changes as follows. Denote the isomorphism from $M$ to $\sigma^\ell(M)$ by $\sigma^*_\ell$, then
\[
X \sigma^*_\ell(v) := \sigma^*_\ell\left( \sigma^{-\ell}\left(X\right)v\right).
\]
The character of a module is
\[
\ch[M](z, q) =\text{tr}_M(z^{h_0}q^{L_0-\frac{c}{24}}), \qquad c=\frac{3k}{k+2}, 
\]
so that
\[
\ch[\sigma^\ell(M)](z, q) = z^{\ell k} q^{\frac{\ell^2k}{4}}\ch[M]\left(zq^{\frac{\ell}{2}}, q\right)
\]
It follows then that the characters of Weyl modules satisfy
\[
\ch[\sigma^{-\ell}(L_k(m-1)] = q^{-\frac{c}{24}}\frac{q^{\frac{\ell^2t}{4}}z^{-\ell t}q^{\frac{m^2-1}{4t}}(z^mq^{-\frac{m\ell}{2}}-z^{-m}q^{\frac{m\ell}{2}})}{z^{1-2\ell}q^{\frac{\ell(\ell-1)}{2}} \prod\limits_{n=1}^\infty (1-z^2q^{n-\ell})(1-q^n)(1-z^{-2}q^{n+\ell-1})}
\]
where this is as formal power series, i.e. 
\[
\frac{1}{1-x} = 1+x+x^2 +\dots
\]
and $t=k+2$. This means that the character converges in the domain $$|q^\ell|<|z^2| < |q^{\ell-1}|.$$
Note, that the character simplifies using that the denominator is a Jacobi form
\[
\ch[\sigma^{-\ell}(L_k(m-1)] = q^{-\frac{c}{24}}\frac{q^{\frac{\ell^2t}{4}}z^{-\ell t}q^{\frac{m^2-1}{4t}}(z^mq^{-\frac{m\ell}{2}}-z^{-m}q^{\frac{m\ell}{2}})}{z \prod\limits_{n=1}^\infty (1-z^2q^{n})(1-q^n)(1-z^{-2}q^{n-1})}
\]
The supercharacter of the ghosts is just
\[
q^{\frac{1}{12}}\prod_{n=1}^\infty (1-z^2q^n)(1-z^{-2}q^{n-1}).
\]
The Euler-Poincar\'e character is then given by the limit $z$ to $q^{-1/2}$ of the character times the ghost supercharacter times $q^{k/4}$. It is then an easy computation that tells us this equals to the character of the Virasoro module $M^{k-2}(n, \ell+1)$ up to a possible sign,
\begin{equation}\nonumber
\begin{split}
\text{EP}(\ch[\sigma^{-\ell}(L_k(n-1))] ) &= \lim_{z\rightarrow q^{-\frac{1}{2}}} \ch[\sigma^{-\ell}(L_k(n-1))](z, q)q^{\frac{1}{12}+\frac{k}{4}}\prod_{n=1}^\infty (1-z^2q^n)(1-z^{-2}q^{n-1})  \\
&= (-1)^\ell \ch[M^{k-2}(n, \ell+1)].
\end{split}
\end{equation}
In other words we get 
\[
\text{EP}\left(\ch[\sigma^{-\ell}(D(2, 1;-\lambda)_1)]\right) =  (-1)^\ell \ch[L_{k_2-1}(\ell) \otimes F(4)]
\]
as desired. 
It thus remains to prove a vanishing theorem for cohomology of spectrally flown modules. Here is a possible way:
Recall that the DS-differential is
\[
d= d_{st} + \chi
\]
the sum of the standard differential for the semi-infinite cohomology of the affine Lie algebra and a character. Spectral flow twists $d_{st}$ and leaves the character $\chi$ invariant. We can however also apply opposite spectral flow to the ghosts of the reduction so that $d_{st}$ stays invariant but $\chi$ changes. Spectral flow by $-\ell$ on ghosts is just an isomorphism of ghost \voa{} but it changes the ghost grading by $\ell$. One can now try to apply the standard steps as in e.g. \cite{Aranoghost} and hopefully get a vanishing theorem, i.e. all homologies vanish, except in degree $\ell$. We hope to be able work out the arguments for this soon. 

\subsection{Conformal embeddings and decomposition}\label{sec:dec}

We will work out the decopositions of $D(2, 1; -\lambda)_1$ and $\psl_1$ into modules of its even affine vertex operator subalgebra.

\subsubsection{Constructing the relevant \voas}\label{sec:constVOAs}

First of all let us give a concrete basis of the relevant finite dimensional Lie super algebras. 
We start with $\mathfrak{sl}(2|1)$. Its even sub algebra is spanned by $e, f, h, x$ with non-vanishing relations
\[
[h, e]=e, \qquad [h, f]=-f, \qquad [e, f]=2h
\]
the odd part is spanned by $a^\pm, b^\pm$ with action of the even part given by
\begin{equation}\nonumber
\begin{split}
[h, a^\pm] &= \pm \frac{1}{2} a^\pm, \qquad [h, b^\pm] = \pm \frac{1}{2} b^\pm, \qquad 
[x, a^\pm] =  \frac{1}{2} a^\pm, \qquad [x, b^\pm] = - \frac{1}{2} b^\pm, \\
[e, a^-] &= -a^+, \qquad\quad [e, b^-]=b^+, \qquad\quad\ \, [f, a^+]=-a^-, \qquad [f, b^+]=b^-
\end{split}
\end{equation}
and with non-vanishing anti-commutation relations
\[
[a^+, b^+] = e, \qquad  [a^-, b^-] = f, \qquad [a^\pm, b^\mp] = x\mp h.
\]
The invariant bilinear form satisfies
\[
(h, h)=\frac{1}{2}, \qquad (e, f)=1, \qquad (x, x)=-\frac{1}{2}, \qquad  (a^+, b^-)=1=(a^-, b^+).
\]
Consider the universal affine \voa{} $V_k\left(\mathfrak{sl}(2|1)\right)$ of $\mathfrak{sl}(2|1)$. It is strongly-generated by $\{ y(z) | y = e, f, h, x, a^\pm, b^\pm \}$ with operator products
\[
y(z)v(w) \sim \frac{k(y, v)}{(z-w)^2} +\frac{[y,v](w)}{(z-w)}, \qquad y, v \in \{e, f, h, x, a^\pm, b^\pm \}
\] 
as usual. 
Consider the lattice
\[
L = A_1 \oplus \sqrt{-1}A_1 \cup \left( A_1 + \omega \oplus \sqrt{-1}\left(A_1+\omega\right)\right) 
\]
with $\omega$ the fundamental weight of $A_1$. 
We claim that 
\[
D_k = \text{Com}\left( H, W_k \right), \qquad W_k = V_k\left(\mathfrak{sl}(2|1)\right) \otimes V_L
\]
is $D(2, 1; -\lambda)_1$ for generic $k$. Here $H$ denotes the Heisenberg sub \voa{} generated by $x-y$, where $y$ is the Heisenberg field of the lattice \voa{} $V_{\sqrt{-1}A_1}$ and normalized such that $\sqrt{-1}\omega$ has eigenvalue $1/2$, i,e, $y$ has norm $-1/2$. 
\begin{lem}\label{lem:inclusions}
\begin{equation}
\begin{split}
D_{k_1} \supset D(2, 1; -\lambda)_1 \supset & \left(L_1(\sltwo) \otimes \bigoplus_{\substack{ m=0\\  m \ \text{even}}}^\infty  L_{k_1}(m) \otimes  L_{k_2}(m)\right)  \oplus \\
& \oplus \left(L_1(\omega) \otimes \bigoplus_{\substack{ m=0\\  m \ \text{odd}}}^\infty L_{k_1}(m) \otimes L_{k_2}(m)\right)
\end{split}
\end{equation}
\end{lem} 
\begin{proof}
Set $k=k_1$.
The argument goes as follows. First of all, denote the lattice \voa{} primary fields corresponding to $n\omega$ by $\phi_n$ and those corresponding to $n\sqrt{-1}\omega$ by $\varphi_n$. 
Then the fields
\[
a^\pm(z)\phi_{\pm 1}(z) \varphi_1(z)\qquad \text{and} \qquad b^\pm(z)\phi_{\pm 1}(z) \varphi_{-1}(z)
\]
are all elements of $D_k$. Also $:a^+a^-\varphi_2:$ and $:b^+b^-\varphi_{-2}:$ are in $D_k$. The conformal dimension of all these fields is one and it is a straight forward OPE computation to verify that the three affine $\sltwo$ sub \voas{} have the same levels as the one of 
$D(2, 1; -\lambda)_1$. Moreover by construction the eight fermionic fields carry the tensor product of the standard representations which fixes the sub \voa{} generated by these fields to be
 $D(2, 1; -\lambda)_1$. Alternatively, this statement is also straight-forward to check via OPE computations.

It thus follows that $D(2, 1; -\lambda)_1 \subset D_k$ for generic $k$. 
Define the fields
\[ 
X_n = : a^+ \partial a^+ \partial^2 a^+ \dots \partial^{n-1} a^+: \varphi_n \phi_{\beta_n}, \qquad \beta_n = \begin{cases} 0 & \quad \text{if} \ n \ \text{is even} \\
1 & \quad \text{if} \ n \ \text{is odd} \end{cases}
\]
then we clearly have that these fields are elements of $D_k$. Moreover they have weight $(n\omega, n\omega, \beta_n\omega)$ with $\beta_n=0$ for even $n$ and $\beta_n=1$ for odd $n$.  This is the weight with respect to the zero-mode of the three $\sltwo$'s, the last one being the one of level one. The conformal dimension is
\[
\Delta(X_n) = \frac{n(n+2)}{4} +\frac{\beta_n}{4}
\]
and it has regular OPE with all three $e_i(z)$. 
It follows that $X_n$ is a primary fields for the three affine $\sltwo$'s. We compute that 
\[
X_1^+X_n^+ \sim (z-w)^{\frac{(n+\beta_n)}{2}} X_{n+1}^+ +\dots
\]
so that all thes fields are contained in $D(2, 1; -\lambda)_1$ and hence 
\begin{equation}
\begin{split}
D(2, 1; -\lambda)_1 \supset& \left(L_1(\sltwo) \otimes \bigoplus_{\substack{ m=0\\  m \ \text{even}}}^\infty  L_{k_1}(m) \otimes  L_{k_2}(m)\right)  \oplus \\
& \oplus \left(L_1(\omega) \otimes \bigoplus_{\substack{ m=0\\  m \ \text{odd}}}^\infty L_{k_1}(m) \otimes L_{k_2}(m)\right)
\end{split}
\end{equation}
where we write $L_{k}(m)$ for $L_{k}(m\omega)$.
\end{proof}
\begin{lem}
\begin{equation}\nonumber
\begin{split}
\ch[D_k] &= \ch[L_1(\sltwo)] \sum  _{\substack{ m=0\\  m \ \text{even}}}^\infty  \ch[L_{k_1}(m)] \sum_{n\in \mathbb Z} x^nq^{\frac{(m+1)^2-1}{4}} \left(\chi(n+m+1)-\chi(n-m-1)\right) + \\
&\qquad  \ch[L_1(\omega)] \sum  _{\substack{ m=0\\  m \ \text{odd}}}^\infty  \ch[L_{k_1}(m)] \sum_{n\in \mathbb Z} x^nq^{\frac{(m+1)^2-1}{4}} \left(\chi(n+m+1)-\chi(n-m-1)\right) \\
= \ch&\left[ \left(L_1(\sltwo) \otimes \bigoplus_{\substack{ m=0\\  m \ \text{even}}}^\infty  L_{k_1}(m) \otimes  L_{k_2}(m)\right)  \oplus \left(L_1(1) \otimes \bigoplus_{\substack{ m=0\\  m \ \text{odd}}}^\infty L_{k_1}(m) \otimes L_{k_2}(m)\right) \right] 
\end{split}
\end{equation}
with 
\[
\chi(r) = \frac{q^{-r^2/8}}{\eta(q)^3} \sum_{s=0}^\infty (-1)^s q^{(2s+|r-1|+1)^2/8}
\]
\end{lem}
\begin{proof}
The first identity follows immediately from Example 4.4 of \cite{BCL} while for the second one one needs to decompose $\ch[L_{k}\left(\mathfrak{sl}(2|1) \right)]$ into characters of Weyl modules $L_{k_1}(m)$ of $L_{k}\left(\mathfrak{sl}(2) \right)$. The corresponding multiplicity is $B_{m\omega}/\eta$ of  Corollary 5.8 of \cite{CLsf}. The Fourier coefficient in $x^n$ of $D_{k_1}$ is then $q^{-n^2/4}B_{m\omega}/\eta$ and a careful analysis then shows that this coefficient agrees with $q^{\frac{(m+1)^2-1}{4}}( \chi(n+m+1)-\chi(n-m-1))$. 
\end{proof}
 \begin{cor}
\begin{equation}\nonumber
\begin{split}
D_k &\cong  D(2, 1; -\lambda)_1\\  &\cong  \left(L_1(\sltwo) \otimes \bigoplus_{\substack{ m=0\\  m \ \text{even}}}^\infty  L_{k_1}(m) \otimes  L_{k_2}(m)\right)  \oplus \left(L_1(1) \otimes \bigoplus_{\substack{ m=0\\  m \ \text{odd}}}^\infty L_{k_1}(m) \otimes L_{k_2}(m)\right) 
\end{split}
\end{equation}
as  $L_1(\sltwo)\otimes L_{k_1}(\sltwo) \otimes L_{k_2}(\sltwo)$-module.
\end{cor}
Another corollary is now the character formula of $D(2, 1; -\lambda)_1$. For this recall, that 
\[
\ch[L_{k}(m)](q, z) = \frac{\left(z^{m+1}-z^{-(m+1)}\right)q^{\frac{m(m+2)-12k}{k+2}+\frac{1}{8}} }{\Pi(z)}
\]
with Weyl denominator 
\[
\Pi(z) = q^{\frac{1}{8}}\left(z-z^{-1}\right) \prod_{n=1}^\infty (1-z^2q^n)(1-q^n)(1-z^{-2}q^n).  
\]
The root lattice of $A_1$ is $\sqrt{2}\ZZ$ and its discriminant is isomorpic to $\ZZ/2\ZZ$ with non-trivial coset representative $\sqrt{2}\ZZ+\frac{1}{\sqrt{2}}:= A_1+(1)$. The Jacobi theta functions are
\[
\theta_{A_1}(q, z) = \sum_{\substack{m\in\ZZ\\ m \ \text{even}}} q^{\frac{m^2}{4}}z^m, \qquad
\theta_{A_1+(1)}(q, z) = \sum_{\substack{m\in\ZZ\\ m \ \text{odd}}} q^{\frac{m^2}{4}}z^m
\]
and the characters of $L_1(\sltwo)$-modules are
\[
\ch[L_1(\sltwo)] = \frac{\theta_{A_1}(q, z)}{\eta(q)}, \qquad \ch[L_1(1)] = \frac{\theta_{A_1+(1)}(q, z)}{\eta(q)} \]
with the usual Dedekind's eta function $\eta(q)$. 
We now compute
\begin{equation}\nonumber
\begin{split}
\ch[D&(2, 1; -\lambda)_1](z, w, v, q) = \ch[L_1(\sltwo)](v, q)\sum_{\substack{m=0 \\ m \ \text{even}}}  \ch[L_{k_1}(m)](q, z)\ch[L_{k_2}(m)](q, w)+ \\
&\qquad\qquad\qquad + \ch[L_1(1)](v, q)\sum_{\substack{m=0 \\ m \ \text{odd}}}  \ch[L_{k_1}(m)](q, z)\ch[L_{k_2}(m)](q, w)\\
&= \frac{\theta_{A_1}(q, v)}{\eta(q)} \sum\limits_{\substack{m=0\\ m \ \text{even}}}^\infty\frac{ q^{\frac{(m+1)^2}{4}}\left((zw)^{m+1}+(zw)^{-(m+1)} -(zw^{-1})^{m+1}-(z^{-1}w)^{m+1}\right) }{\Pi(z) \Pi(w)} + \\
&\ \  + \frac{\theta_{A_1+(1)}(q, v)}{\eta(q)} \sum\limits_{\substack{m=0\\ m \ \text{odd}}}^\infty\frac{ q^{\frac{(m+1)^2}{4}}\left((zw)^{m+1}+(zw)^{-(m+1)} -(zw^{-1})^{m+1}-(z^{-1}w)^{m+1}\right) }{\Pi(z) \Pi(w)}\\
&= \frac{\theta_{A_1}(q, v)\left(\theta_{A_1+(1)}(q, zw) -\theta_{A_1+(1)}(q, zw^{-1}) \right) +  \theta_{A_1+(1)}(q, v)\left(\theta_{A_1}(q, zw) -\theta_{A_1}(q, zw^{-1}) \right) }{\eta(q)\Pi(z)\Pi(w)}.
\end{split}
\end{equation}
The super character is obtained by changing the sign of the second summand. We summarize.
\begin{cor}
The character $\ch^+$ and super character $\ch^-$ of $D(2, 1; -\lambda)_1$ are meromorphic Jacobi forms,
\begin{equation}\nonumber
\begin{split}
\ch^\pm &:= \ch^\pm[D(2, 1; -\lambda)_1](z, w, v, q)\\
 &= \frac{\theta_{A_1}(q, v)\left(\theta_{A_1+(1)}(q, zw) -\theta_{A_1+(1)}(q, zw^{-1}) \right) \pm  \theta_{A_1+(1)}(q, v)\left(\theta_{A_1}(q, zw) -\theta_{A_1}(q, zw^{-1}) \right) }{\eta(q)\Pi(z)\Pi(w)},
\end{split}
\end{equation}
where $|z|, |w| < |q|^{\pm 1}$.
\end{cor}
There are a few remarks in order
\begin{rem}
The limit $z, w \rightarrow 1$ is by L'H\^opital's rule
\begin{equation}\nonumber
\begin{split}
\lim_{z, w \rightarrow 1 } \ch^\pm[D(2, 1; -\lambda)_1](z, w, v, q)&= 
\frac{1}{2} \frac{\theta_{A_1}(q, v)\theta''_{A_1+(1)}(q)\pm  \theta_{A_1+(1)}(q, v)\theta''_{A_1}(q)}{\eta(q)^7},
\end{split}
\end{equation}
with $\theta''(q):= \frac{d^2}{dz} \theta(q, z)\Big\vert_{z=1}$. This is a holomorphic Jacobi form. Recently the notion of quasi-lisse \voas{} has been established.  Ordinary modules of such \voas{} satisfy modular differential equations and the associated variety of such \voas{} is symplectic with finetly many symplectic leaves \cite{Arakawa:2016hkg}. 
It is thus an interesting question if  $D(2, 1; -\lambda)_1$ is a deformable family of quasi-lisse vertex operator superalgebras. 
\end{rem}
\begin{rem} 
Having the describtion of $D(2, 1; -\lambda)_1$ as 
\[
D(2, 1; -\lambda)_1 = \text{Com}\left( H, W_{k_1} \right), \qquad W_k = V_k\left(\mathfrak{sl}(2|1)\right) \otimes V_L
\]
it is easy to derive weight conditions on special modules. Recall that a module of $V_k\left(\mathfrak{sl}(2|1)\right)$ is atypical if the two weights satisfy $j=\pm b$, where $j$ is the $h_0$ eigenvalue and $b$ the $x_0$ eigenvalue of the highest-weight state. 
The Cartan element of the $V_{k_2}(\mathfrak{sl}_2)$ is 
\[
\frac{1}{k_2+1} \left(x+ky\right) =\lambda  \left(x+ky\right)
\]
it follows that atypical modules of $W_k$ satisfy the weight condition on the two $\sltwo$ weights in $D(2, 1; -\lambda)_1$ to be
\[
j = \pm \lambda  b.
\]
\end{rem}

\begin{rem}
In \cite{CL1, CL2} the notion of deformable families of \voas{} has been introduced. 
The philosophy is that such \voas{} allow for a large level (or large $\lambda$-limit) in which the \voa{} becomes often quite simple. For example 
\[
\lim_{k\rightarrow \infty} V_k\left(\mathfrak{sl}(2|1)\right) \cong H(4) \otimes SF(2)
\]
is just four Heisenberg \voas{} together with two pairs of symplectic fermions. Moreover a coset problem reduces to an orbifold problem in the large level limit, as e.g.
\[
\lim_{k\rightarrow \infty} D_k \cong   H(3) \otimes \left( SF(2) \otimes V_L\right)^{U(1)}.
\]
On the other hand, understanding that $D(2, 1; -\lambda)_1$ is a deformable family of \voas{}  implies \cite{CL3}
\[
\lim_{\lambda \rightarrow \infty}D(2, 1; -\lambda)_1 \cong H(3) \otimes L_1\left(\psl\right)
\]
with an abelian rank three Heisenberg \voa{} $H(3)$. 
Especially
\[
\psl \cong \left(L_1(\sltwo) \otimes \bigoplus_{\substack{ m=0\\  m \ \text{even}}}^\infty  \rho_m \otimes  L_{-1}(m)\right)  \oplus \left(L_1(1) \otimes \bigoplus_{\substack{ m=0\\  m \ \text{odd}}}^\infty \rho_m \otimes L_{-1}(m)\right)
\]
as  $L_1(\sltwo)\otimes \sltwo \otimes L_{-1}(\sltwo)$-module.
\end{rem}

\subsubsection{$L_k(\mathfrak{osp}(1|2))$ as a coset}

Recall that the odd elements $x_1, y_1$ defined in \eqref{eq:oddosp} of $\mathfrak{d}(2, 1; -\lambda)$ together with $e_1, f_1, h_1$ generate the subalgebra $\mathfrak{osp}(1|2)$ in $\mathfrak{d}(2, 1; -\lambda)$. We will now see that this extends to a conformal embeddings of the affine vertex superalgebra in $D(2, 1; -\lambda)_1$.

Using Theorem  \ref{thm:dec}  and section \ref{sec:W} we have 
\begin{equation}
\begin{split}
D(2, 1; -\lambda)_1 &\cong \left(L_1(\sltwo) \otimes \bigoplus_{\substack{ m=0\\  m \ \text{even}}}^\infty  L_{k_1}(m) \otimes  L_{k_2}(m)\right)  \oplus \left(L(\omega_1) \otimes \bigoplus_{\substack{ m=0\\  m \ \text{odd}}}^\infty L_{k_1}(m) \otimes L_{k_2}(m)\right)\\
&\cong \bigoplus_{\substack{ s, m=0\\  s \ \text{even}}}^\infty \left( L_{k_1}(m) \otimes  L_{k_2+1}(s) \otimes  M^{1+\lambda^{-1}}(s+1, m)\right)
\end{split}
\end{equation}
We thus see that 
\[
\text{Com}\left(L_{k_2+1}(\sltwo), D(2, 1; -\lambda)_1\right) \cong    \bigoplus_{m=0}^\infty \left( L_{k_1}(m) \otimes  M^{1+\lambda^{-1}}(1, m)\right)
\]
moreover this coset obviously contains $L_{k_1}(\mathfrak{osp}(1|2))$. Equality for generic level follows from equality of characters which is most easily seen by noting that the vacuum character of $L_{k_1}(\mathfrak{osp}(1|2))$ is the vacuum character of $L_{k_1}(\sltwo)$ times the one of symplectic fermions. But symplectic fermions are well-known \cite{AM} to decompose as Virasoro module as 
\[
SF(1) \cong   \bigoplus_{m=0}^\infty  (m+1)M^2(1, m).
\]
We thus have
\begin{cor}\label{cor:ospascoset}
The coset is $\text{Com}\left(L_{k_2+1}(\sltwo), D(2, 1; -\lambda)_1\right)=L_{k_1}(\mathfrak{osp}(1|2))$.
\end{cor}
Consider now the quantum Hamiltonian reduction on the $L_{k_1}(\mathfrak{osp}(1|2))$. The differential is computed from \cite{Kac:2003jh}
to be
\[
d_{osp} = :e_1(z)b(z): -\sqrt{\frac{\lambda^{-1}-1}{2}} :x_1(z)\beta(z): -\frac{1}{2}:\beta(z)\beta(z)c(z): +b(z) +:\beta(z)\phi(z):
\]
where we introduced the bosonic ghosts $\beta, \gamma$, the fermionic ghosts $b, c$ and the free fermion $\phi$ with non-vanishing OPEs
\[
\beta(z)\gamma(w) \sim b(z)c(w) \sim \phi(z)\phi(w) \sim \frac{1}{(z-w)}.
\]
Theorem 6.2 of \cite{Kac:2003jh} is a no-ghost Theorem for the reduction of universal Weyl and Verma modules so that for generic $\lambda$ the Euler-Poincar\'e character coincides with the true character. 
Let $M$ be a module of $L_{k_1}(\mathfrak{osp}(1|2))$ then the Euler-Poincar\'e character of $H_{d_{osp}}(M)$ is
\[
\text{EP}_{osp}(\ch[M]) = \lim_{z\rightarrow q^{-1/2}} \left( \ch[M](z, q) \text{sch}[\text{ghosts}](z, q)\right). 
 \]
 Let now $M=D(2, 1; -\lambda)_1$
Comparing with the Euler-Poincar\'e character of the reduction of section \ref{sec:reduction} we see that the two reductions only differ by the extra ghosts contribution of $\beta, \gamma, \phi$ so that
\begin{equation}\nonumber
\begin{split}
\text{EP}_{osp}(\ch[D(2, 1; -\lambda)_1]) &=\text{EP}(\ch[D(2, 1; -\lambda)_1]) \prod_{n=1}^\infty \frac{\left(1-q^{n+\frac{1}{2}}\right)}{\left(1-q^{n+\frac{1}{2}}\right)^2}\\
&= \ch[L_{k_2-1}(\sltwo) \otimes F(3)].
\end{split}
\end{equation}
It remains to show that $L_{k_2-1}(\sltwo) \otimes F(3)$ is indeed a subagebra of $H_{d_{osp}}(D(2, 1; -\lambda)_1)$. But this is the same type of argument as in section \ref{sec:reduction}. We still have the $L_{k_2+1}(\sltwo)$ in the cohomology of $d_osp$. Further the commutation relations of $\mathfrak{d}(2, 1; -\lambda)$ immediately imply that $\psi(+, +, +)(z), \psi(+, +, -)(z) + \psi(+, -, +)(z)$ and $\psi(+, -, -)(z)$ are in the kernel of the zero-mode of $d_{osp}$. They have conformal dimension $1/2$ and their OPE is given in \eqref{eq:OPEfermions}
but $e_1$ is in the same class as the identity since $[d_{osp, 0}, c]=1+e_1$. In other words, the cohomology contains three free fermions together with $L_{k_2+1}(\sltwo)$. The fermions carry the adjoint representation of $\sltwo$ under the zero-modes of $L_{k_2+1}(\sltwo)$. Hence we have that $L_{k_2-1}(\sltwo) \otimes F(3)$ is a sub vertex algebra of the cohomology and since characters coincide we have equality. We summarize 
\begin{thm}\label{thm:ospreduction}
For generic $\lambda$ we have $H_{d_{osp}}(D(2, 1; -\lambda)_1) \cong L_{k_2-1}(\sltwo) \otimes F(3)$.
\end{thm}
Recall that $F(n)$ is a simple current extension of $L_1(\mathfrak{so}(n))$.

\section*{Acknowledgements}
DG thanks S. Bravermann, K. Costello, P. Yoo for many instructive conversations. 
TC appreciates various discussions with T. Arakawa and A. Linshaw on related topics. 
\appendix 

\newcommand{\etalchar}[1]{$^{#1}$}

\end{document}